\documentclass{article}
\usepackage[utf8]{inputenc}

\title{Clustering Mixtures with Almost Optimal Separation in Polynomial Time}
\author{Jerry Li\thanks{\url{jerrl@microsoft.com}} \and Allen Liu\thanks{\url{cliu568@mit.edu}.  This work was partially done while working as an intern at Microsoft Research.  This work was supported in part by an NSF Graduate Research Fellowship and a Fannie and John Hertz Foundation Fellowship.}}
\date{}

\input{preamble.sty}

\begin{document}

\maketitle
\begin{abstract}
    We consider the problem of clustering mixtures of mean-separated Gaussians in high dimensions.
    We are given samples from a mixture of $k$ identity covariance Gaussians, so that the minimum pairwise distance between any two pairs of means is at least $\Delta$, for some parameter $\Delta > 0$, and the goal is to recover the ground truth clustering of these samples.
    It is folklore that separation $\Delta = \Theta (\sqrt{\log k})$ is both necessary and sufficient to recover a good clustering, at least information theoretically.
    However, the estimators which achieve this guarantee are inefficient.
    We give the first algorithm which runs in polynomial time, and which almost matches this guarantee.
    More precisely, we give an algorithm which takes polynomially many samples and time, and which can successfully recover a good clustering, so long as the separation is $\Delta = \Omega (\log^{1/2 + c} k)$, for any $c > 0$.
    Previously, polynomial time algorithms were only known for this problem when the separation was polynomial in $k$, and all algorithms which could tolerate $\poly \log k$ separation required quasipolynomial time.
    We also extend our result to mixtures of translations of a distribution which satisfies the Poincar\'{e} inequality, under additional mild assumptions.
    Our main technical tool, which we believe is of independent interest, is a novel way to implicitly represent and estimate high degree moments of a distribution, which allows us to extract important information about high-degree moments without ever writing down the full moment tensors explicitly.
\end{abstract}

\section{Introduction}

Gaussian mixture models are some of the most popular generative models in both theory and practice.
A $k$-Gaussian mixture model $\mcl{M}$ (henceforth $k$-GMM) is a distribution specified by $k$ non-negative mixing weights $w_1, \ldots, w_k$ which sum to $1$, and $k$ component Gaussians $N (\mu_1, \Sigma_1), \ldots,N (\mu_k,\Sigma_k)$, and is given by the probability density function
\[
\mcl{M} = \sum_{i = 1}^k w_i N (\mu_i, \Sigma_i) \; .
\]
In other words, to draw a sample from $\mcl{M}$, we select the $i$-th component with probability $w_i$, and then draw an independent sample from $N (\mu_i,\Sigma_i)$.
An important special case of this model is the \emph{isotropic} (also sometimes referred to as \emph{spherical}) case, where $\Sigma_i = I$ for all $i = 1, \ldots, k$.

One of the most important and well-studied problems for GMMs and isotropic GMMs in particular is \emph{clustering}.
Given $n$ independent samples $X_1, \ldots, X_n$ drawn from an unknown $k$-GMM, the goal of clustering is to recover a partition of the data points into $k$ parts $S_1, \ldots, S_k$ so that (a) almost all the points in every cluster are drawn from the same component Gaussian, and (b) for every component, there is some cluster which contains almost all of the samples drawn from it.
Of course, not every isotropic GMM is clusterable: for instance, if two components are identical, then it is information-theoretically impossible to detect which component a sample came from, and so recovering a good clustering is also impossible.
In the isotropic setting, a necessary and sufficient condition for the GMM to be clusterable is some amount of mean separation, namely, that $\norm{\mu_i - \mu_j}_2 \geq \Delta$ for all $i \neq j$, for some parameter $\Delta$.

The question then becomes: how small can we take $\Delta$ so that we can still cluster?
For simplicity, for the remainder of this section, let us assume that $w_i \geq 1 / \poly (k)$ for all $i = 1, \ldots, k$, that is, all the components have a nontrivial amount of weight.
Information-theoretically,  it is well-known (see e.g.~\cite{regev2017learning}) that $\Delta = \Theta (\sqrt{\log k })$ separation is necessary and sufficient to achieve a clustering which correctly clusters more than an $1 - \eps$ fraction of the points with high probability, for all $\eps > 0$. 
Unfortunately, all known efficient algorithms for clustering at this separation rely on brute force methods, and run in exponential time. 

If we restrict our attention to efficient estimators, the state of affairs is a bit more complicated.
In fact, for over 20 years, the smallest separation that polynomial time algorithms could tolerate was $\Delta = \Omega (k^{1/4})$~\cite{vempala2004spectral}.
It was not until recently that three concurrent papers~\cite{diakonikolas2018list,hopkins2018mixture,kothari2018robust} gave algorithms which could handle separation $\Delta = \Omega (k^{\gamma})$ with runtime and sample complexity $O(d^{\poly (1 / \gamma)})$.
In particular, whenever the separation is some fixed polynomial in $k$, these algorithms run in polynomial time.
Unfortunately, if we wish to achieve the optimal separation of $\Omega (\sqrt{\log k})$---or indeed, any polylogarithmic amount of separation---these algorithms would require quasipolynomial time and sample complexity.

\paragraph{The barrier at quasipolynomial time} These algorithms all get stuck at quasipolynomial time when $\Delta = \Theta (\poly \log k)$ for the same reason.
Fundamentally, all current algorithmic approaches to this problem rely on the following geometric identifiability fact:

\begin{quote}
    \emph{Given enough samples from an isotropic $k$-GMM with separation $\Delta = \Omega (k^\gamma)$, then any sufficiently large subset of samples whose empirical $t = \Theta(1 / \gamma)$ moments approximately match those of a Gaussian along all projections must essentially recover one of the true clusters.}
\end{quote}
Algorithmically, this amounts to finding a subset of points whose empirical $t$-th moment tensor is close to that of a Gaussian in the appropriate norm.
Since this results in a hard optimization problem whenever $t > 2$, these algorithms often solve some suitable relaxation of this using something like the Sum of Squares hierarchy.
But, as the separation decreases, the algorithms must match more and more moments.
In particular, to achieve $\Delta = \Theta (\poly \log k)$, one must set $t = \poly \log k$, that is, one must match polylogarithmically many moments.
However, even writing down the degree $t = \poly \log k$ moment tensor requires quasipolynomial time, and guaranteeing that the empirical moment tensor concentrates requires quasipolynomially many samples.
As a result, the aforementioned algorithms all require quasipolynomial time and sample complexity, as they need to not only write down the moment tensor, but perform some fairly complex optimization tasks on top of it.

On the flip side, there is no concrete reason for pessimism either.
While there are lower bounds against large classes of efficient algorithms for clustering mixtures of arbitrary Gaussians, see e.g.~\cite{diakonikolas2017statistical,brennan2020statistical}, none of these apply when the components are isotropic.
In particular, this leaves open the appealing possibility that one could even cluster down to separation $\Delta = \Theta (\sqrt{\log k})$ in polynomial time.
Stated another way, the question becomes:
\begin{quote}
\emph{Can we cluster any clusterable mixture of isotropic Gaussians in polynomial time?}
\end{quote}

\subsection{Our Results}
In this paper, we (almost) resolve this question in the affirmative.
Namely, for all constants $c > 0$, we give an algorithm which takes polynomially many samples and time, and which can cluster with high probability, so long as the separation satisfies $\Delta = \Omega (\log^{1/2 + c} k)$.
In other words, our algorithm can almost match the information theoretically optimal separation, up to sub-polylogarithmic factors.
Our main theorem is:
\begin{theorem}[informal, see Theorem \ref{thm:main-GMM}, Corollary \ref{coro:cluster-GMM}]
Let $c > 0$ be fixed but arbitrary.
Let $\mcl{M}$ be a mixture of $k$ isotropic Gaussians with minimum mixing weight lower bounded by $1 / \poly (k)$, and minimum mean separation at least $\Delta = \Omega (\log^{1/2 + c} k)$.
Then, there is an algorithm which, given samples $X_1, \ldots, X_n \sim \mcl{M}$ for $n = \poly (k, d)$, outputs a clustering which is correct for all of the points, with high probability.
Moreover, this algorithm runs in time $\poly(d, k)$.
\end{theorem}
%\begin{remark}
%When we say with high probability, we mean with failure probability smaller than any polynomial in $k,d$ i.e. the error probability of our clustering algorithm is smaller than any polynomial in $k,d$.
%\end{remark}
\noindent
We briefly remark that we can handle arbitrary mixing weights as well, but for simplicity we only state the theorem here in the more natural regime where all the mixing weights are not too small.
We also remark that a simple corollary of this is that we can also estimate the parameters of $\mcl{M}$ to good accuracy in polynomial time.
In fact, by using our algorithm as a warm start for the method proposed in~\cite{regev2017learning}, we can achieve arbitrarily good accuracy for parameter estimation, in polynomial time.
It is known that $\Delta = \Omega (\sqrt{\log k})$ is also necessary to achieve nontrivial parameter estimation with polynomially many samples~\cite{regev2017learning}, so our results for parameter estimation are also almost-optimal, in terms of the separation that they handle.
We note that our formal theorems are actually stated for parameter estimation, rather than clustering, however, in light of~\cite{regev2017learning}, these two problems are equivalent in the regime we consider.

Our main technique (as we will discuss in more detail later) extends to beyond Gaussians, and in fact also allows us to cluster any mixture of translations of a distribution $\mcl{D}$, so long as this distribution satisfies the \emph{Poincar\'{e} inequality}, under a mild technical condition.
Recall that a distribution $\mcl{D}$ over $\R^d$ is said to be $\sigma$-Poincar\'{e} if for all differentiable functions $f: \R^d \to \R$, we have
\[
\Var_{X \sim \mcl{D}} \left[ f(X) \right] \leq \sigma^2 \cdot \E_{X \sim \mcl{D}} \left[ \norm{\nabla f(X)}_2^2 \right] \; .
\]
This condition is widely studied in probability theory, and is satisfied by many natural distribution classes.
For instance, isotropic Gaussians are $1$-Poincar\'{e}, and any isotropic logconcave distribution is $\psi$-Poincar\'{e}, where $\psi$ is the value of the KLS constant (see e.g.~\cite{tat2018kannan} for an overview of the KLS conjecture).

In fact, the family of Poincar\'{e} distributions is the most general family of distributions for which the previously mentioned Sum of Squares-based methods for clustering are known to work.
For any well-separated mixture of $1$-Poincar\'{e} distributions,~\cite{kothari2018robust} shows that one can recover the same guarantees as mentioned above: if the minimum mean separation is $\Delta = \Omega (k^\gamma)$, then their algorithm successfully clusters the points in time $O(d^{\poly (1 / \gamma)})$.
As before, when the separation is polylogarithmic, the runtime and sample complexity of their method is once again quasipolynomial.

We show that one can improve the runtime and sample complexity to polynomial time, under two additional assumptions: first, the mixture must consist of translated versions of the same Poincar\'{e} distribution which we can get samples from, and second, the maximum and minimum separations between any two components must be polynomially related.
More concretely, we show:
\begin{theorem}[informal, see Theorem \ref{thm:main-Poincare}, Corollary \ref{coro:cluster-Poincare}]
Let $c > 0$ be fixed but arbitrary.
Let $\mcl{D}$ be a fixed distribution with mean zero over $\R^d$ which is $1$-Poincar\'{e}.
Let $\mcl{M}$ be a mixture of $k$ distributions where each component is of the form $\mu_i + \mcl{D}$.
Assume that the minimum mixing weight in this distribution is at least $1 / \poly (k)$, and moreover, assume that
\begin{align*}
&\min_{i \neq j} \norm{\mu_i - \mu_j}_2 \geq \Omega (\log^{1 + c} k)    \\
&\max_{i \neq j} \norm{\mu_i - \mu_j}_2 \leq \poly \left( \min_{i \neq j} \norm{\mu_i - \mu_j}_2 \right) \; .
\end{align*}
Then, there is an algorithm which, given samples $X_1, \ldots, X_n \sim \mcl{M}$ and samples $z_1, \dots , z_n \sim \mcl{D}$ for $n = \poly (k, d)$, outputs a clustering of $X_1, \dots , X_n$ which is correct for all of the points, with high probability.
Moreover, this algorithm runs in time $\poly(d, k)$.
\end{theorem}
\begin{remark}
Note that we only need the samples from $\mcl{D}$.  We do not need to actually know the distribution or access the p.d.f. 
\end{remark}

\noindent
We note that separation $\Omega (\log k)$ is optimal for general Poincar\'{e} distributions, as Poincar\'{e} distributions include some distributions with exponential tails, for which $\Omega (\log k)$ separation is necessary to cluster.
Thus, the separation that we require is almost optimal for general Poincar\'{e} distributions.
We also note that one immediate consequence of this theorem, alongside Chen's recent breakthrough result~\cite{chen2021almost} for KLS that $\psi \leq \exp (C \sqrt{\log d \log \log d})$
for some universal constant $C$, and a simple application of PCA, is that we can cluster a mixture of translates of a logconcave distribution in polynomial time, as long as the separation is $\Omega (\exp (C \sqrt{\log k \log \log k}))$.

\subsection{Our Techniques}

In this section, we describe how our techniques work at a high level.
Our goal will be to devise a method which can, given samples $X, X' \sim \mcl{M}$, detect whether or not $X$ and $X'$ are from the same components or from different ones.

We first make the following reduction.
Notice that if $\mcl{M}$ is a mixture with separation $\Delta$, then $(X - X') / \sqrt{2}$ can be thought of as a sample from the \emph{difference mixture} $\mcl{M}'$.
This is a mixture with $\binom{k}{2} + 1$ components, each with covariance $I$.
It has one component with mean zero, and the means of the remaining components all have norm at least $\Delta / \sqrt{2}$.
Moreover, given $X, X' \sim \mcl{M}$, we have that $X - X'$ is drawn from the mean zero component of the difference mixture if and only if  $X, X'$ were drawn from the same component in the original mixture.
Thus, to check if two samples from the original mixture $X, X'$ are from the same component, it suffices to be able to detect, given a sample from the difference mixture, whether or not this sample comes from the mean zero component or not.

In the remainder of this section, in a slight abuse of notation, we will let $\mcl{M}$ denote the difference mixture, we will assume it has $k$ components with nonzero mean, and we will assume that all nonzero means of the difference mixture have norm at least $\Delta$.
We will henceforth refer to the components with nonzero mean as the nonzero components of the mixture.
This reparameterization of the problem only changes things by polynomial factors, which do not impact our qualitative results.

\subsubsection{Rough clustering via implicit moment tensors}
The main conceptual contribution of our paper is a novel way to implicitly access degree $t = O(\log k / \log \log k)$ moment information with polynomially many samples and time.
We do so by carefully constructing an implicitly maintained projection map from $\R^{d^t}$ down to a subspace of dimension $k$, which still preserves meaningful information about the nonzero components.
For now, let us first focus on the case where the maximum norm of any mean in the difference mixture is upper bounded by $\poly (\Delta)$, or equivalently, the maximum separation between any two components in the original mixture is at most polynomially larger than the minimum separation.

\paragraph{Low rank estimators for Hermite polynomials} 
Central to our methods is the \emph{$t$-th Hermite polynomial tensor}, a classical object in probability theory, which we denote $h_t: \R^{d} \to \R^{d^t}$. 
These are explicit polynomials, and are the natural analog of the univariate Hermite polynomials to high dimensions. 
A well-known fact about the Hermite polynomial tensor is that 
\begin{equation}
\label{eq:hermite-gaussian}
\E_{X \sim  N (\mu, I)} [h_t(X)] = \mu^{\otimes t} \; .
\end{equation}
Throughout the introduction, we will treat all tensors as flattened into vectors in $\R^{d^t}$ (in a canonical way) e.g. we view the RHS of the above as a vector in $\R^{d^t}$.
One simple but important implication of~\eqref{eq:hermite-gaussian} is that
\begin{equation}
\label{eq:hermite-mixture}
\E_{X \sim \mcl{M}} [h_t(X)] = \sum_{i = 1}^k w_i \mu_i^{\otimes t} \; .
\end{equation}
This fact will be crucial for us going forward.

However, a major bottleneck for algorithmically using the Hermite polynomial tensors is that we cannot write down $h_{t} (X)$ in polynomial time when $t$ gets large, e.g. when $t = \Omega (\log k / \log \log k)$, which is the regime we will require.

To get around this, we will use a modification of the Hermite polynomial tensors that can still be used to estimate the RHS of (\ref{eq:hermite-mixture}) but can also be easily manipulated implicitly.  In particular, we will construct a \emph{random} polynomial $R_{t}: \R^d \to \R^{d^t}$ i.e. we can imagine $R_t$ is actually a polynomial in $x$ whose coefficients are randomly  chosen.  The polynomial $R_t$ (constructed in Corollary \ref{coro:rank1-hermite-identity-p2}) satisfies two key properties.
\begin{enumerate}
\item $R_t (x)$ is an unbiased estimator for $h_t(x)$ with bounded variance i.e. for a fixed $x$, $\E [R_t(x)] = h_t(x)$ where the expectation is over the random coefficients of $R_t$.
%for $h_t(x)$ with bounded variance i.e. $\E_{R_t}[R_t(x)] = h_t(x)$.  Note that this implies
\item For any choice of the randomness in the coefficients, the polynomial $R_t (x)$ can be written as a \emph{sum of polynomially many rank-$1$ tensors} i.e. tensors of the form 
%Here, we say that a vector $v \in \R^{d^t}$ is a rank-$1$ tensor (or more precisely, the flattening of a rank-$1$ tensor, but we will abuse notation slightly here) if 
\[
v = v_1 \otimes \ldots \otimes v_t \; , \; \mbox{where } v_i \in \R^d \; \mbox{for all} \; i = 1, \ldots, t \; .
\]
\end{enumerate}
Note that the first property implies that 
\begin{equation}
\label{eq:r-expectation}
\E_{R_t, X \sim  N (\mu, I)} [R_t(X)] = \mu^{\otimes t} \; , \; \mbox{and } \; \E_{R_t, X \sim \mcl{M}} [R_t(X)] = \sum_{i = 1}^k w_i \mu_i^{\otimes t} \; .
\end{equation}
\noindent
The second property is the main motivation behind the definition of $R_t$, as it implies that we can have efficient, but restricted access to it.
The key point is that if our algorithm can be implemented with techniques that only require accesses to rank-$1$ tensors, we can implicitly access $R_t$ in polynomial time.

\paragraph{Implicitly finding the span of the tensorized means}
Motivated by the above discussion, our goal will be to find a projection matrix $\Pi: \R^{d^t} \to \R^{k}$ with the following properties:
\begin{enumerate}[(i)]
\item \textbf{Efficient application} If $v \in \R^{d^t}$ is a rank-$1$ tensor, then $\Pi v$ can be evaluated in polynomial time. \label{enum:fast}
\item \textbf{Zero component is small} If $X \sim N(0, I)$, then $\norm{\Pi R_t (X)}_2 < k^{50}$ with high probability. \label{enum:small}
\item \textbf{Nonzero components are large} If $X \sim \mcl{M}$ is a sample from a nonzero component of the difference mixture, then $\norm{\Pi R_t (X)}_2 \geq k^{100}$ with high probability.  \label{enum:large}
\end{enumerate}
Given such a projection map, our clustering procedure is straightforward: given two samples $X, X'$ from the original mixture, we apply the projection map to many copies of $R_t ((X - X') / \sqrt{2})$, and cluster them in the same component if and only if their projected norm is small on average.
We show that the above properties, as well as some facts about the concentration of $R_t$, imply that this clustering algorithm succeeds with high probability, assuming we have access to $\Pi$.

It thus remains how to construct $\Pi$.
In fact, there is a natural candidate for such a map.
Let $\mu_1, \ldots, \mu_k$ denote the means of the nonzero components, and let 
\[
S_t = \mathrm{span} \left( \left\{ \mu_i^{\otimes t} \right\}_{i = 1}^k \right) \; .
\]
If we could find the projection $\Gamma_t: \R^{d^t} \to \R^k$ onto the subspace $S_t$, it can be verified that this projection map would satisfy Conditions~(\ref{enum:small}) and~(\ref{enum:large}).\footnote{Here and throughout the introduction, we will assume for simplicity of exposition that the vectors $\mu_1^{\otimes s}, \ldots, \mu_k^{\otimes s}$ are linearly independent, for all $s = 1, \ldots, t$, so that $S_s$ is always a $k$-dimensional subspace, for all $s = 1, \ldots, t$. 
In general, our algorithms work even if they are not linearly independent, and will always find a subspace which contains $S_t$, which will suffice for our purposes.}

Moreover, there is a natural estimator for this subspace.  In particular \eqref{eq:r-expectation} implies that $\E[ R_{2t}(X) ]$ rearranged as a $d^t \times d^t$ matrix in a canonical way is exactly
%In fact,~\eqref{eq:hermite-mixture} implies that if we could estimate $\E [h_{2t} (X)] \in \R^{d^{2t}}$, and rewrite this as a $\R^{d^t} \times \R^{d^t}$-sized matrix in the canonical way, then this matrix is exactly
\[
\sum_{i = 1}^k w_i \left( \mu_i^{\otimes t} \right) \left( \mu_i^{\otimes t} \right)^\top \; .
\]
Notice that this matrix is rank $k$, and moreover, the span of its nonzero eigenvectors is exactly $S_t$.
Consequently, if we could estimate this quantity, then find the projection onto the span of the top $k$ eigenvectors, we would be done.

As alluded to earlier, the difficulty is that doing this naively would not be efficient; writing down any of these objects would take quasipolynomial time.
Instead, we seek to find an implicit representation of $\Gamma_t$ with the key property that it can be applied to rank-$1$ tensors in polynomial time.

We will do so by iteratively building an approximation to this subspace.
Namely, we give a method which, given a good approximation to $\Gamma_{s - 1}: \R^{d^{s - 1}} \to \R^k$ which can be efficiently applied to flattenings of rank-$1$ tensors, constructs a good approximation to $\Gamma_s$ with the same property.
We do so by obtaining a good approximation to the $dk \times dk$ sized matrix
\begin{equation}
\label{eq:mt}
M_t = \sum_{i=1}^k w_i \left( \mu_i \otimes \Gamma_{s - 1} \mu_i^{\otimes (s - 1)} \right)\left( \mu_i \otimes \Gamma_{s - 1} \mu_i^{\otimes (s - 1)} \right)^\top \; .
\end{equation}
Notice that $M_s$ has rank $k$, and moreover, the span of its $k$ largest eigenvectors is equal to the span of $\left\{  \mu_i \otimes \Gamma_{s - 1} \mu_i^{\otimes (s - 1)} \right\}_{i = 1}^k$.
Therefore, if we let $\Pi_s: \R^{dk} \to \R^k$ denote the projection onto the span of the $k$ largest eigenvectors of $M_s$, then one can easily verify that 
\begin{equation}
\label{eq:projection-map}
\Gamma_s = \Pi_s \left(I \otimes \Gamma_{s - 1} \right) \; .
\end{equation}
Moreover, if $\Gamma_{s - 1}$ can be efficiently applied to flattenings of rank-$1$ tensors in $\R^{d^{s - 1}}$, then the form of~(\ref{eq:projection-map}) immediately implies that $\Gamma_s$ can also be applied efficiently to flattenings of rank-$1$ tensors in $\R^{d^s}$.

It remains to demonstrate how to efficiently approximate $M_s$, given $\Gamma_{s - 1}$.
There is again a fairly natural estimator for this matrix.
Namely, each component of the sum in~\eqref{eq:mt} can be formed by rearranging the length-$(dk)^2$ vector
\[
\left( I \otimes \Gamma_{s - 1} \right)^{\otimes 2} \mu_i^{\otimes 2s} = \E_{R_{2s}, X \sim N(\mu_i, I)} \left[ \left( I \otimes \Gamma_{s - 1} \right)^{\otimes 2} R_{2s} (X)  \right] \; .
\]
into a $dk \times dk$ sized matrix in the canonical way.
In particular, this implies that the overall matrix is the rearrangement of the length-$(dk)^2$ vector 
\begin{equation}
\label{eq:flattened-mt}
\sum_{i = 1}^k w_i \left( I \otimes \Gamma_{s - 1} \right)^{\otimes 2} \mu_i^{\otimes 2s} = \E_{R_{2s}, X \sim \mcl{M}} \left[ \left( I \otimes \Gamma_{s - 1} \right)^{\otimes 2} R_{2s} (X)  \right]
\end{equation}
into a $dk \times dk$ sized matrix in the canonical way.
Since $R_{2s}$ is a sum of polynomially many rank-$1$ tensors, and by our inductive hypothesis, $\Gamma_{s - 1}$ can be efficiently applied to rank-$1$ tensors, we can efficiently estimate the right hand side of~\eqref{eq:flattened-mt}, given samples from $\mcl{M}$.

Putting it all together, this allows us to approximate $M_s$ efficiently given $\Gamma_{s - 1}$, which, by~\eqref{eq:projection-map}, gives us the desired expression for $\Gamma_s$.
Iterating this procedure gives us a way to estimate $\Gamma_t$ as a sequence of nested projection maps, i.e.
\[
\Gamma_t \approx \Pi_t \left( I \otimes \Pi_{t - 1} \left( I \otimes \ldots \right) \right) \; ,
\]
where we can compute $\Pi_1, \ldots, \Pi_t$ efficiently, given samples from $\mcl{M}$.
This form allows us to evaluate $\Gamma_t$ efficiently on any flattening of a rank-$1$ tensor, thus satisfying Condition~(\ref{enum:fast}), and we previously argued that $\Gamma_t$ satisfies Conditions~(\ref{enum:small}) and~(\ref{enum:large}).
Combining all of these ingredients gives us our clustering algorithm, when the minimum and maximum separations are at most polynomially separated.

\paragraph{Implicit moment tensors for Poincar\'{e} distributions} So far, we have only discussed how to do this implicit moment estimation for isotropic Gaussians.
It turns out that all of the quantities discussed above have very natural analogues for \emph{any} Poincar\'{e} distribution.
For instance, given any Poincar\'{e} distribution $\mcl{D}$ with zero mean, there is an explicit polynomial tensor we call the \emph{$\mcl{D}$-adjusted polynomial} $P_{t, \mcl{D}}$ (see Section \ref{sec:moments1}) that essentially satisfies all the same properties that we needed above.
If we use these polynomials instead of the Hermite polynomial tensor, it turns out that all of these proofs directly lift to any Poincar\'{e} distribution.
In fact, in the actual technical sections, we directly work with arbitrary Poincar\'{e} distributions, as everything is stated very naturally there.
The resulting clustering algorithm immediately gives us Theorem~\ref{thm:main-Poincare}.

\paragraph{Sample complexity} Previous approaches always needed to estimate high degree moment tensors, and as a result, their sample complexity was quasipolynomial.
While we may also need to estimate fairly high degree polynomials, notice that all quantities that we will deal with will be polynomially bounded.
This is because we can terminate our procedure at any $t$ which satisfies Conditions~(\ref{enum:small}) and~(\ref{enum:large}).
Therefore, all the quantities that we need are polynomially large. 
As a result, one can verify that the polynomials we construct will also only ever have polynomially large range, with high probability.
Therefore, all quantities we need to estimate can be estimated using polynomially many samples.
We defer the detailed proofs outlined in this discussion to Sections \ref{sec:moments1}, \ref{sec:projection}, \ref{sec:implicit-est} and \ref{sec:testing1}.  Note that in those sections, we work with a general Poincar\'{e} distribution but the outline follows the description here.

\subsubsection{Fine-grained clustering for Gaussians}
We now discuss how to handle general mixtures of isotropic Gaussians, without any assumption on the maximum separation.
The problem with applying the implicit moment estimation method outlined above to this general setting is that the signal from the components in the difference mixture with relatively small mean will be drowned out by the signal from the components with much larger norm.
Consequently, we can reliably cluster points from the components with large mean, but we could obtain an imperfect clustering for some components with somewhat smaller mean, and we will be unable to detect components with very small mean.

To overcome this, we devise a recursive clustering strategy.
One somewhat simple approach is as follows.
We first use our rough clustering algorithm described above to find a ``signal direction'' $v \in \R^d$.
This direction will have the property that there is a pair of well-separated means along this direction.
Thus, if we project the data points on this direction, and take only points which lie within a randomly chosen small interval on this interval, we can guarantee that with reasonable probability, we only accept points from at most half of the components of the mixture.  Of course, after restricting to this interval, the resulting distribution is no longer a mixture of Gaussians.  However, if we consider the projection of these accepted points to the subspace orthogonal to $v$, the resulting distribution will again be a mixture of fewer isotropic Gaussians.
We can then recurse on this mixture with fewer components.
Here, we are crucially using the fact that isotropic Gaussians remain isotropic Gaussians after slicing and projecting orthogonally.

While this strategy described above, when implemented carefully, would work down to $\Delta = \poly (\log k)$, it would not be able to achieve the nearly optimal separation in Theorem \ref{thm:main-GMM}.
To achieve the nearly optimal separation, there are several more technical details that need to be dealt with and thus there will be several additional steps in the algorithm.  We defer the details of this to Sections \ref{sec:gaussian-moments}, \ref{sec:clustering1} and \ref{sec:clustering2}.
%as well as slightly refine the analysis of the implicit moment estimation method when the samples come from a Gaussian, as opposed to an arbitrary Poincar\'{e} distribution.

\subsection{Related work}
The literature on mixture models---and Gaussian mixture models in particular---dates back to seminal work of Pearson~\cite{pearson1894contributions} and is incredibly vast. 
For conciseness, we will focus only on the most related papers here.
Our results are most closely related to the aforementioned line of work on studying efficient algorithms for clustering and parameter estimation under mean-separation conditions~\cite{dasgupta1999learning,dasgupta2007probabilistic,arora2005learning,vempala2004spectral,diakonikolas2018list,hopkins2018mixture,regev2017learning,kothari2018robust,diakonikolas2020small}.
However, none of these algorithms can handle polylogarithmic separation in polynomial time.

We also note that a number of papers also generalize from mixtures of Gaussians to mixtures of more general classes of distributions~\cite{achlioptas2005spectral,kumar2010clustering,mixon2017clustering,hopkins2018mixture,kothari2018robust}.
These algorithms fall into two classes: either they require separation which is at least $\Omega (k^{1/2})$ or even larger, but they can handle general subgaussian distributions, or they require more structure on the higher moments of the distribution, but they can tolerate much less separation. 
The most general condition under which the latter is known to work is the condition commonly referred to as \emph{certifiable hypercontractivity}, which roughly states that the Sum-of-Squares hierarchy can certify that the distribution has bounded tails.
While there is no complete characterization of what distributions satisfy this condition, the most general class of distributions for which it is known to hold is the class of Poincar\'{e} distributions~\cite{kothari2018robust}, which is also the class of distributions we consider here.

Another line of work focuses on parameter estimation for mixtures models without separation conditions~\cite{kalai2010efficiently,belkin2015polynomial,moitra2010settling,hardt2015tight,diakonikolas2020small}. 
These papers typically make much weaker assumptions on the components, namely, that they are statistically distinguishable, and the goal is to recover the parameters of the mixture, even in settings where clustering is impossible.
However, typically, these algorithms incur a sample complexity and runtime which is exponential in the number of components; indeed,~\cite{hardt2015tight} demonstrate that this is tight, even in one dimension, when the Gaussians can have different variances.

To circumvent this, researchers have also considered the easier notion of \emph{proper} or \emph{semi-proper} learning, where the goal is to output a mixture of Gaussians which is close to the unknown mixture in statistical distance.
A learning algorithm is said to be proper if its output is a mixture of $k$ Gaussians, where $k$ is the number of components in the true mixture, and semi-proper if it outputs a mixture of $k' \geq k$ Gaussians.
While the sample complexity of proper learning is polynomial in all parameters, all known algorithms still incur a runtime which is exponential in $k$, even in the univariate setting~\cite{feldman2008learning,daskalakis2014faster,acharya2014near,li2017robust,ashtiani2018nearly}.
When the hypothesis is only constrained to be semi-proper, polynomial time algorithms are known in the univariate setting~\cite{devroye2001combinatorial,bhaskara2015sparse,li2021sparsification}, but these do not extend to high dimensional settings.
In the more challenging high dimensional regime, a remarkable recent result of~\cite{diakonikolas2020small} demonstrate that for a mixture of isotropic Gaussians, one can achieve semi-proper learning with quasipolynomial sample complexity and runtime.
They do this by explicitly constructing a small cover for the candidate means, by techniques inspired by algebraic geometry.

An even weaker goal that has been commonly considered is that of \emph{density estimation}, where the objective is to output any hypothesis which is close to the unknown mixture in statistical distance.
Efficient (in fact, nearly optimal) algorithms are again known for this problem in low dimensions, see e.g.~\cite{devroye2001combinatorial,chan2014near,chan2014efficient,acharya2017sample}, but as was the case with proper learning, these techniques do not extend nicely to high dimensional settings.

An alternate assumption that has been considered in the literature is that the means satisfy some algebraic non-degeneracy assumptions.
For instance, such assumptions are satisfied in smoothed analysis settings~\cite{hsu2013learning,anderson2014more,bhaskara2014smoothed,ge2015learning}.
Often in these settings, access to constant order moments suffice (e.g. $3$rd or $4$th moments), although we note that~\cite{bhaskara2014smoothed} is a notable exception.
In contrast, we do not make any non-degeneracy assumptions, and our methods need to access much higher moments.

We also mention a line of work studying the theoretical behavior of the popular \emph{expectation-maximization} (EM) algorithm for learning mixtures of Gaussians~\cite{dasgupta2007probabilistic,daskalakis2017ten,xu2016global}.
However, while the above works can prove that the dynamics of EM converge in limited settings, it is known that EM fails to converge in general, even for mixtures of 3 Gaussians~\cite{wu1983convergence,daskalakis2017ten}.

Finally, we note that our work bears some vague resemblance to the general line of work that uses spectral-based methods to speed up Sum-of-Squares (SoS) algorithms.
Spectral techniques have been used to demonstrate to speed up SoS-based algorithms in various settings such as tensor decomposition~\cite{hopkins2016fast,ma2016polynomial,schramm2017fast,hopkins2019robust} and refuting random CSPs~\cite{raghavendra2017strongly}.
Similarly, it has been observed that in some settings, SoS-based algorithms can be sped up, when the SoS proofs are much smaller than the overall size of the program~\cite{guruswami2012faster,steurer2021sos}.
Our algorithm shares some qualitative similarities with some of these approaches---for instance, it is based on a (fairly complicated) spectral algorithm.
However, we do not know of a concrete connection between our algorithm and this line of work.
It is possible, for instance, that our algorithm could be interpreted as extracting a specific randomized polynomially-sized SoS proof of identifiability, but we leave further investigations of this to future work.

\section{Formal Problem Setup and Our Results}

In this section, we formally define the problems we consider, and state our formal results.
For the remainder of this paper, we will always let $\norm{\cdot}$ denote the $\ell_2$ norm.

\subsection{Clustering Mixtures of Poincare Distributions}\label{sec:setup-poincare}

The general problem that we study involves clustering mixtures of Poincare distributions.  We begin with a few definitions.

\begin{definition}[Poincare Distribution]\label{def:poincare}
For a parameter $\sigma$, we say a distribution $\mcl{D}$ on $\R^d$ is $\sigma$-Poincare if for all differentiable functions $f: \R^d \rightarrow \R$, 
\[
\Var_{z \sim \mcl{D}}[ f(z) ] \leq \sigma^2 \E_{z \sim \mcl{D}}[ \norm{\nabla f(z)}^2]  \,.
\]
\end{definition}

\begin{definition}
Let $\mcl{D}$ be a distribution on $\R^d$.  We use $\mcl{D}(\mu_1)$ for $\mu_1 \in \R^d$ to denote the distribution obtained by shifting $\mcl{D}$ by the vector $\mu_1$.
\end{definition}

We assume that there is some $\sigma$-Poincare distribution $\mcl{D}$ on $\R^d$ that we have sample access to.  Since everything will be scale invariant, it will suffice to focus on the case $\sigma = 1$.  For simplicity we assume that $\mcl{D}$ has mean $0$ (it is easy to reduce to this case since we can simply estimate the mean of $\mcl{D}$ and subtract it out).    We also assume that we have sample access to a mixture 
\[
\mcl{M} = w_1 \mcl{D}(\mu_1) + \dots + w_k \mcl{D}(\mu_k)
\]
where the mixing weights $w_1, \dots , w_n$ and means $\mu_1, \dots , \mu_k$ are unknown.  We will assume that we are given a lower bound on the mixing weights $w_{\min}$.  We consider the setting where there is some minimum separation between all pairs of means $\mu_i , \mu_j$ so that the mixture is clusterable.  In the proceeding sections, when we say an event happens with high probability, we mean that the failure probability is smaller than any inverse polynomial in $k, 1/w_{\min}$.  Our main theorem is stated below.

\begin{theorem}\label{thm:main-Poincare}
Let $\mcl{D}$ be a $1$-Poincare distribution on $\R^d$.  Let 
\[
\mcl{M} = w_1 \mcl{D}(\mu_1) + \dots + w_k \mcl{D}(\mu_k)
\]
be a mixture of translated copies of $\mcl{D}$.  Let $w_{\min}, s$ be parameters such that $w_i \geq w_{\min}$ for all $i$ and $\norm{\mu_i - \mu_j} \geq s$ for all $i \neq j$.  Let $\alpha = (w_{\min}/k)^{O(1)}$ be some desired accuracy ( that is inverse polynomial) \footnote{If $d$ is much larger than $k$ and we wanted inverse polynomial accuracy like $1/d$ then we can simply decrease the parameter $w_{\min}$ (and then we would need separation $\log (dk)$ instead of $\log k$)}. Assume that 
\[
s \geq (\log (k/w_{\min} ))^{1 + c}
\]
for some $0 < c < 1$.  Also assume that 
\[
\max \norm{\mu_i - \mu_j} \leq s^{C}
\]
for some $C$.  There is an algorithm that takes $n = \poly((kd/(w_{\min}\alpha) )^{C/c} )$ samples from $\mcl{M}$ and $\mcl{D}$ and runs in $\poly(n)$ time and with high probability, outputs estimates 
\[
\wt{w_1}, \dots , \wt{w_k}, \wt{\mu_1}, \dots , \wt{\mu_k}
\]
such that for some permutation $\pi$ on $[k]$, 
\[
|w_i - \wt{w_{\pi(i)}}| \leq \alpha , \norm{\mu_i - \wt{\mu_{\pi(i)}}} \leq \alpha
\]
for all $i$.
\end{theorem}
\begin{remark}
If we could remove the assumption $\norm{\mu_i - \mu_j} \leq s^{C}$, then we would get a complete polynomial time learning result.  Still, our learning algorithm works in polynomial time for mixtures where the maximum separation is polynomially bounded in terms of the minimum separation.  
%For arbitrary mixtures (note that by the reduction in Section \ref{sec:reduce-max-norm}, we can assume that the maximum separation is $\poly(k)$), we recover the quasipolynomial guarantees of previous algorithms.
\end{remark}

An immediate consequence of Theorem \ref{thm:main-Poincare} is that we can cluster samples from the mixture with accuracy better than any inverse polynomial.  
\begin{corollary}\label{coro:cluster-Poincare}
Under the same assumptions as Theorem \ref{thm:main-Poincare}, we can recover the ground-truth clustering of the samples with high probability i.e. we output $k$ clusters $S_1, \dots S_k$ such that for some permutation $\pi$ on $[k]$, the set $S_{\pi(i)}$ consists precisely of the samples from the component $\mcl{D}(\mu_i)$ for all $i$. 
\end{corollary}
%Note that the separation of $(\log (k/w_{\min}))^{1 + c}$ is essentially optimal for accurate clustering to be possible ($\Omega(\log (k/w_{\min}))$ is necessary).  

\subsection{Clustering Mixtures of Gaussians}\label{sec:setup-Gaussians}

In the case where the distribution $\mcl{D}$ in the setup in Section \ref{sec:setup-poincare} is a Gaussian, we can obtain a stronger result that works in full generality, without any assumption about the maximum separation.  The result for Gaussians also works with a smaller separation of $(\log (k/w_{\min}))^{1/2 + c}$ which, as mentioned before, is essentially optimal for Gaussians.

\begin{theorem}\label{thm:main-GMM}
Let $\mcl{M} = w_1N(\mu_1, I) + \dots + w_k N(\mu_k, I)$ be an unknown mixture of Gaussians in $\R^d$ such that $w_i \geq w_{\min}$ for all $i$ and $\norm{\mu_i - \mu_j} \geq (\log (k/w_{\min}))^{1/2 + c}$ for some constant $c > 0$.  Then for any desired (inverse polynomial) accuracy $\alpha \geq (w_{\min}/k)^{O(1)}$, given $n = \poly(( dk/(w_{\min} \alpha ))^{1/c} )$ samples and $\poly(n)$ runtime, there is an algorithm that with high probability outputs estimates $\{ \wt{\mu_1}, \dots , \wt{\mu_k} \}$ and $\{ \wt{w_1}, \dots , \wt{w_k} \}$ such that for some permutation $\pi$ on $[k]$, we have
\[
|w_i - \wt{w_{\pi(i)}}|, \norm{\mu_i - \wt{\mu_{\pi(i)}}} \leq \alpha
\]
for all $i \in [k]$.
\end{theorem}

Again, once we have estimated the parameters of $\mcl{M}$, it is easy to cluster samples from $\mcl{M}$ into each of the components with accuracy better than any inverse polynomial.  
\begin{corollary}\label{coro:cluster-GMM}
Under the same assumptions as Theorem \ref{thm:main-GMM}, with high probability, we can recover the ground-truth clustering of the samples i.e. we output $k$ clusters $S_1, \dots S_k$ such that for some permutation $\pi$ on $[k]$, the set $S_{\pi(i)}$ consists precisely of the samples from the component $N(\mu_i, I)$ for all $i$. 
\end{corollary}

Note that throughout this paper, we do not actually need to know the true number of components $k$.  All of the algorithms that we write will work if instead of the number of components being $k$, the number of components is upper bounded by $k$ i.e. we only need to be told an upper bound on the number of components.  In fact, we can simply use $1/w_{\min}$ as the upper bound on the number of components. 

%we don't even need to know this upper bound because  we can simply guess $k = 1,2, \dots $, run our learning algorithm and then test if our output hypothesis is close to the truth (see Lemma \ref{lem:test-hypothesis1} and Lemma \ref{lem:test-hypothesis2}).  
%Thus, we actually only need to know an upper bound on the number of components. 

\subsection{Organization}

In Section \ref{sec:prelims}, we introduce basic notation and prove a few basic facts that will be used later on.

\paragraph{Clustering Test:} In Sections \ref{sec:moments1} - \ref{sec:testing1}, we develop our key clustering test i.e. we show how to test if two samples are from the same component or not.  Note that throughout these sections, when we work with a mixture
\[
\mcl{M} = w_1 \mcl{D}(\mu_1) + \dots + w_k \mcl{D}(\mu_k) \,,
\]
this will correspond to the ``difference mixture" of the mixture that we are trying to learn and thus our goal will be to test whether a sample came from a component with $\mu_i = 0$ or $\mu_i$ far from $0$.

In Section \ref{sec:moments1}, we discuss how to construct estimators for the moments of a Poincare distribution that can be manipulated implicitly.  In the end, we construct a random polynomial $R_t$ such that 
\[
\E_{x \sim \mcl{D}(\mu)}[ R_t(x) ] = \mu^{\otimes t}
\]
and such that $R_t$ can be written as the sum of polynomially many rank-$1$ tensors (see Corollary \ref{coro:rank1-identity-p2}).  In Section \ref{sec:projection}, we describe our iterative projection technique and explain how it can be used to implicitly store and apply a projection map $\Gamma: \R^{d^t} \rightarrow \R^k$ in polynomial time to rank-$1$ tensors.  In Section \ref{sec:implicit-est}, we combine the techniques in Section \ref{sec:moments1} and Section \ref{sec:projection} to achieve the following:  given samples from
\[
\mcl{M} = w_1 \mcl{D}(\mu_1) + \dots + w_k \mcl{D}(\mu_k)
\]
we can find a $k \times d^t$ projection matrix $\Gamma_t$ (where $t \sim \log k/ \log \log k$) whose row span essentially contains all of $\mu_1^{\otimes t}, \dots \mu_k^{\otimes t}$ (see Lemma \ref{lem:projection-accuracy}).  Finally, in Section \ref{sec:testing1}, we use the projection map computed in the previous step and argue that if $x \sim \mcl{D}(0)$ then  $\norm{\Gamma_t R_t(x)}$ is small and if $x \sim \mcl{D}(\mu_i)$ for $\mu_i$ far from $0$, then $\norm{\Gamma_t R_t(x)}$ is large with high probability.  Thus, to test a sample $x$, it suffices to measure the length of $\norm{\Gamma_t R_t(x)}$.

\paragraph{Main Result for Mixtures of Poincare Distributions:} In Section \ref{sec:poincare-main}, we prove Theorem \ref{thm:main-Poincare}, our main result for mixtures of Poincare distributions.  It will follow fairly easily from the guarantees of the clustering test in Section \ref{sec:testing1}.

\paragraph{Main Result for Mixtures of Gaussians:}  Proving our main result for mixtures of Gaussians requires some additional work although all of the machinery from Sections \ref{sec:moments1} - \ref{sec:testing1} can still be used.  In particular, we will need a few quantitatively stronger versions of the bounds in Sections \ref{sec:moments1} - \ref{sec:testing1} that exploit special properties of Gaussians in order to get down from $\log^{1 + c} k $ to $\log^{1/2 + c} k $ separation.  We prove these stronger bounds in Section \ref{sec:gaussian-moments}.  Then in Sections \ref{sec:clustering1} and \ref{sec:clustering2}, we show how to do recursive clustering to eliminate the assumption that the maximum separation is polynomially bounded in terms of the minimum separation.  In Section \ref{sec:clustering1}, we introduce a few basic building blocks in our recursive clustering algorithm and we put them together in Section \ref{sec:clustering2}.

\section{Preliminaries}\label{sec:prelims}

We now introduce some notation that will be used throughout the paper.  We use $I_n$ to denote the $n \times n$ identity matrix.  For matrices $A,B$ we define $A \otimes_{\kr} B$ to be their Kronecker product.  This is to avoid confusion with our notation for tensor products.  For a tensor $T$, we use $\flatten(T)$ to denote the flattening of $T$ into a vector.  We assume that this is done in a canonical way throughout this paper.

\subsection{Manipulating Tensors}

We will need to do many manipulations with tensors later on so we first introduce some notation for working with tensors.

\begin{definition}[Tensor Notation]\label{def:tensor-notation}
For an order-$t$ tensor, we index its dimensions $\{1,2, \dots , t \}$.  For a partition of $[t]$ into subsets $S_1, \dots  , S_a$ and tensors $T_1, \dots , T_a$ of orders $|S_1|, \dots , |S_a|$ respectively we write
\[
T_1^{(S_1)} \otimes \dots \otimes T_a^{(S_a)}
\]
to denote the tensor obtained by taking the tensor product of $T_1$ in the dimensions indexed by $S_1$, $T_2$ in the dimensions indexed by $S_2$, and so on for $T_3, \dots , T_a$.
\end{definition}
\begin{definition}
For a vector $x$ (viewed as an order-$1$ tensor), we will use the shorthand $x^{\otimes S}$ to denote $(x^{\otimes |S|})^{(S)}$ (i.e. the product of copies $x$ in dimensions indexed by elements of $S$).  For example,
\[
x^{\otimes \{1,3 \} } \otimes y^{\otimes \{2,4 \} }  = x \otimes y \otimes x \otimes y \,.
\]
\end{definition}

\begin{definition}[Tensor Slicing]\label{def:tensor-slicing}
For an order-$t$ tensor $T$, we can imagine indexing its entries with indices $(\eta_1, \dots , \eta_t ) \in [d_1] \times \dots \times [d_t]$ where $d_1, \dots , d_t$ are the dimensions of $T$.  We use the notation 
\[
T_{\eta_{a_1} =  b_1, \dots , \eta_{a_j} = b_j}
\]
to denote the slice of $T$ of entries whose indices satisfy the constraints $\eta_{a_1} =  b_1, \dots , \eta_{a_j} = b_j$.
\end{definition}

\begin{definition}[Unordered Partitions]
We use $Z_t(S)$ to denote all partitions of $S$ into $t$ unordered, possibly empty, parts.  
\end{definition}
\begin{remark}
Note the partitions are not ordered so $\{ \{1 \}, \{ 2 \} \}$ is the same as $\{ \{2 \} , \{ 1 \} \}$. 
\end{remark}

\begin{definition}
 For a collection of sets $S_1, \dots , S_t$, we define $\mcl{C}(\{ S_1, \dots , S_t\})$ to be the number of sets among $S_1, \dots , S_t$ that are nonempty. 
\end{definition}

\begin{definition}[Symmetrization]
Let $A_1, \dots , A_n$ be tensors such that $A_i$ is an order $a_i$ tensor having dimensions $\underbrace{d \times \dots \times d}_{a_i}$ for some integers $a_1, \dots , a_n$.  We define 
\[
\sum_{sym} (A_1 \otimes \dots \otimes A_n) = \sum_{ \substack{S_1 \cup \dots \cup S_n = [a_1 + \dots + a_n] \\ S_i \cap S_j = \emptyset \\ |S_i| = a_i} } A_1^{(S_1)} \otimes \dots \otimes A_n^{(S_n)} \,.
\]
In other words, we sum over all ways to tensor $A_1, \dots , A_n$  together to form a tensor of order $a_1 + \dots + a_n$.
\end{definition}

\subsection{Properties of Poincare Distributions}

Here we state a few standard facts about Poincare distributions that will be used later on.

\begin{fact}\label{fact:basic-Poincare}
Poincare distributions satisfy the following properties:
\begin{itemize}
    \item \textbf{Direct Products:} If $\mcl{D}_1$ and $\mcl{D}_2$ are $\sigma$-Poincare distributions then their product $\mcl{D}_1 \times \mcl{D}_2$ is $\sigma$-Poincare 
    \item \textbf{Linear Mappings:} If $\mcl{D}$ is $\sigma$-Poincare and $A$ is a linear mapping then the distribution $Ax$ for $x \sim \mcl{D}$ is $\sigma \norm{A}_{\textsf{op}}$-Poincare
    \item \textbf{Concentration:} If $\mcl{D}$ is $\sigma$-Poincare then for any $L$-Lipchitz function $f$ and any parameter $t \geq 0$, we have
    \[
    \Pr_{z \sim \mcl{D}}[ |f(z) - \E[f(z)]| \geq t ] \leq 6e^{-t/(\sigma L)} \,.
    \]
\end{itemize}
\end{fact}

The following concentration inequality for samples from a Poincare distribution is also standard.

\begin{claim}\label{claim:Poincare-mean}
Let $\mcl{D}$ be a distribution in $\R^d$ that is $1$-Poincare.  Let $0 < \eps < 0.1$ be some parameter.  Given $n \geq (d/\eps)^8$ independent samples $z_1, \dots , z_n \sim \mcl{D}$, with probability at least $1 - 2^{-d/\eps}$, we have
\[
\norm{\frac{z_1 + \dots + z_n}{n} - \E_{z \sim \mcl{D}}[z] } \leq \eps \,.
\]
%the following property: for any subset $S \subset [n]$ of size at least $(1 - \eps^2)n $, we have
%\[
%\norm{\frac{1}{|S|}\sum_{i \in S}z_i - \E_{z \sim \mcl{D}}[z] } \leq \eps \,.
%\]
\end{claim}

%\begin{corollary}\label{coro:basic-Poincare}
%Let $\mcl{D}$ be a distribution on $\R^d$ that is $s$-Poincare and has mean $0$.  Then for any vector $v \in \R^d$ with $\norm{v} < 1$,
%\[
%\E_{z \sim \mcl{D}}[ e^{|v \cdot z| /s} ] \leq \frac{6}{1 - \norm{v}} \,.
%\]
%\end{corollary}
%\begin{proof}
%Note that using Fact \ref{fact:basic-Poincare},
%\[
%\E_{z \sim \mcl{D}}[ e^{|v \cdot z| /s} ] \leq 1 + \int_{t = 0}^{\infty} \Pr_{z \sim \mcl{D}}[ |v \cdot z| \geq ts] e^t dt \leq 1 + \int_{t = 0}^{\infty} 6e^{(-1/\norm{v} + 1)t} = 1 + 6 \frac{\norm{v}}{1 - \norm{v}}  \leq \frac{6}{1 - \norm{v}}
%\]
%as desired.
%\end{proof}

\subsection{Basic Observations}\label{sec:reductions}

%The reason for this scaling is so that the distribution $\mcl{D}' = \mcl{D} - \mcl{D}$ given by the difference of two independent samples from $\mcl{D}$ is $1$-Poincare. 
Before we begin with the main proofs of Theorem \ref{thm:main-Poincare} and Theorem \ref{thm:main-GMM}, it will be useful to make a few simple reductions that allow us to make the following simplifying assumptions:
\begin{itemize}
    \item \textbf{Means Polynomially Bounded:} For all $i$, we have $\norm{\mu_i - \mu_j} \leq O((k/w_{\min} )^2) $, and
    \item \textbf{Dimension Not Too High:} We have $d \leq k$.
\end{itemize}
Since reducing to the case when the above assumptions hold is straight-forward, we defer the details to Appendix \ref{appendix:reductions}.  The reductions work in both settings  (general Poincare distributions and Gaussians) so in all future sections, we will be able to work assuming that the above properties hold.

\section{Moment Estimation}\label{sec:moments1}

We will now work towards proving our result for Poincare distributions.  A key ingredient in our algorithm will be estimating the moment tensor of a mixture $\mcl{M}$ i.e. for an unknown mixture
\[
\mcl{M} = w_1\mcl{D}(\mu_1) + \dots + w_k \mcl{D}(\mu_k)
\]
we would like to estimate the tensor
\[
w_1 \mu_1^{\otimes t} + \dots + w_k \mu_k^{\otimes t}
\]
for various values of $t$ using samples from $\mcl{M}$.  Naturally, it suffices to consider the case where we are given samples from $\mcl{D}(\mu)$ for some unknown $\mu$ and our goal is to estimate the tensor $\mu^{\otimes t}$.   For our full algorithm, we will need to go up to $t \sim \log k/ \log \log k$ but of course for such $t$, our estimate has to be implicit because we cannot write down the full tensor in polynomial time.  In this section, we address this task by constructing an unbiased estimator with bounded variance that can be easily manipulated implicitly.
\\\\
We make the following definition to simplify notation later on. 
\begin{definition}
For integers $t$ and a distribution $\mcl{D}$, we define the tensor 
\[
D_{t, \mcl{D}} = \E_{z \sim \mcl{D}}[ z^{\otimes t}] \,.
\]
We will drop the subscript $\mcl{D}$ when it is clear from context.
\end{definition}

\subsection{Adjusted Polynomials}

First, we just construct an unbiased estimator for $\mu^{\otimes t}$ (without worrying about making it implicit). This estimator is given in the definition below.

\begin{definition}\label{def:adjusted-polynomial}
Let $\mcl{D}$ be a distribution on $\R^d$.  For $x \in \R^d$, define the polynomials $P_{t, \mcl{D}}(x)$ for positive integers $t$ as follows. $P_{0, \mcl{D}}(x) = 1$ and for $t \geq 1$,
\begin{equation}\label{eq:recursive-def}
P_{t, \mcl{D}}(x) = x^{\otimes t} - \sum_{sym} D_{1, \mcl{D}} \otimes P_{t-1, \mcl{D}}(x) - \sum_{sym} D_{2, \mcl{D}} \otimes P_{t-2, \mcl{D}}(x) - \dots - D_{t, \mcl{D}} \,.
\end{equation}
We call $P_{t, \mcl{D}}$ the $\mcl{D}$-adjusted polynomials and will sometimes drop the subscript $\mcl{D}$ when it is clear from context.
\end{definition}

We now prove that the $\mcl{D}$-adjusted polynomials give an unbiased estimator for $\mu^{\otimes t}$ when given samples from $\mcl{D}(\mu)$.

\begin{claim}\label{claim:adjusted-polynomial}
For any $\mu \in \R^d$,
\[
\E_{z \sim \mcl{D}(\mu)} [P_{t, \mcl{D}}(z)] = \mu^{\otimes t} \,.
\]
\end{claim}
\begin{proof}
To simplify notation, we will drop all of the $\mcl{D}$ from the subscripts as there will be no ambiguity.  We prove the claim by induction on $t$.  The base case for $t = 1$ is clear.  Now for the inductive step, note that 
\begin{align*}
\E_{z \sim \mcl{D}(\mu)} [P_t(z)]    &= \E_{z \sim \mcl{D}(\mu)} \left[ z^{\otimes t} - \sum_{sym} D_1 \otimes P_{t-1}(z) - \sum_{sym} D_2 \otimes P_{t-2}(z) - \dots - D_t \right] \\ &= \E_{x \sim \mcl{D}}[ (x + \mu)^{\otimes t}] -  \E_{z \sim \mcl{D}(\mu)} \left[ \sum_{sym} D_1 \otimes P_{t-1}(z) + \sum_{sym} D_2 \otimes P_{t-2}(z) + \dots + D_t \right] \\ & = \E_{x \sim \mcl{D}} \left[ \mu^{\otimes t} - \sum_{sym} (D_1 - x^{\otimes 1}) \otimes \mu^{t - 1} - \sum_{sym} (D_2 - x^{\otimes 2}) \otimes \mu^{t - 2} - \dots -  (D_t - x^{\otimes t}) \right] \\ & = \mu^{\otimes t} \,.
\end{align*}
where we used the induction hypothesis and then the definition of $D_t$ in the last two steps.
\end{proof}

\subsection{Variance Bounds for Poincare Distributions}

In the previous section, we showed that the $\mcl{D}$-adjusted polynomials give an unbiased estimator for $\mu^{\otimes t}$.  We now show that they also have bounded variance when $\mcl{D}$ is Poincare.  This will rely on the following claim which shows that the $\mcl{D}$-adjusted polynomials recurse under differentiation. 

\begin{claim}\label{claim:derivative-recurrence}
Let $\mcl{D}$ be a distribution on $\R^d$.  Then
\[
\frac{\partial P_{t, \mcl{D}}(x)}{\partial x_i} = \sum_{\sym } e_i \otimes  P_{t - 1, \mcl{D}}(x)
\]
where we imagine $x = (x_1, \dots , x_d)$ so $x_i$ is the $i$\ts{th} coordinate of $x$ and 
\[
e_i = ( \underbrace{0, \dots , 1}_i , \dots, 0 )
\]
denotes the $i$\ts{th} coordinate basis vector.
\end{claim}
\begin{proof}
We will prove this by induction on $t$.  The base case for $t = 1$ is clear.  In the proceeding computations, we drop the $\mcl{D}$ from all subscripts as there will be no ambiguity.  Differentiating the definition of $P_{t, \mcl{D}}$ and using the induction hypothesis, we get
\begin{align*}
\frac{\partial P_{t}(x)}{\partial x_i} &= \sum_{sym} e_i \otimes x^{\otimes t - 1} - \sum_{sym} D_{1} \otimes e_i \otimes P_{t-2}(x) - \dots - \sum_{sym} D_{t - 2} \otimes e_i \otimes P_1(x) -  \sum_{sym} D_{t - 1} \otimes e_i  \\ &= \sum_{sym} e_i \otimes \left( x^{\otimes t-1} - \sum_{sym} D_1 \otimes P_{t-2}(x) - \dots - D_{t-1} \right) \\ & = \sum_{sym} e_i \otimes P_{t-1}(x)
\end{align*}
where in the last step we again used (\ref{eq:recursive-def}), the recursive definition of $P_{t-1}$.  This completes the proof.
\end{proof}

Now, by using the Poincare property, we can prove a bound on the variance of the estimator $P_{t, \mcl{D}}(x)$.

\begin{claim}\label{claim:variance-bound}
Let $\mcl{D}$ be a distribution on $\R^d$ that is $1$-Poincare.  Let $v \in \R^{d^t}$ be a vector.  Then
\[
\E_{z \sim \mcl{D}(\mu)} [ \left( v \cdot \flatten( P_{t, \mcl{D}}(z)) \right)^2 ] \leq ( \norm{\mu}^2 + t^2)^t \norm{v}^2 \,. 
\]
\end{claim}
\begin{proof}
We will prove the claim by induction on $t$.  The base case for $t = 1$ follows because 
\begin{align*}
\E_{z \sim \mcl{D}(\mu)} [ \left( v \cdot \flatten( P_{1, \mcl{D}}(z)) \right)^2 ] = (v \cdot \mu)^2 + \Var_{z \sim \mcl{D}(\mu)} \left( v \cdot \flatten( P_{1, \mcl{D}}(z)) \right) \leq (v \cdot \mu)^2 + \norm{v}^2 \\ \leq ( \norm{\mu}^2 + 1) \norm{v}^2
\end{align*}
where we used the fact that $\mcl{D}$ is $1$-Poincare.  Now for the inductive step, we have
\begin{align*}
\E_{z \sim \mcl{D}(\mu)} [ \left( v \cdot \flatten( P_{t, \mcl{D}}(z)) \right)^2 ] &= (v \cdot \flatten( \mu^{\otimes t}) )^2 + \Var_{z \sim \mcl{D}(\mu)} \left( v \cdot \flatten( P_{t, \mcl{D}}(z)) \right) \\ & \leq \norm{\mu}^{2t}\norm{v}^2 + \E_{z \sim \mcl{D}(\mu)} \left[ \sum_{i = 1}^d \left(v \cdot \frac{\partial P_{t, \mcl{D}}(x)}{\partial x_i} \right)^2 \right]  \\ & = \norm{\mu}^{2t}\norm{v}^2 + \E_{z \sim \mcl{D}(\mu)} \left[ \sum_{i = 1}^d \left( \sum_{j = 1}^t v_{\eta_j = i} \cdot P_{t-1, \mcl{D}}(z)\right)^2 \right]  \\ & \leq \norm{\mu}^{2t}\norm{v}^2 + \E_{z \sim \mcl{D}(\mu)} \left[ t  \sum_{j = 1}^t \sum_{i = 1}^d \left( v_{\eta_j = i} \cdot P_{t-1, \mcl{D}}(z)\right)^2 \right] \\ & \leq \norm{\mu}^{2t} \norm{v}^2 + t \sum_{j = 1}^t \sum_{i = 1}^d ( \norm{\mu}^2 + (t-1)^2)^{t-1} \norm{v_{\eta_j = i}}^2 \\ & = \norm{\mu}^{2t} \norm{v}^2+ t^2 ( \norm{\mu}^2 + (t-1)^2)^{t-1} \norm{v}^2 \\ & \leq ( \norm{\mu}^2 + t^2)^t \norm{v}^2 
\end{align*}
where in the above manipulations, we first used Claim \ref{claim:adjusted-polynomial}, then the fact that $\mcl{D}$ is $1$-Poincare,  then Claim \ref{claim:derivative-recurrence}, then Cauchy Schwarz, then the inductive hypothesis, and finally some direct manipulation.  This completes the inductive step and we are done.
\end{proof}

\subsection{Efficient Implicit Representation} \label{sec:implicit-moments1}

In the previous section, we showed that for $x \sim \mcl{D}(\mu)$ for unknown $\mu$, $P_{t, \mcl{D}}(x)$ gives us an unbiased estimator of $\mu^{\otimes t}$ with bounded variance.  Still, it is not feasible to actually compute $P_{t, \mcl{D}}(x)$ in polynomial time because we cannot write down all of its entries and there is no nice way to implicitly work with terms such as $D_{t, \mcl{D}}$ that appear in $P_{t, \mcl{D}}(x)$.  In this section, we construct a modified estimator that is closely related to $P_{t, \mcl{D}}(x)$ but is also easy to work with implicitly because all of the terms will be rank-$1$ i.e. of the form $v_1 \otimes \dots \otimes v_t$ for some vectors $v_1, \dots , v_t \in \R^d$.  Throughout this section, we will assume that the distribution $\mcl{D}$ that we are working with is fixed and we will drop it from all subscripts as there will be no ambiguity.  

Roughly, the way that we construct this modified estimator is that we take multiple variables $x_1, \dots , x_t \in \R^d$.  We start with $P_t(x_1)$.  We then add various products 
\[
P_{a_1}(x_1) \otimes \dots \otimes P_{a_t}(x_t)
\]
to it in a way that when expanded as monomials, only the leading terms, which are rank-$1$ since they are a direct product of the form $x_1^{\otimes a_1} \otimes \dots \otimes x_t^{\otimes a_t}$, remain.  If we then take $x_1 \sim \mcl{D}(\mu)$ and $x_2, \dots, x_t \sim \mcl{D}$, then Claim \ref{claim:adjusted-polynomial} will immediately imply that the expectation is $\mu^{\otimes t}$.  The key properties are stated formally in  Corollary \ref{coro:rank1-identity-p2}, Corollary \ref{coro:rank1-estimator-mean} and Corollary \ref{coro:rank1-estimator-variance}.  
\\\\\
First, we write out an explicit formula for $P_t(x)$.

\begin{claim}\label{claim:explicit-formula}
We have
\[
P_{t}(x) = \sum_{ S_0 \subseteq [t] }  \left( x^{\otimes S_0 } \right) \otimes  \left( \sum_{\{S_1, \dots , S_t \} \in Z_t( [t] \backslash S_0)} (-1)^{\mcl{C} \{S_1, \dots , S_t \} } (\mcl{C} \{S_1, \dots , S_t \} )! (D_{|S_1|})^{(S_1)} \otimes \dots \otimes (D_{|S_t|})^{(S_t)}   \right)
\]
\end{claim}
\begin{proof}
We use induction.  The base case is trivial.  Now, it suffices to compute the coefficient of some ``monomial" in $P_{t}$.  Note that the coefficient of the monomial $x^{\otimes t}$ is clearly $1$ which matches the desired formula.  Otherwise, consider a monomial
\[
A = \left( x^{\otimes S_0} \right)\otimes (D_{|S_1|})^{(S_1)} \otimes \dots \otimes (D_{|S_a|})^{(S_a)} 
\]
where $S_1, \dots , S_a$ are nonempty for some $1 \leq a \leq t$.  Recall the recursive definition of $P_t$ in (\ref{eq:recursive-def}).  There are exactly $a$ terms on the RHS of (\ref{eq:recursive-def})  that can produce the monomial $A$.  These terms are
\[
-D_{|S_1|}^{(S_1)} \otimes P_{t - |S_1|}(x)^{([t]\backslash S_1)}, \dots , -D_{|S_a|}^{(S_a)} \otimes P_{t - |S_a|}(x)^{([t]\backslash S_a)} \,.
\]
However, by the inductive hypothesis, the coefficient of $A$ produced by each of these terms is exactly $(-1)^a ( a-1)!$.  Thus, combining over all $a$ terms, the resulting coefficient is $(-1)^a a!$ which matches the desired formula.  This completes the proof.
\end{proof}

Now, we are ready to write out our estimator that can be written as a sum of few rank-$1$ terms.  The key identity is below.

\begin{definition}
For $x_1, \dots , x_t \in \R^d$, define the polynomial 
\begin{equation}\label{eq:rank1}
Q_t(x_1, \dots , x_t) = \sum_{\substack{S_1 \cup \dots \cup S_t = [t] \\ |S_i \cap S_j| = 0}} \frac{(-1)^{\mcl{C} \{ S_1, \dots , S_t \}}}{\binom{t - 1}{\mcl{C} \{ S_1, \dots , S_t \} - 1 } } \left( P_{|S_1|}(x_1) \right)^{(S_1)} \otimes \dots \otimes \left( P_{|S_t|}(x_t) \right)^{(S_t)} \,.
\end{equation}
\end{definition}

\begin{lemma}\label{lem:rank1-identity-p1}
All nonzero monomials in $Q_t$ either have total degree $t$ in the variables $x_1, \dots , x_t$ or are constant.
\end{lemma}
\begin{proof}
Consider substituting the formula in Claim \ref{claim:explicit-formula} into the RHS for all occurrences of $P$.  Consider a monomial
\[
A = \left( x_{i_1}^{\otimes U_1} \right) \otimes \dots \otimes  \left( x_{i_a}^{\otimes U_a} \right) \otimes  (D_{|V_1|})^{(V_1)} \otimes \dots \otimes (D_{|V_b|})^{(V_b)} 
\]
where $U_1, \dots , U_a, V_1, \dots , V_b$ are nonempty.  Note that when we expand out the RHS, all monomials are of this form.  It now suffices to compute the coefficient of this monomial $A$ in the expansion of the RHS.  
\\\\
The terms in the sum on the RHS are of the form 
\[
\frac{(-1)^c}{\binom{t-1}{c-1}} P_{|S_{j_1}|}( x_{j_1})^{(S_{j_1})} \otimes \dots \otimes P_{|S_{j_c}|}( x_{j_c})^{(S_{j_c})}
\]
where $S_{j_1}, \dots , S_{j_c}$ are nonempty.  We now consider summing over all such terms that can produce the monomial $A$ and sum the corresponding coefficient to get the overall coefficient of $A$.
\\\\
In order for the monomial $A$ to appear in the expansion of this term, we need $\{i_1, \dots , i_a \} \subseteq \{ j_1, \dots , j_c \}$.  Once the indices $j_1, \dots , j_c$ are fixed, it remains to assign each of the terms 
\[
(D_{|V_1|})^{(V_1)} , \dots ,  (D_{|V_b|})^{(V_b)} 
\]
to one of the variables $x_{j_1}, \dots , x_{j_c}$.  This will correspond to which of the polynomials 
\[
P_{|S_{j_1}|}( x_{j_1})^{(S_{j_1})}, \dots ,  P_{|S_{j_c}|}( x_{j_c})^{(S_{j_c})}
\]
that the term came from.  Specifically, for each integer $f $ with $1 \leq f \leq c$, let $B_f \subset [b]$ be the indices of the terms that are assigned to the variable $x_{j_f}$.  The sets $B_1, \dots , B_c$ uniquely determine the sets $S_{j_1}, \dots , S_{j_c}$ but also need to satisfy the constraint that if $j_f \notin \{i_1, \dots , i_a \}$ then $B_f \neq \emptyset $.  Once these sets are all fixed, by Claim \ref{claim:explicit-formula}, the desired coefficient is simply
\[
\frac{(-1)^c}{\binom{t-1}{c-1}} (-1)^{|B_1| + \dots  +|B_c|} |B_1|! \cdots |B_c|! = \frac{(-1)^c}{\binom{t-1}{c-1}} (-1)^{b} |B_1|! \cdots |B_c|! \,.
\]
Now overall, the desired coefficient is
\begin{equation}\label{eq:binomsum}
\sum_{c = a}^{a + b} \frac{(-1)^c}{\binom{t-1}{c-1}} \binom{t - a}{c - a}  (-1)^{b} \sum_{\substack{B_1 \cup \dots \cup B_c = [b]  \\ B_i \cap B_j = \emptyset \\ B_{a+1}, \dots , B_c \neq \emptyset} } |B_1|! \cdots |B_c|! \,.
\end{equation}
This is because there are $\binom{t - a}{c - a}$ to choose the set $\{ j_1, \dots , j_c \}$ (since it must contain $ \{ i_1, \dots , i_a \}$ ).  Once we have chosen this set, WLOG we can label $j_1 = i_1, \dots, j_a = i_a$ so that $j_{a+1}, \dots , j_c$ are the elements that are not contained in $ \{ i_1, \dots , i_a \}$ and thus the constraint on $B_1, \dots , B_c$ is simply that $B_{a+1}, \dots , B_c$ are nonempty.
\\\\
To evaluate (\ref{eq:binomsum}), we will evaluate the inner sum differently.  Imagine first choosing the sizes $s_1 = |B_1|, \dots , s_c = |B_c|$ and then choosing the sets $B_1, \dots , B_c$ to satisfy these size constraints.  The inner sum can then be rewritten as
\[
\sum_{ \substack{s_1 + \dots + s_c = b \\ s_{a+1}, \dots , s_c > 0}}  \binom{b}{s_1, \dots , s_c }  s_1! \cdots s_c ! = b! \sum_{ \substack{s_1 + \dots + s_c = b \\ s_{a+1}, \dots , s_c > 0}} 1  = b!  \binom{b + a - 1}{c - 1} 
\]
where the last equality follows from counting using stars and bars.  Now we can plug back into (\ref{eq:binomsum}).  Assuming that $a,b \geq 1$, the coefficient of the monomial $A$ is 
\begin{align*}
\sum_{c = a}^{a + b} \frac{(-1)^c}{\binom{t-1}{c-1}} \binom{t - a}{c - a}  (-1)^{b} b!  \binom{b + a - 1}{c - 1} &= (-1)^b b! \sum_{c = a}^{a + b} \frac{(-1)^c(t-a)! (b+a - 1)!}{(c-a)! (t-1)! (b+a - c)!} \\ &= (-1)^b   \frac{(t-a)! (b+a - 1)!}{ (t-1)! }  \sum_{c = a}^{a + b} (-1)^c \binom{b}{c - a} \\ & = 0 \,.
\end{align*}
Thus, the only monomials that have nonzero coefficient either have $a = 0$ (meaning they are constant) or $b = 0$ (meaning they have degree $t$).  This completes the proof.
\end{proof}

Now, we can easily eliminate the constant term by subtracting off $Q(x_{t+1}, \dots , x_{2t})$ for some additional variables $x_{t+1}, \dots , x_{2t}$ and we will be left with only degree-$t$ terms.  It will be immediate that the degree-$t$ terms are all rank-$1$ and this will give us an estimator that can be efficiently manipulated implicitly.

\begin{definition}\label{def:key-polynomial}
For $x_1, \dots , x_{2t} \in \R^d$, define the polynomial
\[
R_t(x_1, \dots , x_{2t}) = - Q_t(x_1, \dots , x_t) + Q_t(x_{t+1}, \dots , x_{2t}) \,.
\]
\end{definition}

\begin{corollary}\label{coro:rank1-identity-p2}
We have the identity
\[
R_t(x_1, \dots , x_{2t}) = \sum_{\substack{S_1 \cup \dots \cup S_t = [t] \\ |S_i \cap S_j| = 0}}   \frac{(-1)^{\mcl{C} \{ S_1, \dots , S_t \} -1 }}{\binom{t - 1}{\mcl{C} \{ S_1, \dots , S_t \} - 1 } } \left( x_1^{\otimes S_1} \otimes \dots \otimes x_t^{\otimes S_t} -  x_{t+1}^{\otimes S_1} \otimes \dots \otimes x_{2t}^{\otimes S_t} \right)\,.
\]
\end{corollary}
\begin{proof}
This follows immediately from Lemma \ref{lem:rank1-identity-p1} and the definition of $R_t(x_1, \dots , x_{2t})$ because the constant terms cancel out and the degree-$t$ terms clearly match the RHS of the desired expression.
\end{proof}

Corollary \ref{coro:rank1-identity-p2} gives us a convenient representation for working implicitly with $R_t(x_1, \dots , x_{2t})$.  We now show why this polynomial is actually useful.  In particular, we show that for $x_1 \sim \mcl{D}(\mu)$ and $x_2, \dots , x_{2t} \sim \mcl{D}$, $R_{t}(x_1, \dots , x_{2t})$ is an unbiased estimator of $\mu^{\otimes t}$ and furthermore that its variance is bounded.  These properties will follow directly from the definitions of $R_t, Q_t$ combined with Claim \ref{claim:adjusted-polynomial} and Claim \ref{claim:variance-bound}.

\begin{corollary}\label{coro:rank1-estimator-mean}
We have
\[
\E_{z_1 \sim \mcl{D}(\mu), z_2, \dots , z_{2t} \sim \mcl{D}}[  R_t(z_1, \dots , z_{2t}) ] = \mu^{\otimes t} 
\]
and for fixed $z_1$, we have
\[
\E_{ z_2, \dots , z_{2t} \sim \mcl{D}}[  R_t(z_1, \dots , z_{2t}) ] = P_t(z_1) \,.
\]
\end{corollary}
\begin{proof}
Using the definition of $Q_t$ in (\ref{eq:rank1}) and Claim \ref{claim:adjusted-polynomial}, the expectations of all of the terms are $0$ except for the leading term $P_t(z_1)$.  Thus,
\[
\E_{ z_2, \dots , z_{2t} \sim \mcl{D}}[  R_t(z_1, \dots , z_{2t}) ] = P_t(z_1) \,.
\]
Also by Claim \ref{claim:adjusted-polynomial}, 
\[
\E_{z_1 \sim \mcl{D}(\mu)}[ P_{t}(z_1)] = \mu^{\otimes t}
\]
and this gives us the two desired identities.
\end{proof}

\begin{corollary}\label{coro:rank1-estimator-variance}
Let $\mcl{D}$ be a distribution that is $1$-Poincare.  We have
\[
\E_{z_1 \sim \mcl{D}(\mu), z_2, \dots , z_{2t} \sim \mcl{D}}\left[ \flatten(R_t(z_1, \dots , z_{2t}))^{\otimes 2} \right] \preceq (20t)^{2t} (\norm{\mu}^{2t} + 1)I_{d^t}  \,.
\]
where recall $I_{d^t}$ denotes the $d^t$-dimensional identity matrix.
\end{corollary}
\begin{proof}
Using the definition of $R_t$ and $Q_t$ and Cauchy Schwarz, we have
\begin{align*}
& \E_{z_1 \sim \mcl{D}(\mu), z_2, \dots , z_{2t} \sim \mcl{D}}\left[ \flatten(R_t(z_1, \dots , z_{2t}))^{\otimes 2} \right] \\ & \preceq  2\E_{z_1 \sim \mcl{D}(\mu), z_2, \dots , z_{t} \sim \mcl{D}} \left[ \flatten \left( Q_t(z_1, \dots , z_t)\right)^{\otimes 2} \right] +  2 \E_{z_{t+1}, \dots , z_{2t} \sim \mcl{D}} \left[ \flatten \left( Q_t(z_{t+1}, \dots , z_{2t})\right)^{\otimes 2} \right]  \,.
\end{align*}

Now by Cauchy Schwarz again,
\begin{align*}
 &\E_{z_1 \sim \mcl{D}(\mu), z_2, \dots , z_{t} \sim \mcl{D}} \left[ \flatten \left( Q_t(z_1, \dots , z_t)\right)^{\otimes 2} \right] \preceq \left( \sum_{\substack{S_1 \cup \dots \cup S_t = [t] \\ |S_i \cap S_j| = 0}}  |S_1|! \cdots |S_t|!\right) \\ & \quad \cdot  \left( \sum_{\substack{S_1 \cup \dots \cup S_t = [t] \\ |S_i \cap S_j| = 0}}  \frac{\E_{z_{1} \sim \mcl{D}(\mu), z_2  \dots , z_{t} \sim \mcl{D}} \left[ \flatten \left( \left( P_{|S_1|}(z_{1}) \right)^{(S_1)} \otimes \dots \otimes \left( P_{|S_t|}(z_{t}) \right)^{(S_t)} \right)^{\otimes 2} \right] }{|S_1|! \cdots |S_t|!} \right) \,.
 \end{align*}
 
Let the first term above be $C_1$ and the second term be $C_2$.  We can rearrange the sum over partitions of $[t]$ as follows.  We can first choose the sizes $s_1 = |S_1|, \dots , s_t = |S_t|$ and then choose the partition according to these constraints.  We get
\[
C_1 = \sum_{s_1 + \dots + s_t = t} \binom{t}{s_1, \dots , s_t} s_1 ! \cdots s_t ! = t! \binom{2t - 1}{t} \leq (2t)^t 
\]
and similarly (and also using Claim \ref{claim:variance-bound})
\begin{align*}
C_2 &\preceq \left(\sum_{\substack{S_1 \cup \dots \cup S_t = [t] \\ |S_i \cap S_j| = 0}} \frac{(\norm{\mu}^2 + |S_1|^2)^{|S_1|} |S_2|^{2|S_2|} \cdots |S_t|^{2|S_t|} }{|S_1|! \cdots |S_t|! } \right) I_{d^t} \\ &\preceq  \left(20^t \sum_{\substack{S_1 \cup \dots \cup S_t = [t] \\ |S_i \cap S_j| = 0}} \frac{\left( \norm{\mu}^{2|S_1|} + (|S_1|!)^2 \right) (|S_2|!)^2 \cdots (|S_t|!)^2 }{|S_1|! \cdots |S_t|! } \right) I_{d^t} \\ & \preceq \left( 20^t (\norm{\mu}^{2t} + 1) \sum_{s_1 + \dots + s_t = t} \binom{t}{s_1, \dots , s_t}  s_1 ! s_2! \cdots s_t!  \right) I_{d^t} \\ &\preceq (40t)^t( \norm{\mu}^{2t} + 1) I_{d^t} \,.
\end{align*}
Note that in the first step above, we also used the fact that $z_1, \dots , z_t$ are drawn independently.  Thus, overall we have shown
\[
\E_{z_1 \sim \mcl{D}(\mu), z_2, \dots , z_{t} \sim \mcl{D}} \left[ \flatten \left( Q_t(z_1, \dots , z_t)\right)^{\otimes 2} \right] \preceq (10t)^{2t} (\norm{\mu}^{2t} + 1) I_{d^t} \,.
\]
Similarly, we have
\[
\E_{z_{t+1} , \dots , z_{2t} \sim \mcl{D}} \left[ \flatten \left( Q_t(z_1, \dots , z_t)\right)^{\otimes 2} \right] \preceq (10t)^{2t} I_{d^t}
\]
and putting everything together, we conclude
\[
\E_{z_1 \sim \mcl{D}(\mu), z_2, \dots , z_{2t} \sim \mcl{D}}\left[ \flatten(R_t(z_1, \dots , z_{2t}))^{\otimes 2} \right] \preceq  (20t)^{2t} (\norm{\mu}^{2t} + 1) I_{d^t}
\]
and we are done.
\end{proof}

\section{Iterative Projection}\label{sec:projection}

In this section, we explain our technique for implicitly working with tensors that have too many entries to write down.  Recall that we would like to estimate the moment tensor
\[
w_1 \mu_1^{\otimes t} + \dots + w_k \mu_k^{\otimes t}
\]
for 
\[
t \sim \frac{\log (k/w_{\min})}{ \log \log(k/w_{\min})} \,.
\] 
However doing this directly requires quasipolynomial time  (because there are quasipolynomially many entries).  Roughly, the way we get around this issue is by, iteratively for each $t$, computing a $k$-dimensional subspace that contains the span of $\mu_1^{\otimes t}, \dots, \mu_k^{\otimes t}$.  We then only need to compute the projection of $w_1 \mu_1^{\otimes t} + \dots + w_k \mu_k^{\otimes t}$ onto this subspace.  Of course, the subspace and projection need to be computed implicitly because we cannot explicitly write out these expressions in polynomial time.

\subsection{Nested Projection Maps}

At a high level, to implicitly estimate the span of $\mu_1^{\otimes t}, \dots, \mu_k^{\otimes t}$, we will first estimate the span of $\mu_1^{\otimes t - 1}, \dots, \mu_k^{\otimes t -1}$ and then bootstrap this estimate to estimate the span of $\mu_1^{\otimes t}, \dots, \mu_k^{\otimes t}$.  Since we cannot actually write down the span even though it is $k$-dimensional (because the vectors have super-polynomial length), we will store the span implicitly through a sequence of projections.  We explain the details below.

\begin{definition}[Nested Projection]\label{def:nested-projection}
Let $c_0 = 1$ and $c_1, \dots , c_t$ be positive integers.  Let $\Pi_1 \in \R^{c_1 \times dc_0}, \Pi_{2} \in \R^{c_2 \times dc_1}, \dots , \Pi_t \in \R^{c_t \times dc_{t-1}}$ be matrices.  Define the $c_t \times d^t$ nested projection matrix
\[
\Gamma_{\Pi_t, \dots , \Pi_1} =  \Pi_t \left(I_d \otimes_{\kr} \left(\Pi_{t-1} \left(I_d \otimes_{\kr} \cdots \right)\right) \right) \,.
\]
\end{definition}

It is not hard to verify (see below) that when $\Pi_1, \dots , \Pi_t$ are projection matrices then $\Gamma_{\Pi_t, \dots , \Pi_1}$ is as well.
\begin{claim}\label{claim:orthogonal-rows}
Let $c_0 = 1$ and $c_1, \dots , c_t$ be positive integers.  Let $\Pi_1 \in \R^{c_1 \times dc_0}, \Pi_{2} \in \R^{c_2 \times dc_1}, \dots , \Pi_t \in \R^{c_t \times dc_{t-1}}$ be matrices whose rows are orthonormal.  Then $\Gamma_{\Pi_t, \dots , \Pi_1}$ has orthonormal rows.
\end{claim}
\begin{proof}
We prove the claim by induction on $t$.  The base case is clear.  Next, by the induction hypothesis, the matrix 
\[
\Gamma_{\Pi_{t-1}, \dots , \Pi_1} =  \Pi_{t-1} \left(I_d \otimes_{\kr} \left(\Pi_{t-2} \left(I_d \otimes_{\kr} \cdots \right)\right) \right)
\]
has orthonormal rows.  Thus, the matrix 
\[
 \Pi_t\left(I_d \otimes_{\kr}\Gamma_{\Pi_{t-1}, \dots , \Pi_1} \right)
\]
has orthonormal rows as well, completing the induction.
\end{proof}

Note that in our paper, $\Pi_1, \dots , \Pi_t$ will always have orthonormal rows so $\Gamma_{\Pi_t, \dots , \Pi_1}$ always does as well.  This fact will often be used without explicitly stating it.  The key point about the construction of $\Gamma_{\Pi_t, \dots , \Pi_1}$ is that instead of storing a full $c_t \times d^t$-sized matrix, it suffices to store the individual matrices $\Pi_1, \dots , \Pi_t$ which are all polynomially sized.  The next important observation is that for certain vectors $v \in \R^{d^t}$ that are ``rank-$1$" i.e. those that can be written in the form
\[
v = \flatten(v_t \otimes \dots \otimes v_1) \,,
\]
the expression $\Gamma_{\Pi_t, \dots , \Pi_1} v$ can be computed efficiently.  This is shown in the following claim.

\begin{claim}\label{claim:efficient-projection}
Let $c_0 = 1$ and $c_1, \dots , c_t$ be positive integers.  Let $\Pi_1 \in \R^{c_1 \times dc_0}, \Pi_{2} \in \R^{c_2 \times dc_1}, \dots , \Pi_t \in \R^{c_t \times dc_{t-1}}$ be matrices.  Let $v \in \R^{d^t}$ satisfy $v = \flatten(v_1 \otimes \dots \otimes v_t)$ for some $v_1, \dots , v_t \in \R^d$.  Then in $\poly(d,t, \max(c_i))$ time, we can compute $\Gamma_{\Pi_t, \dots , \Pi_1} v$.
\end{claim}
\begin{proof}
We will prove the claim by induction on $t$.  For each $t' = 1,2, \dots , t$, we compute
\[
\Gamma_{\Pi_{t'}, \dots , \Pi_1} \flatten(v_{t'} \otimes \dots \otimes v_1 ) \,.
\]
The base case of the induction is clear.  To do the induction step, note that
\[
\Gamma_{\Pi_{t' + 1}, \dots , \Pi_1}\flatten(v_{t' + 1} \otimes \cdots \otimes v_1) = \Pi_{t' + 1} \flatten\left(v_{t' + 1} \otimes \left(\Gamma_{\Pi_{t'}, \dots , \Pi_1} \textsf{flat}(v_{t'} \otimes \cdots \otimes v_1) \right) \right)  \,.
\]
It is clear that this computation can be done in $\poly(d,t, \max(c_i))$ time so iterating this operation completes the proof.
\end{proof}

As a trivial consequence of the above, we can also compute direct products of nested projections applied to a ``rank-$1$" vector $v$.
\begin{corollary}\label{coro:efficient-projection}
Let $c_0 = 1$ and $c_1, \dots , c_t$ be positive integers.  Let $\Pi_1 \in \R^{c_1 \times dc_0}, \Pi_{2} \in \R^{c_2 \times dc_1}, \dots , \Pi_t \in \R^{c_t \times dc_{t-1}}$ be matrices.  Let $v \in \R^{d^t}$ satisfy $v = \flatten(v_1 \otimes \dots \otimes v_t)$ for some $v_1, \dots , v_t \in \R^d$.  Then for any integers $1 < s_1 < s_2 < \dots < s_n < t$, in $\poly(d,t, \max(c_i))$ time, we can compute the expression
\[
\left(\Gamma_{\Pi_t, \dots , \Pi_{s_n + 1}} \otimes_{\kr} \dots  \otimes_{\kr} \Gamma_{\Pi_{s_1}, \dots , \Pi_{1}} \right)v \,.
\]
\end{corollary}

Throughout our paper, we will only compute nested projections of the above form so it will be easy to verify that all steps can be implemented in polynomial time.

\section{Implicitly Estimating the Moment Tensor}\label{sec:implicit-est}
In this section, we combine the iterative projection techniques from Section \ref{sec:projection} with the estimators from Section \ref{sec:moments1} to show how to implicitly estimate the moment tensor given sample access to the distribution $\mcl{D}$ and the mixture
\[
\mcl{M} = w_1 \mcl{D}(\mu_1) + \dots + w_k \mcl{D}(\mu_k) \,.
\]
By implicitly estimate, we mean that we will compute projection matrices $\Pi_t, \dots , \Pi_1$ such that the row-span of $\Gamma_{\Pi_t, \dots , \Pi_1}$ essentially contains all of the flattenings of $\mu_1^{\otimes t}, \dots , \mu_k^{\otimes t}$.  
\begin{remark}
Once we have these matrices, we will also be able to estimate expressions such as
\[
\Gamma_{\Pi_t, \dots , \Pi_1} \flatten\left( w_1 \mu_1^{\otimes t} + \dots + w_k \mu_k^{\otimes t}\right) \,.
\]
\end{remark}

It will be convenient to make the following definitions.
\begin{definition}\label{def:moment-tensor}
For a mixture $\mcl{M} = w_1 \mcl{D}(\mu_1) + \dots + w_k \mcl{D}(\mu_k)$, we use $T_{t, \mcl{M}}$ to denote the tensor $w_1 \mu_1^{\otimes t} + \dots + w_k \mu_k^{\otimes t}$.  We may drop the subscript $\mcl{M}$ and just write $T_t$ when it is clear from context.
\end{definition}
\begin{definition}\label{def:moment-matrices}
For a mixture $\mcl{M} = w_1 \mcl{D}(\mu_1) + \dots + w_k \mcl{D}(\mu_k)$, we define $M_{2s, \mcl{M}}$ to be the tensor $T_{2s, \mcl{M}}$ rearranged (in a canonical way) as a $d^s \times d^s$ square matrix.  Again, we may drop the subscript $\mcl{M}$ when it is clear from context.
\end{definition}

We define $\mu_{\max} = \max( 1, \norm{\mu_1}, \dots , \norm{\mu_k})$.  We do not assume that we know $\mu_{\max}$ in advance.  However, the reduction in Section \ref{sec:reductions} means that it suffices to consider when $\mu_{\max}$ is polynomially bounded.  Also we can assume  $d = k$ i.e. the dimension of the underlying space is equal to the number of components.  This is because we can use the reduction in Section \ref{sec:reductions} and if $d < k$, then we can simply add independent standard Gaussian entries in the remaining $k-d$ dimensions.
\\\\
We now describe our algorithm for implicitly estimating the moment tensor.  For the remainder of this section, we will only work with a fixed mixture $\mcl{M}$ so we will drop it from all subscripts e.g. in Definitions \ref{def:moment-tensor} and \ref{def:moment-matrices}.  At a high level, we will recursively compute a sequence of projection matrices $\Pi_1 \in \R^{k \times d} , \Pi_2, \dots , \Pi_s \in \R^{k \times dk}$.  Our goal will be to ensure that $\Gamma_{\Pi_s, \dots, \Pi_1}$ (which is a $k \times d^s$ matrix) essentially contains the flattenings of $\mu_1^{\otimes s}, \dots , \mu_k^{\otimes s}$ in its row span.  

To see how to do this, assume that we have computed $\Pi_{s-1}, \dots , \Pi_{1}$ so far.  By the inductive hypothesis, $\Gamma_{ \Pi_{s-1}, \dots , \Pi_1}$ tells us a $k$-dimensional subspace that essentially contains the flattenings of $\mu_1^{\otimes s - 1}, \dots , \mu_k^{\otimes s - 1}$.  Thus, we trivially have a $dk$ dimensional subspace, given by the rows of $\left(I_d \otimes_\kr \Gamma_{ \Pi_{s-1}, \dots , \Pi_1} \right)$ that must essentially contain all of the flattenings of $\mu_1^{\otimes s}, \dots , \mu_k^{\otimes s}$.  It remains to reduce from this $dk$-dimensional space back to a $k$-dimensional space.  However, we can now write everything out in this $dk$-dimensional space and simply run PCA and take the top-$k$ singular subspace.  Formally, we estimate the $dk \times dk$ matrix 
\[
A_{2s} = \left(I_d \otimes_\kr \Gamma_{ \Pi_{s-1}, \dots , \Pi_1} \right) M_{2s}\left(I_d \otimes_\kr \Gamma_{ \Pi_{s-1}, \dots , \Pi_1} \right)^T = \sum_{i = 1}^k w_i  \left(\left(I_d \otimes_\kr \Gamma_{ \Pi_{s-1}, \dots , \Pi_1} \right) \flatten(\mu_i)^{\otimes s} \right)^{\otimes 2} 
\]
using techniques from Section \ref{sec:moments1} and then simply set $\Pi_s$ to have rows given by the top $k$ singular vectors of $A_{2s}$.  To gain some intuition for why this works, imagine that the subspace spanned by the rows of $\Gamma_{ \Pi_{s-1}, \dots , \Pi_1}$ exactly contains  the flattenings of $\mu_1^{\otimes s - 1}, \dots , \mu_k^{\otimes s - 1}$.  Also assume that our estimate of $A_{2s}$ is exact.  Then $A_{2s}$ has rank at most $k$ and the top-$k$ singular subspace must contain all of the vectors 
\[
\left(I_d \otimes_\kr \Gamma_{ \Pi_{s-1}, \dots , \Pi_1} \right) \flatten(\mu_i)^{\otimes s}  = \flatten \left( \mu_i \otimes \Gamma_{ \Pi_{s-1}, \dots , \Pi_1} \flatten(\mu_i)^{\otimes s - 1} \right) \,.
\]
The above then immediately implies that the subspace spanned by the rows of $\Gamma_{ \Pi_{s}, \dots , \Pi_1}$ exactly contains  the flattenings of $\mu_1^{\otimes s}, \dots , \mu_k^{\otimes s}$.  Of course, the actual analysis will need to be much more precise quantitatively in tracking the errors in each step.  

Our algorithm is described in full below.  The main algorithm, Algorithm \ref{alg:iterative-projection}, computes the projection matrices $\Pi_1, \dots \Pi_s$ following the outline above.  As a subroutine, it needs to estimate the matrix $A_{2s}$.  This is done in Algorithm \ref{alg:est-moment-tensor} which relies on the results in Section \ref{sec:moments1}, namely Corollary \ref{coro:rank1-estimator-mean} and Corollary \ref{coro:rank1-estimator-variance}.
%We also define $\mu_{\min} = \min_{\mu_i \neq 0} \norm{\mu_i}$.  The analysis will be significantly simpler with the assumption that $\mu_{\min} \geq 1$.  

%The assumption that $\mu_{\min} \geq 1$ will hold for all mixtures for which we estimate the moment tensor.  The reason is because we will actually take pairwise differences of the samples from our original mixture, which we assume has separated means.  This maps from a mixture of $k$ translared copies of some distribution $\mcl{D}$ to a mixture of $\sim k^2$ translated copies of $\mcl{D} - \mcl{D}$ but now with the condition that all of the means are either $0$ (corresponding to a difference between two samples from the same component) or separated from $0$ (corresponding to a difference between two samples from different components).  

%Let $t^*$ be the minimum positive integer such that $\mu_{\max}^{t^*} \geq (k/w_{\min})^{10}$.  We do not need to know $t^*$ in advance.  

\begin{algorithm}[H]
\caption{{\sc Iterative Projection Step} }
\begin{algorithmic}
\State \textbf{Input:} Samples $z_1, \dots , z_n$ from unknown mixture
\[
\mcl{M} = w_1 \mcl{D}(\mu_1) + \dots + w_k \mcl{D}(\mu_k)
\]

\State \textbf{Input:}  integer $t > 0$ 
\State Split samples into $t$ sets $S_1, \dots , S_t$ of equal size
\State Let $\Pi_1 = I_d$ (recall $k = d$)
\For{s = 2, \dots , t }
\State Run {\sc Estimate Moment Tensor} using samples $S_s$ to get approximation $\wt{A_{2s}} \in \R^{dk \times dk}$ to
\[
A_{2s} = \left(I_d \otimes_\kr \Gamma_{ \Pi_{s-1}, \dots , \Pi_1} \right) M_{2s}\left(I_d \otimes_\kr \Gamma_{ \Pi_{s-1}, \dots , \Pi_1} \right)^T
\]

\State Let $\Pi_s \in \R^{k \times dk}$ have rows forming an orthonormal basis of the top $k$ singular subspace of $\wt{A_{2s}}$
\EndFor
\State \textbf{Output:}  $ (\Pi_t,\dots , \Pi_1)$
\end{algorithmic}
\label{alg:iterative-projection}
\end{algorithm}

\begin{algorithm}[H]
\caption{{\sc Estimate Moment Tensor} }
\begin{algorithmic}
\State \textbf{Input:} Samples $z_1, \dots , z_n$ from unknown mixture
\[
\mcl{M} = w_1 \mcl{D}(\mu_1) + \dots + w_k \mcl{D}(\mu_k)
\]
\State \textbf{Input:} Integer $s > 0$
\State \textbf{Input:} Matrices $\Pi_{s-1} \in \R^{k \times dk},  \dots, \Pi_2 \in \R^{k \times dk},  \Pi_{1} \in \R^{k \times d}$ 
\For{i = 1,2, \dots , n}
\State Independently draw samples $x_1, \dots , x_{4s - 1}$ from $\mcl{D}$
\State Compute the $(kd)^2$-dimensional vector (recall   Definition \ref{def:key-polynomial})
\[
X_i = ( (I_d \otimes_{\text{kr}} \Gamma_{\Pi_{s-1}, \dots , \Pi_1}) \otimes_{\text{kr}} (I_d \otimes_{\text{kr}} \Gamma_{\Pi_{s-1}, \dots , \Pi_1})) \flatten \left(R_{2s}(z_i , x_1, \dots , x_{4s - 1}) \right) \,.
\]
\State Let $K_i$ be the rearrangement of $X_i$ into a square $dk \times dk$-dimensional matrix
\EndFor
\State \textbf{Output:} $A = (K_1 + \dots + K_n)/n$
\end{algorithmic}
\label{alg:est-moment-tensor}
\end{algorithm}

\subsection{Efficient Implementation}

A naive implementation of Algorithm \ref{alg:iterative-projection} requires $d^t$ time, which is too large.  However, in this section, we show that we can implement all of the steps more efficiently using only $\poly(n dk, t^t  )$ time. 
\begin{remark}
We will later show that it suffices to consider
\[
t \sim  O\left( \frac{\log (k/w_{\min})}{\log \log (k/w_{\min})} \right)
\]
so this runtime is actually polynomial in all parameters that we need. 
\end{remark}

\begin{claim}
Algorithm \ref{alg:iterative-projection} can be implemented to run in $\poly(n ,d, k, t^t  )$ time.
\end{claim}
\begin{proof}
First we analyze the runtime of Algorithm \ref{alg:est-moment-tensor}.  Note that by Corollary \ref{coro:rank1-identity-p2}, we can write
\[
R_s(z_i, x_1, \dots , x_{4s - 1}) = V_1 + \dots + V_l
\]
where $V_1, \dots , V_l$ are all rank-$1$ tensors and $l = s^{O(s)}$.  Now by Corollary \ref{coro:efficient-projection}, this means that computing 
\begin{align*}
X_i = ( (I_d \otimes_{\text{kr}} \Gamma_{\Pi_{s-1}, \dots , \Pi_1}) \otimes_{\text{kr}} (I_d \otimes_{\text{kr}} \Gamma_{\Pi_{s-1}, \dots , \Pi_1}))\flatten \left(Q_{2s}(z_i, x_1, \dots , x_{4s - 1}) \right) \\ = ( (I_d \otimes_{\text{kr}} \Gamma_{\Pi_{s-1}, \dots , \Pi_1}) \otimes_{\text{kr}} (I_d \otimes_{\text{kr}} \Gamma_{\Pi_{s-1}, \dots , \Pi_1})) (\flatten(V_1) + \dots + \flatten(V_l))
\end{align*}
can be done in $\poly(d , k,s^s)$ time by expanding out the RHS and computing each term separately.  It is then immediate that all of Algorithm \ref{alg:est-moment-tensor} runs in $\poly(n, d,k,  s^s)$ time.  Now in Algorithm \ref{alg:iterative-projection}, the only additional steps involve computing an SVD of the matrices $\wt{A_{2s}}$ (which are polynomially sized) so we conclude that the entire algorithm runs in $\poly(n, d,k  ,t^t)$ time.
\end{proof}

\subsection{Accuracy Analysis}

Now, we analyze the correctness of Algorithm \ref{alg:iterative-projection} namely that the span of the rows of the matrix $\Gamma_{\Pi_{t}, \dots , \Pi_1}$ indeed essentially contain all of $\mu_1^{\otimes t}, \dots , \mu_k^{\otimes t}$.  To simplify notation, we make the following definition. 

\begin{definition}
For all $s$, we define the matrix
\[
A_{2s} = \left(I_d \otimes_\kr \Gamma_{ \Pi_{s-1}, \dots , \Pi_1} \right) M_{2s}\left(I_d \otimes_\kr \Gamma_{ \Pi_{s-1}, \dots , \Pi_1} \right)^T \,.
\]
\end{definition}
\begin{remark}
Note that in the execution of Algorithm \ref{alg:iterative-projection}, $\wt{A_{2s}}$ is intended to be an estimate of $A_{2s}$.
\end{remark}

The main result that we will prove is stated below.  Note that this lemma does not require any assumptions about minimum mixing weights or means in the mixture.  Instead, it simply says that the subspace spanned by the rows of $\Gamma_{\Pi_{s}, \dots , \Pi_1}$ essentially contains $\flatten(\mu_i^{\otimes s})$ for all components $\mcl{D}(\mu_i)$ with mean and mixing weight bounded away from $0$. 

\begin{lemma}\label{lem:projection-accuracy}
Let $\mcl{D}$ be a distribution on $\R^d$ that is $1$-Poincare and let $\mcl{M} = w_1\mcl{D}(\mu_1) + \dots + w_k \mcl{D}(\mu_k)$ be a mixture of translations of $\mcl{D}$.  Let $w^* , \eps > 0$ be parameters.  Assume that the number of samples satisfies
\[
n \geq \left(\frac{t^t \mu_{\max}^t k d}{w^* \eps }\right)^C
\]
for some sufficiently large universal constant $C$.  Then with probability at least $1 - n^{-0.2}$, in the execution of Algorithm \ref{alg:iterative-projection}, the following condition holds: for all $i \in [k]$ such that $\norm{\mu_i} \geq 1$ and $w_i \geq w^*$, we have
\[
\norm{\Gamma_{\Pi_{s}, \dots , \Pi_1} \flatten( \mu_i^{\otimes s}) } \geq ( 1 - s\eps) \norm{\mu_i}^s  
\]
for all $s = 1,2, \dots , t$.    
\end{lemma}
\begin{remark}
The parameter $\eps$ represents the desired accuracy and the parameter $w^*$ is a weight cutoff threshold where we guarantee to recover ``significant" components whose mixing weight is at least $w^*$.
\end{remark}

Roughly, the proof of Lemma \ref{lem:projection-accuracy} will involve following the outline at the beginning of this section but quantitatively tracking the errors more precisely.  Before we prove Lemma \ref{lem:projection-accuracy}, we first prove a preliminary claim that our estimation  error $\norm{A_{2s} - \wt{A_{2s}}}_F$ is small.

\begin{claim}\label{claim:estimation-accuracy}
Assume that for a fixed integer $s$, Algorithm \ref{alg:est-moment-tensor} is run with a number of samples
\[
n \geq \left(\frac{s^s \mu_{\max}^s kd }{w^* \eps }\right)^C
\]
for some sufficiently large universal constant $C$.  Then with probability at least $1 - n^{-0.4}$, its output $\wt{A_{2s}}$ satisfies
\[
\norm{ \wt{A_{2s}} - A_{2s}}_F \leq 0.5 w^* \eps^2\,.
\]
\end{claim}
\begin{proof}
Recall that we are trying to estimate $A_{2s}$ which is a $dk \times dk$ matrix.  It will suffice for us to obtain a concentration bound for our estimate of each entry and then union bound.  Recall that in Algorithm \ref{alg:est-moment-tensor}, we estimate $A_{2s}$ by averaging $K_1, \dots , K_n$.  By Corollary \ref{coro:rank1-estimator-mean}, we have that 
\begin{equation}\label{eq:unbiased}
\E[ K_i] = A_{2s}
\end{equation}
so our estimator is unbiased.  Next, observe that each entry of $K_i$, say $K_i[a,b]$ where $1 \leq a,b \leq dk$ can be written as 
\[
K_i[a,b] = v \cdot \flatten( R_{2s}(z_i, x_1, \dots , x_{4s - 1})) 
\]
where $v \in \R^{d^{2s}}$ is some unit vector (this is by Claim \ref{claim:orthogonal-rows}).  By Corollary \ref{coro:rank1-estimator-variance}, we have
\[
\E\left[ (K_i[a,b] - A_{2s}[a,b])^2 \right] \leq \E [ K_i[a,b]^2 ] \leq (20s)^{2s} ( \mu_{\max}^{2s} + 1) 
\]
where the first inequality above is true by (\ref{eq:unbiased}).  Since our final estimate is obtained by averaging over $n$ independent samples, we have
\[
\E\left[ (\wt{A_{2s}}[a,b] - A_{2s}[a,b])^2 \right] \leq \frac{(20s)^{2s} ( \mu_{\max}^{2s} + 1)}{n} \,.
\]
Thus, with probability at least $1 - n^{-0.5}$, we must have
\[
\left \lvert \wt{A_{2s}}[a,b] - A_{2s}[a,b] \right \rvert \leq \frac{(20s)^{2s} ( \mu_{\max}^{2s} + 1)}{\sqrt{n}} \leq \frac{0.5w^*\eps^2}{dk}
\]
where the last inequality holds as long as we choose $n$ sufficiently large.  Union bounding the above over all entries (there are only $(dk)^2$ entries to union bound over) and ensuring that $n$ is sufficiently large, we get the desired bound.
\end{proof}

Now we are ready to prove Lemma \ref{lem:projection-accuracy}.
\begin{proof}[Proof of Lemma \ref{lem:projection-accuracy}]
We will prove the claim by induction on $s$.  The base case for $s = 0$ is clear.  Now let $i \in [k]$ be such that $\norm{\mu_i} \geq 1$ and $w_i \geq w^*$.  Define the vector $v_{i,s} \in \R^{dk}$ as 
\[
v_{i,s} = \flatten \left( \mu_i \otimes \Gamma_{\Pi_{s-1}, \dots , \Pi_1} \flatten( \mu_i^{s-1}) \right) \,.
\]
Note that this allows us to rewrite the matrix $A_{2s}$ as 
\[
A_{2s} = w_1 (v_{1,s} \otimes v_{1,s}) + \dots + w_k (v_{k,s} \otimes v_{k,s}) \,.
\]
Let $u_{i,s}$ be the projection of $v_{i,s}$ onto the orthogonal complement of $\Pi_s$.  Note that 
\[
u_{i,s} \cdot v_{i,s} = \norm{u_{i,s}}^2 \,.
\]
Thus, we must have
\[
u_{i,s}^T A_{2s} u_{i,s} \geq w_i \norm{u_{i,s}}^4 \,.
\]
On the other hand, note that $A_{2s}$ has rank at most $k$.  Assuming that the hypothesis of Claim \ref{claim:estimation-accuracy} holds, the $k+1$st singular value of $\wt{A_{2s}}$ has size at most $w^* \eps^2$.  Thus,
\[
u_{i,s}^T \wt{A_{2s}} u_{i,s} \leq 0.5 w^* \eps^2 \norm{u_{i,s}}^2 \,.
\]
Finally, using the hypothesis of Claim \ref{claim:estimation-accuracy} again, we must have
\[
\left \lvert u_{i,s}^T (A_{2s} - \wt{A_{2s}}) u_{i,s} \right \rvert \leq 0.5 w^* \eps^2 \norm{u_{i,s}}^2 \,.
\]
Putting the previous three inequalities together,  we deduce that we must have
\[
w_i \norm{u_{i,s}}^4 \leq w^* \eps^2 \norm{u_{i,s}}^2
\]
which implies $\norm{u_{i,s}} \leq \eps$.  Also, the induction hypothesis implies that 
\[
\norm{\Gamma_{\Pi_{s-1}, \dots , \Pi_1} \flatten( \mu_i^{s-1})} \geq (1 - (s-1)\eps) \norm{\mu_i}^{s-1}
\]
and thus
\[
\norm{v_{i,s}} \geq (1 - (s-1)\eps) \norm{\mu_i}^{s} \,.
\]
Finally, note that 
\[
\norm{\Gamma_{\Pi_{s}, \dots , \Pi_1} \flatten( \mu_i^{s})} = \norm{v_{i,s} - u_{i,s}} \geq (1 - (s-1)\eps) \norm{\mu_i}^{s} - \eps \geq (1 - s\eps) \norm{\mu_i}^s \,.
\]
This completes the inductive step.  Finally, it remains to note that the overall failure probability can be bounded by union bounding over all applications of Claim \ref{claim:estimation-accuracy} and is clearly at most $n^{-0.2}$ as long as we choose $n$ sufficiently large.  This completes the proof.
\end{proof}

\section{Testing Samples Using Implicit Moments}\label{sec:testing1}

Now we show how to use the projection maps $\Pi_t, \dots , \Pi_1$ computed by Algorithm \ref{alg:iterative-projection} to test whether a sample came from a component with mean close to $0$ or mean far away from $0$.  Roughly, given a sample $z$, the test simply works by computing $R_t(z, z_1, \dots , z_{2t-1})$ for $z_1, \dots , z_{2t-1} \sim \mcl{D}$ and computing 
\[
\norm{ \Gamma_{\Pi_t, \dots , \Pi_1}  \flatten(R_t(z, z_1, \dots , z_{2t-1}))}  \,.
\]
We output {\sc Far} if the above is larger than some threshold and otherwise we output {\sc Close}.  For technical reasons, we will actually average over multiple independent draws for $z_1, \dots , z_{2t -1} \sim \mcl{D}$.  
\\\\
Roughly, the intuition for why this test works is as follows.  Note that if $z \sim \mcl{D}$ then by Corollary \ref{coro:rank1-estimator-mean},
\[
\E \left[  \Gamma_{\Pi_t, \dots , \Pi_1}  \flatten(R_t(z, z_1, \dots , z_{2t-1})) \right] = 0 
\]
and if we control the variance using Corollary \ref{coro:rank1-estimator-variance}, then we can upper bound the length with reasonable probability.  On the other hand if $z \sim \mcl{D}(\mu_i)$ for some $\mu_i$ with large norm, then 
\[
\E \left[  \Gamma_{\Pi_t, \dots , \Pi_1}  \flatten(R_t(z, z_1, \dots , z_{2t-1}) ) \right] =   \Gamma_{\Pi_t, \dots , \Pi_1} \flatten(\mu_i^{\otimes t} ) 
\]
and since the algorithm in the previous section can ensure that $\mu_i^{\otimes t}$ is essentially contained in the row span of $\Gamma_{\Pi_t, \dots , \Pi_1}$, the RHS above has large norm.  The details of our algorithm for testing samples are described below.

\begin{algorithm}[H]
\caption{{\sc Test Samples} }
\begin{algorithmic}
\State \textbf{Input:} Projection matrices $\Pi_t, \dots ,\Pi_2 \in \R^{k \times dk}, \Pi_1 \in \R^{k \times d}$
\State \textbf{Input:} Sample $z \in \R^d$ to test
\State \textbf{Input:} Threshold $\tau$, desired accuracy $\delta$
\State Set $n = ( (10^3 t)^t / \delta)^3$
\For{$i = 1,2, \dots , n$}
\State Draw samples $z_1, \dots , z_{2t-1} \sim \mcl{D}$
\State Let $A_i = \Gamma_{\Pi_t, \dots , \Pi_1}  \flatten(R_t(z, z_1, \dots , z_{2t-1}))$
\EndFor
\State Set $A = (A_1 + \dots + A_n)/n $
\If{$\norm{A} \geq \tau$}
\State \textbf{Output:} {\sc Far}
\Else
\State \textbf{Output:} {\sc Close}
\EndIf
\end{algorithmic}
\label{alg:test-sample}
\end{algorithm}

\subsection{Analysis of {\sc Test Samples}}

Now we analyze the behavior of Algorithm \ref{alg:test-sample}.  The key properties that the test satisfies are summarized in the following two lemmas.  Lemma \ref{lem:length-of-nonzero} say that with $1 - \delta$ probability the test will successfully output {\sc Far} for samples from a component with mean far from $0$ and Lemma \ref{lem:length-of-zero} says that with $1 - \delta$ probability, the test will successfully output {\sc Close} for samples from a component with mean $0$.  

Note that Lemma \ref{lem:length-of-nonzero} requires that the row span of $\Gamma_{\Pi_t, \dots , \Pi_1}$ essentially contains $\flatten( \mu_i^{\otimes t})$ (which can be guaranteed by Algorithm \ref{alg:iterative-projection} and Lemma \ref{lem:projection-accuracy}).  Lemma \ref{lem:length-of-zero} actually does not require anything about $\Pi_t, \dots , \Pi_1$ (other than the fact that they are actually projections).

\begin{lemma}\label{lem:length-of-zero}
Let $\mcl{D}$ be a distribution that is $1$-Poincare.  Let $t \in \N$ and $0 < \delta < 0.01$ be some parameters.  Let  $\Pi_t, \dots, \Pi_2 \in \R^{k \times dk} , \Pi_1 \in \R^{k \times d}$ be any matrices whose rows are orthonormal.  Let $\tau$ be some parameter satisfying $\tau \geq  (20t)^t k/\delta$.  Let $z \sim \mcl{D}$.  Then with probability at least $1 - \delta$, Algorithm \ref{alg:test-sample} run with these parameters outputs {\sc Close} where the randomness is over $z$ and the random choices within Algorithm \ref{alg:test-sample}.
\end{lemma}

\begin{lemma}\label{lem:length-of-nonzero}
Let $\mcl{D}$ be a distribution that is $1$-Poincare.  Let $t \in \N$ and $0 < \delta <  0.01$ be some parameters.  Let $z \sim \mcl{D}(\mu_i)$ where $\norm{\mu_i} \geq 10^4 ( \log 1/\delta + t)$.  Let $\tau$ be some parameter satisfying $\tau \leq (0.4 \norm{\mu_i})^t$.  Assume that the matrices $\Pi_t, \dots,  \Pi_2 \in \R^{k \times dk}, \Pi_1 \in \R^{k \times d}$ satisfy that
\[
\norm{\Gamma_{\Pi_t, \dots , \Pi_1} \flatten( \mu_i^{\otimes t})} \geq (1 - t\eps) \norm{\mu_i}^t \,. 
\]
where
\[
\eps < \frac{\delta}{(10t \norm{\mu_i})^{4t}} \,.
\]
Then with probability at least $1 - \delta$, Algorithm \ref{alg:test-sample} run with these parameters outputs {\sc Far} (where the randomness is over $z$ and the random choices within Algorithm \ref{alg:test-sample}).
\end{lemma}

\begin{remark}
Note that we will want to run the test with some inverse polynomial failure probability i.e. $\delta = \poly(k/w^*)$ for some weight threshold $w^*$.  In order to be able to combine Lemma \ref{lem:length-of-zero} and Lemma \ref{lem:length-of-nonzero} meaningfully, we need 
\[
(0.4 \norm{\mu_i})^t \geq (20t)^t k/\delta \,. 
\]
If $\norm{\mu_i} \geq (\log (k/w^*))^{1 + c}$ for some constant $c > 0$ then setting $t \sim   O\left( c^{-1}\log (k/w^*)/ \log \log (k/w^*)\right)$ ensures that the above inequality is true.  Note that for this setting, $t^t = \poly(k/w^*)$ (where we treat $c$ as a constant) and thus we will be able to ensure that our overall runtime is polynomial.
\end{remark}

We will first prove Lemma \ref{lem:length-of-zero} (which is much easier to prove than Lemma \ref{lem:length-of-nonzero}).  In fact, a direct variance bound using Corollary \ref{coro:rank1-estimator-variance} will suffice.

\begin{proof}[Proof of Lemma \ref{lem:length-of-zero}]
Note that the matrix $\Gamma_{\Pi_t, \dots , \Pi_1}$ has (up to) $k$ orthonormal rows, say $v_1, \dots , v_k$.  Let $R$ be the average of $R_t(z, z_1, \dots , z_{2t-1})$ over $n = ((10^3t)^t/\delta)^3$ trials where $z$ is drawn once and $z_1, \dots , z_{2t-1}$ are sampled independently in each trial.  By Cauchy Schwarz, for any vector $v$,
\[
\E[(v \cdot \flatten(R))^2 ] \leq \E_{z \sim \mcl{D}(\mu_i), z_1, \dots , z_{2t-1} \sim \mcl{D}} [(v \cdot \flatten(R_t(z, z_1, \dots , z_{2t-1})))^2 ]\,.
\]
In other words, the covariance matrix of $R$ is smaller than the covariance of $R_t(z, z_1, \dots , z_{2t-1})$ (in the semidefinite ordering).  Now fix $i \in [k]$.   By Corollary \ref{coro:rank1-estimator-variance} and Markov's inequality, we have with probability at least $1 - \delta/k$
\[
\left \lvert  v_i \cdot \flatten( R ) \right \rvert \leq \sqrt{\frac{k}{\delta}} (20t)^{t} \,.
\]
Union bounding over all $i$, with probability at least $1 - \delta$, we have
\[
\norm{\Gamma_{\Pi_t, \dots , \Pi_1} \flatten( R )} \leq \frac{k}{\delta} (20t)^t \,.
\]
Note that the expression $\Gamma_{\Pi_t, \dots , \Pi_1} \flatten(R)$ is exactly equivalent to the vector $A$ that is tested by Algorithm \ref{alg:test-sample} so with probability at least $1- \delta$, the final output is {\sc Close}.
\end{proof}

\noindent A direct variance bound will not work for Lemma \ref{lem:length-of-nonzero}.  This is because we want the norm to be large with $1 - \delta$ probability but the variance is comparable to the squared length of the mean so we cannot get strong enough concentration with just a variance bound.  Instead, we need a more precise argument.  We will first, in the next claim, prove a bound on evaluations of the polynomial $P_{t, \mcl{D}}$ (recall Definition \ref{def:adjusted-polynomial}).  Essentially, we will argue that under the conditions of Lemma \ref{lem:length-of-nonzero}, for $z \sim \mcl{D}(\mu_i)$, with high probability, the tensor $P_{t, \mcl{D}}(z)$ has large inner product with $\mu_i^{\otimes t}$.  Since $\flatten(\mu_i^{\otimes t})$ is essentially contained in the row span of $\Gamma_{\Pi_t, \dots , \Pi_1}$, this implies that 
\[
\Gamma_{\Pi_t, \dots , \Pi_1} \flatten( P_{t, \mcl{D}}(z)) 
\]
must have large norm.  By Corollary \ref{coro:rank1-estimator-mean},
\[
\E_{z_1, \dots , z_{2t-1} \sim \mcl{D}}[R_t(z, z_1, \dots , z_{2t-1})] = P_{t, \mcl{D}}(z)
\]
so if we average over enough independent samples for $z_1, \dots , z_{2t-1}$ then the estimator $A$ computed in Algorithm \ref{alg:test-sample} will concentrate around its mean and also have large norm which is exactly what we want.  From now on, we will drop the subscript $\mcl{D}$ as the adjusted polynomial will always be defined with respect to $\mcl{D}$.

\begin{claim}\label{claim:polynomial-is-nonzero}
Let $x \in \R^d$.  Let $v \in \R^d$ be a vector such that $v \cdot x \geq 200 t \norm{v}$.  Then 
\[
\langle P_t(x) , v^{\otimes t} \rangle \geq  (0.9 v \cdot x)^t \,.
\]
\end{claim}
\begin{proof}
WLOG we may assume $v$ is a unit vector.  Let $a = v \cdot x$.  We will use Claim \ref{claim:explicit-formula} to rewrite $P_t(x)$.  Note that for any integer $s$, $\langle D_s, v^{\otimes s} \rangle = \E_{z \sim \mcl{D}}[ (z \cdot v)^s]$.  Combining with Fact \ref{fact:basic-Poincare}, we have
\[
| \langle D_s, v^{\otimes s} \rangle | \leq 6 \cdot s!  \,.
\]
Now by Claim \ref{claim:explicit-formula}, we have
\begin{align*}
&\langle P_t(x) , v^{\otimes t} \rangle \\ &= \Bigg\langle v^{\otimes t}, \sum_{ S_0 \subset [t] }  \left( x^{\otimes S_0 } \right) \otimes  \left( \sum_{\{S_1, \dots , S_t \} \in Z_t( [t] \backslash S_0)} (-1)^{\mcl{C} \{S_1, \dots , S_t \} } (\mcl{C} \{S_1, \dots , S_t \} )! (D_{|S_1|})^{(S_1)} \otimes \dots \otimes (D_{|S_t|})^{(S_t)}   \right)    \Bigg\rangle \\ & \geq a^t - \sum_{ S_0 \subset [t] , S_0 \neq [t]} a^{|S_0|} \left( \sum_{\{S_1, \dots , S_t \} \in Z_t( [t] \backslash S_0)}6^{\mcl{C} \{S_1, \dots , S_t \} }  (\mcl{C} \{S_1, \dots , S_t \} )! |S_1|! \cdots |S_t|! \right) \\\ & \geq a^t - \sum_{ S_0 \subset [t] , S_0 \neq [t]} a^{|S_0|} 6^{t - |S_0|}\left( \sum_{c = 1}^{t - |S_0|}  \sum_{\substack{S_1 \cup \dots \cup S_c = [t] \backslash S_0 \\ S_i \cap S_j = \emptyset, S_i \neq \emptyset}} |S_1|! \cdots |S_c|!\right) \\ & =  a^t - \sum_{ S_0 \subset [t] , S_0 \neq [t]} a^{|S_0|} 6^{t - |S_0|}\left( \sum_{c = 1}^{t - |S_0|}  \sum_{\substack{s_1 +  \dots + s_c = t - |S_0| \\  s_i > 0}} (t - |S_0|)!\right)  \\ & \geq a^t - \sum_{ S_0 \subset [t] , S_0 \neq [t]} a^{|S_0|} 6^{t - |S_0|} (t - |S_0|)! 2^{t - |S_0|} \\ & \geq a^t - \sum_{c = 1}^t \binom{t}{c} c! 12^c a^{t - c}  \geq a^t - \sum_{c = 1}^t (12t)^c a^{t - c} = a^t \left( 1 - \sum_{c = 1}^t \left(\frac{12t}{a}\right)^c \right)  \geq 0.9 a^t \,. 
\end{align*}
This completes the proof.
\end{proof}

Now we are ready to prove Lemma \ref{lem:length-of-nonzero}.
\begin{proof}[Proof of Lemma \ref{lem:length-of-nonzero}]
Let $v = \mu_i/\norm{\mu_i}$.  By Fact \ref{fact:basic-Poincare}, with probability at least $1 - \delta/10$, we have 
\begin{equation}\label{eq:cond1}
v \cdot z \geq 0.9 \norm{\mu_i} \,.
\end{equation}
Now by Claim \ref{claim:polynomial-is-nonzero}, we have
\begin{equation}\label{eq:unprojected-bound}
\langle P_t(z), v^{\otimes t} \rangle \geq  (0.8 \norm{\mu_i})^t \,.
\end{equation}
Now by assumption (and the fact that $v$ is a scalar multiple of $\mu_i$), we have 
\[
\Gamma_{\Pi_t, \dots , \Pi_1}^T \Gamma_{\Pi_t, \dots , \Pi_1} \flatten( v^{\otimes t}) = \flatten(v^{\otimes t}) + u
\]
where $u$ is a vector with
\[
\norm{u} \leq  \sqrt{1 - (1 - t\eps)^2} \leq \sqrt{2t\eps} \,.
\]
By Claim \ref{claim:variance-bound} and Markov's inequality, we have that with probability at least $1 - \delta/10$ over the choice of $z$,
\[
\left \lvert u \cdot \flatten(P_t(z)) \right \rvert \leq \sqrt{10/\delta} \cdot (t + \norm{\mu_i})^t \norm{u} \leq \sqrt{20t\eps/\delta} (2t \norm{\mu_i})^t \leq 1 
\]
where the last step follows from our condition on $\eps$.  Now we can combine this with (\ref{eq:unprojected-bound}) to get that
\begin{equation}\label{eq:planted-part}
\langle \flatten(P_t(z)), \Gamma_{\Pi_t, \dots , \Pi_1}^T \Gamma_{\Pi_t, \dots , \Pi_1} \flatten( v^{\otimes t})  \rangle \geq  (0.8 \norm{\mu_i})^t - 1 \geq  ( 0.5 \norm{\mu_i})^t \,.
\end{equation}
Now let 
\[
w = \Gamma_{\Pi_t, \dots , \Pi_1}^T \Gamma_{\Pi_t, \dots , \Pi_1} \flatten( v^{\otimes t}) \,.
\]
Note that $\norm{w} \leq 1$ because it is the projection of a unit vector $\flatten(v^{\otimes t})$ onto some subspace.  By Corollary \ref{coro:rank1-estimator-variance} and Markov's inequality, with probability at least $1  - \delta/10$ over the randomness in $z$, we have
\[
\E_{z_1, \dots , z_{2t-1} \sim \mcl{D}}\left[ \left( w \cdot  \flatten(R(z, z_1, \dots , z_{2t-1}) )\right)^2 \right] \leq  (10/\delta)(20t)^{2t}(\norm{\mu_i}^{2t} + 1) \,.
\]
Assuming that the above holds, we now treat $z$ as fixed.  Note that the mean of $R(z,z_1, \dots , z_{2t-1})$ is $P_t(z)$ by Corollary \ref{coro:rank1-estimator-mean}.  Let $R$ be obtained by averaging $R(z,z_1, \dots , z_{2t-1})$ over $n = ((10^3 t)^t/\delta)^3$ independent trials for $z_1, \dots , z_{2t-1}$.  Of course the mean of $R$ is also $P_t(z)$.   Using the above, we have that 
\begin{align*}
\E\left[ \left( w \cdot  (\flatten(R) - \flatten(P_t(z)))\right)^2 \right] &= \frac{1}{n} \E_{z_1, \dots , z_{2t-1} \sim \mcl{D}}\left[ \left( w \cdot  (\flatten(R(z, z_1, \dots , z_{2t-1}) - P_t(z)) )\right)^2 \right] \\ & \leq \frac{1}{n} \E_{z_1, \dots , z_{2t-1} \sim \mcl{D}}\left[ \left( w \cdot  \flatten(R(z, z_1, \dots , z_{2t-1}) )\right)^2 \right] \\ & \leq  (\delta / (10^3 t)^t)^2 (20t)^{2t}(\norm{\mu_i}^{2t} + 1) \\ & \leq  \frac{\delta}{10} \left(\frac{\norm{\mu_i}}{20} \right)^{2t}  \,.
\end{align*}
Thus, with probability at least $1 - \delta/10$, we have 
\[
|w \cdot (\flatten(R) - \flatten(P_t(z)))| \leq (0.05 \norm{\mu_i})^t  \,.
\]
If this holds, then combining with (\ref{eq:planted-part}) implies
\[
w \cdot \flatten(R) \geq (0.4 \norm{\mu_i})^t
\]
but since $v^{\otimes t}$ is a unit vector and $w$ is its projection onto the subspace spanned by the rows of $\Gamma_{\Pi_t, \dots , \Pi_1}$, the above implies 
\[
\norm{\Gamma_{\Pi_t, \dots , \Pi_1} \flatten(R) } \geq (0.4 \norm{\mu_i})^t \,.
\]
Note that the expression $\Gamma_{\Pi_t, \dots , \Pi_1} \flatten(R)$ is exactly equivalent to the vector $A$ that is tested by Algorithm \ref{alg:test-sample} so combining all of the failure probabilities, we conclude that with probability at least $1 - \delta$, the output of Algorithm \ref{alg:test-sample} will be {\sc Far}.
\end{proof}

\section{Learning Mixtures of Poincare Distributions}\label{sec:poincare-main}

We are now ready to prove Theorem \ref{thm:main-Poincare}, our full result for mixtures of Poincare distributions.  Let $\mcl{D}' = (\mcl{D} - \mcl{D}) / \sqrt{2}$ i.e. $\mcl{D}'$ is the distribution of the difference between two independent samples from $\mcl{D}$ scaled down by $\sqrt{2}$.  The $\sqrt{2}$ scaling ensures that $\mcl{D}'$ is $1$-Poincare by Fact \ref{fact:basic-Poincare}.  Note that we can take pairwise differences between samples to simulate access to the mixture
\[
\mcl{M}' = (w_1^2 + \dots + w_k^2)\mcl{D}' + \sum_{i \neq j} w_iw_j \mcl{D}'((\mu_i - \mu_j)/\sqrt{2}) \,.
\]
We will run Algorithm \ref{alg:iterative-projection} and Algorithm \ref{alg:test-sample} for the distribution $\mcl{D}'$ and mixture $\mcl{M}'$.  Note that the test in Algorithm \ref{alg:test-sample} will test for a pair of samples say $z,z'$, from components $\mcl{D}(\mu_i), \mcl{D}(\mu_j)$ of the original mixture, whether $\norm{\mu_i - \mu_j}$ is large or zero.  Running this test between all pairs of samples will let us form clusters of samples that correspond to each of the components.    We begin with a lemma that more precisely specifies the guarantees of this test and then Theorem \ref{thm:main-Poincare} will follow easily from it.

\begin{lemma}\label{lem:main-test-Poincare}
Let $\mcl{D}$ be a $1$-Poincare distribution on $\R^d$.  Let 
\[
\mcl{M} = w_1 \mcl{D}(\mu_1) + \dots + w_k \mcl{D}(\mu_k)
\]
be a mixture of translated copies of $\mcl{D}$.  Let $w^*, s, \delta$ be parameters and assume that $s \geq (\log k/(w^* \delta))^{1 + c}$ for some $0 < c < 1$.  Also assume that 
\[
\max \norm{\mu_i - \mu_j} \leq \min( (k/(w^*\delta))^2, s^{C} )
\]
for some $C$.  There is an algorithm that takes $n = \poly((kd/(w^*\delta) )^{C/c} )$ samples from $\mcl{M}$ and $\mcl{D}$ and runs in $\poly(n)$ time and achieves the following testing guarantees:  for a pair of (independent) samples $z \sim \mcl{D}(\mu_i) ,z' \sim \mcl{D}(\mu_j)$,
\begin{itemize}
    \item If $i = j$ and $w_i \geq w^*$ then with probability $1 - \delta$ the output is $\textsf{accept}$
    \item If $w_i, w_j \geq w^*$ and $\norm{\mu_i - \mu_j} \geq s$ then with probability $1 - \delta$ the output is $\textsf{reject}$
\end{itemize}
where the randomness is over the samples $z,z'$ and the random choices in the algorithm.
\end{lemma}
\begin{remark}
Note that Lemma \ref{lem:main-test-Poincare} actually does not require a minimum separation or minimum mixing weight but instead just guarantees to detect when two samples come from components that have nontrivial mixing weights and whose means are far apart.   We will focus on achieving inverse polynomial accuracy i.e. $\delta = (w^*/k)^{O(1)}$ so the required separation will be $(\log k/w^*)^{1 + c}$ and the runtime will be polynomial in $k,d, 1/w^*$.  
\end{remark}
\begin{proof}
Let $t$ be the minimum integer such that 
\[
\left( \frac{s}{\log (k/(w^*\delta))} \right)^t \geq (k/(w^* \delta))^{10} \,.
\]
Note that we clearly must have that 
\begin{align*}
&\max \norm{\mu_i - \mu_j}^t \leq s^{tC} \leq  \left( \frac{s}{\log (k/(w^* \delta))} \right)^{Ct/c} \leq (k/(w^* \delta))^{20 C/c} \\
&t \leq \frac{20 \log  (k/(w^* \delta))}{c \log \log  (k/(w^* \delta))} \,.
\end{align*}
Now we run Algorithm \ref{alg:iterative-projection} with distribution $\mcl{D}'$ and mixture $\mcl{M}'$ and parameter $t$ to obtain projection matrices $\Pi_t, \dots , \Pi_1$.  Set 
\[
\eps = \frac{\delta}{(10t s^C)^{4t}} \,.
\]
Note that by the above, the runtime and sample complexity  of $n = \poly((kd/(w^*\delta) )^{C/c} )$ suffices for us to apply Lemma \ref{lem:projection-accuracy} with accuracy parameter $\eps$ and weight threshold $(w^*)^2$.  This implies that for all $i \neq j$ such that $\norm{\mu_i - \mu_j} \geq \sqrt{2}$ and $w_i , w_j \geq w^*$, we have
\begin{equation}\label{eq:accuracy}
\norm{\Gamma_{\Pi_{s}, \dots , \Pi_1} \flatten( (\mu_i - \mu_j)^{\otimes s}) } \geq ( 1 - s\eps) \norm{\mu_i - \mu_j}^s  
\end{equation}
for all $s = 1,2, \dots , t$.  Now we run Algorithm \ref{alg:test-sample} with projection matrices $\Pi_t, \dots , \Pi_1$, accuracy $\delta$ and threshold 
\[
\tau = (0.2s)^t \,.
\]
We use Algorithm \ref{alg:test-sample} to test $(z - z')/\sqrt{2}$ and we output {\sc Accept} if Algorithm \ref{alg:test-sample} returns {\sc Close} and we output {\sc Reject} if Algorithm \ref{alg:test-sample} returns {\sc Far}.  Note that runtime and sample complexity  of $n = \poly((kd/(w^*\delta) )^{C/c} )$ suffice.  Also $(z - z')/\sqrt{2}$ is equivalent to a sample from $\mcl{D}'((\mu_i - \mu_j)/\sqrt{2})$.   We now apply Lemma \ref{lem:length-of-zero} to verify the first guarantee of our test.  Note that it suffices to verify that 
\[
\tau \geq \frac{(20t)^tk}{\delta}
\]
but this clearly holds because
\[
\frac{\tau}{(20t)^t} = \left(\frac{0.2s}{20t}\right)^t \geq \left( \frac{c s}{10^4\log (k/(w^*\delta))} \right)^{t} \geq \frac{k}{\delta},.
\]
Thus, we conclude that if $z,z'$ are drawn from the same component then the output is {\sc Accept} with probability at least $1 - \delta$.  Now we use Lemma \ref{lem:length-of-nonzero} to verify the second guarantee of our test.  Note that 
\[
\frac{\norm{\mu_i - \mu_j}}{\sqrt{2}} \geq \frac{s}{\sqrt{2}} \geq 10^4(\log 1/\delta + t) \,.
\]
Also note that the threshold clearly satisfies $\tau \leq (0.4 \norm{\mu_i - \mu_j}/\sqrt{2})^t $ (note that the $\sqrt{2}$ comes from the fact that we scale the difference down by $\sqrt{2}$).  Combined with the guarantee in (\ref{eq:accuracy}), we conclude that the output is {\sc Reject} with probability at least $1 - \delta$ and we are done. 
\end{proof}

Using Lemma \ref{lem:main-test-Poincare}, it is not difficult to prove Theorem \ref{thm:main-Poincare}.

\begin{proof}[Proof of Theorem \ref{thm:main-Poincare}]
Recall by the reduction in Section \ref{sec:reductions} that we may assume $d \leq k$ and that $\norm{\mu_i - \mu_j} \leq O((k/w_{\min})^2)$ for all $i,j$.  Now we will do the following process to estimate the means of the components.
\begin{enumerate}
    \item Draw a sample $z \sim \mcl{M}$
    \item Take $m = (k/(w_{\min} \alpha))^{10^2}$ samples $z_1, \dots , z_m \sim \mcl{M}$
    \item Use Lemma \ref{lem:main-test-Poincare} with parameters $w^* = w_{\min}$ and $\delta = (w_{\min} \alpha/k)^{10^4}$ to test the pair of samples $z,z_i$ for all for all $i \in [m]$
    \item Let $S \subset [m]$ be the set of all $i$ that are {\sc Accepted} and compute $\mu = \frac{1}{|S|}\sum_{i \in S} z_i$
\end{enumerate}  
We first argue that if $z$ is a sample from $\mcl{D}(\mu_i)$ for some $i$, then with $1 - (w_{\min} \alpha/k)^{10^2}$ probability, the procedure returns $\mu$ such that $\norm{\mu - \mu_i} \leq 0.1\alpha$.  This is because by the guarantees of Lemma \ref{lem:main-test-Poincare}, with probability at least $1 - (w_{\min} \alpha/k)^{10^3}$ the test will accept all samples from among $z_1, \dots , z_m$ that are from the component $\mcl{D}(\mu_i)$ and reject all of the others.  Also, with high probability, there will be at least $(k/(w_{\min}\alpha))^{99}$ samples from the component $\mcl{D}(\mu_i)$ so by Claim \ref{claim:Poincare-mean} and union bounding all of the failure probabilities, we have that if $z$ is a sample from $\mcl{D}(\mu_i)$ then $\norm{\mu - \mu_i} \leq 0.1 \alpha$ with probability at least $1 - (w_{\min} \alpha/k)^{10^2}$.

Now it suffices to repeat the procedure in steps $1-4$ for $l = (k/(w_{\min} \alpha))^{10^2}$ independent samples $z \sim \mcl{M}$.  This gives us a list of means say $S = \{ \wt{\mu_1}, \dots , \wt{\mu_l} \}$.  The previous argument implies that most of these estimates will be close to one of the true means and that all of the true means will be represented.  To ensure that with high probability we output exactly one estimate corresponding to each true mean and no extraneous estimates, we do a sort of majority voting.  

We will inspect the estimates in $S$ one at a time and decide whether to output them or not.  Let $T$ be the set of estimates that we will output.  Note that $T$ is initially empty.  Now for each $i$, let $S_i$ be the subset of $\{ \wt{\mu_1}, \dots , \wt{\mu_l} \}$ consisting of all means with $\norm{\wt{\mu_j} - \wt{\mu_i}} \leq 0.2 \alpha$.  If $|S_i| \geq 0.9 w_{\min} l$ and $\wt{\mu_i}$ is not within $\alpha$ of any element of $T$ then add $\wt{\mu_i}$ to $T$.  Otherwise do nothing.  

We claim that with high probability, this procedure returns one estimate corresponding to each true means and nothing extraneous.  First, with high probability there will be no extraneous outputs because if say $\wt{\mu_i}$ is at least $0.5 \alpha$ away from all of the true means, then with high probability we will have $|S_i| < 0.9 w_{\min} l$.  Now it is also clear that with high probability we output exactly one estimate corresponding to each true mean (since $l$ is sufficiently large that with high probability we get enough samples from each component).  Thus, with high probability, the final output will be a set of means that are within $\alpha$ of the true means up to some permutation.  Once we have learned the means, we can learn the mixing weights by simply taking fresh samples from $\mcl{M}$ and clustering since with high probability, we can uniquely identify which component a sample came from.  This completes the proof.
\end{proof}

The clustering guarantee in Corollary \ref{coro:cluster-Poincare} follows as an immediate consequence of Theorem \ref{thm:main-Poincare}.
\begin{proof}[Proof of Corollary \ref{coro:cluster-Poincare}]
Let the estimated means computed by Theorem \ref{thm:main-Poincare} for $\alpha = (w_{\min}/k)^{10}$ be $\wt{\mu_1}, \dots , \wt{\mu_k}$.  Now for all $j_1, j_2 \in [k]$ with $j_1 \neq j_2$, let
\[
v_{j_1j_2} = \frac{\wt{\mu_{j_1}} - \wt{\mu_{j_2}}}{\norm{\wt{\mu_{j_1}} - \wt{\mu_{j_2}} }} \,.
\]
Now given a sample $z$ from $\mcl{M}$, we compute the index $j$ such that for all $j_1, j_2$, we have
\[
| v_{j_1j_2} \cdot (\wt{\mu_j} - z) | \leq (\log (k/w_{\min}))^{1 + 0.5c} \,. 
\]
Note that by the guarantees of Theorem \ref{thm:main-Poincare}, there is a permutation $\pi$ such that $\norm{\wt{\mu_{\pi(i)}} -\mu_i} \leq \alpha$ for all $i$.  If $z$ is a sample from $\mcl{D}(\mu_i)$, then by the tail bound in Fact \ref{fact:basic-Poincare}, with high probability the unique index $j$ that satisfies the above is exactly $j = \pi(i)$ and thus, we recover the ground truth clustering with high probability. 
\end{proof}

\section{Sharper Bounds for Gaussians}\label{sec:gaussian-moments}

For spherical Gaussians, i.e. when $\mcl{D} = N(0,I)$, we can obtain stronger results.  Our results are stronger in two ways.  First, we can improve the minimum separation to $(\log (k/w_{\min}))^{1/2 + c}$ instead of $(\log (k/w_{\min})^{1 + c}$.  Secondly, we can remove the assumption about the maximum separation by using a recursive clustering routine.  We first focus on improving the minimum separation. To get this improvement, we will prove sharper quantitative bounds in our implicit moment tensor and testing subroutines.

\subsection{Hermite Polynomials}

Naturally, for the standard Gaussian $N(0,I)$, the adjusted polynomials $P_{t, N(0,I)}$ are exactly the Hermite polynomials.  In this section, we define and go over a few basic properties of the Hermite polynomials.  Many of these are from \cite{kane2020robust}.  In the next subsection, we will show how to exploit specific properties of the Hermite polynomials to get sharper versions of the results in Sections \ref{sec:implicit-est} and \ref{sec:testing1}.

\begin{definition}
We will use $P(S)$ to denote the set of partitions of a set $S$.  We use $P^2(S)$ to denote the set of partitions of $S$ into subsets of size $2$.  We use $P^{1,2}(S)$ to denote the set of partitions of $S$ into  subsets of size $1$ and $2$.
\end{definition}

\begin{definition}[Hermite Polynomial Tensor (see \cite{kane2020robust})]\label{def:hermite}
Let $X = (X_1, \dots , X_d)$ be a vector of $d$ formal variables.  We define the Hermite polynomial tensor
\[
h_t(X) = \sum_{ P \in P^{1,2}([t])} \bigotimes_{\{a,b \} \in P} (-I)^{(a,b)} \otimes  \bigotimes_{ \{c \} \in P } X^{(c)}
\]
where $I$ denotes the $d \times d$ identity matrix.
\end{definition}

Below we summarize several well-known properties that the Hermite polynomials satisfy.

\begin{claim}[See \cite{kane2020robust}]\label{claim:hermite-mean}
Let $G = N(\mu, I)$.  Then
\[
\E_{z \sim G}[ h_t(z)] = \mu^{\otimes t} \,.
\]
\end{claim}

\begin{claim}[See \cite{kane2020robust}]\label{claim:Gaussian-moments}
We have the identity
\[
\E_{z \sim N(0,I)}[ z^{\otimes t}] = \sum_{P \in P^2([t])} \bigotimes_{\{a,b \} \in P}I^{(a,b)} \,.
\]
\end{claim}

We also have that the Hermite polynomials are exactly the adjusted polynomials (recall Definition \ref{def:adjusted-polynomial}) for the standard Gaussian.  

\begin{claim}\label{claim:special-case-equivalence}
For any $x \in \R^d$ and integer $t$, we have
\[
P_{t, N(0,I)}(x) = h_t(x) \,.
\]
\end{claim}
\begin{proof}
We can verify the desired statement by induction.  The base cases for $t = 1,2$ are clear.  To finish, we can simply plug Claim \ref{claim:Gaussian-moments} into the recursion in (\ref{eq:recursive-def}) to verify the inductive step.
\end{proof}

\noindent We could get a bound on the variance of the Hermite polynomials using Claim \ref{claim:variance-bound}.  However, since our goal for Gaussians will be to get separation $(\log (k/w_{\min}))^{1/2 + c}$ instead of separation $(\log (k/w_{\min}))^{1 + c}$, we will obtain a stronger bound by hand.  Fortunately, we will only need the stronger bound for when the mean is $0$ and can get away with using the general bounds from Section \ref{sec:moments1} for when the mean $\mu$ is nonzero.

\begin{claim}[From \cite{kane2020robust}]\label{claim:hermite-variance}
We have
\[
\E_{z \sim N(\mu, I)}[ h_t(z) \otimes h_t(z)] = \sum_{S_1 , S_2 \subset [t], |S_1| = |S_2|} \sum_{\substack{\text{Matchings } \mcl{P} \\ \text{of } S_1, S_2}} \bigotimes_{\{a,b \} \in \mcl{P}} I^{(a, t + b)} \bigotimes_{c \notin S_1} \mu^{(c)} \bigotimes_{c \notin S_2} \mu^{(t + c)} \,.
\]
\end{claim}

Flattening the above expression for $\E_{z \sim G}[ h_t(z) \otimes h_t(z)]$ into a $d^t \times d^t$ matrix  in the natural way, we immediately get a bound on the covariance of $h_t(z)$ for $z \sim G$
\begin{corollary}\label{coro:hermite-variancebound}
Let
\[
M(z) = \textsf{flat}(h_t(z)) \otimes \textsf{flat}(h_t(z)) \,.
\] 
Then 
\[
\E_{z \sim N(0,I)}[M(z)]  \preceq t!  I
\]
where $I$ is the $d^t \times d^t$ identity matrix.
\end{corollary}
\begin{proof}
This follows immediately from the formula in Claim \ref{claim:hermite-variance} since there are $t!$ terms that are a tensor product of $t$ copies of the identity matrix and each of these has spectral norm $1$.
\end{proof}

We will also need the following property, that the Hermite polynomials are an orthogonal family.  Again, this property was not true for general Poincare distributions but will be used in getting stronger quantitative bounds that will let us deal with smaller separation.

\begin{claim}\label{claim:hermite-orthogonality}
For integers $t \neq t'$, we have
\[
\E_{z \sim N(0,I)}[ \textsf{flat}(h_t(z)) \otimes  \textsf{flat}(h_{t'}(z))] = 0 \,.
\]
\end{claim}
\begin{proof}
WLOG $t' < t$.  We first prove
\[
\E_{z \sim N(0,I)}[ h_t(z) \otimes  z^{\otimes t'}] = 0 \,.
\]
Substitute the expression in Definition \ref{def:hermite} into the LHS and then use Claim \ref{claim:Gaussian-moments}.  The above then follows from direct computation.  Next, since the above holds for all $t' < t$ and $h_t'(z)$ is a sum of monomials of degree at most $t'$, we immediately get that
\[
\E_{z \sim N(0,I)}[ h_t(z) \otimes  h_{t'}(z)]
\]
as desired.
\end{proof}

\subsection{Implicit Representations of Hermite Polynomials}\label{sec:implicit-hermite}

Similar to Section \ref{sec:implicit-moments1}, we will need implicit representations of the Hermite polynomial tensors.  While technically the representation in Definition \ref{def:hermite} is already implicit, the identity matrices $I^{a,b}$ do not behave well when we try to apply our iterative projection map in Section \ref{sec:projection} and thus we will again, similar to Section \ref{sec:implicit-moments1}, need to obtain a representation as a sum of a polynomial number of ``rank-$1$" tensors.

\noindent The next set of definitions exactly parallel those in Section \ref{sec:implicit-moments1}.

\begin{definition}
For $x_1, \dots , x_t \in \R^d$, define the polynomial 
\begin{equation}
Q_t(x_1, \dots , x_t) = \sum_{\substack{S_1 \cup \dots \cup S_t = [t] \\ |S_i \cap S_j| = 0}} \frac{(-1)^{\mcl{C} \{ S_1, \dots , S_t \}}}{\binom{t - 1}{\mcl{C} \{ S_1, \dots , S_t \} - 1 } } \left( h_{|S_1|}(x_1) \right)^{(S_1)} \otimes \dots \otimes \left( h_{|S_t|}(x_t) \right)^{(S_t)} \,.
\end{equation}
\end{definition}

\begin{definition}
For $x_1, \dots , x_{2t} \in \R^d$, define the polynomial
\[
R_t(x_1, \dots , x_{2t}) = - Q_t(x_1, \dots , x_t) + Q_t(x_{t+1}, \dots , x_{2t}) \,.
\]
\end{definition}
\begin{remark}
Note we will use the same letters $Q_t, R_t$ as in the general case.  All of the proceeding sections deal specifically with GMMs and we will only consider the corresponding $Q_t, R_t$ so there will be no ambiguity.
\end{remark}

By Claim \ref{claim:special-case-equivalence}, we have the following (which is the exact same as Corollary \ref{coro:rank1-identity-p2}).

\begin{corollary}\label{coro:rank1-hermite-identity-p2}
We have the identity
\[
R_t(x_1, \dots , x_{2t}) = \sum_{\substack{S_1 \cup \dots \cup S_t = [t] \\ |S_i \cap S_j| = 0}}   \frac{(-1)^{\mcl{C} \{ S_1, \dots , S_t \} -1 }}{\binom{t - 1}{\mcl{C} \{ S_1, \dots , S_t \} - 1 } } \left( x_1^{\otimes S_1} \otimes \dots \otimes x_t^{\otimes S_t} -  x_{t+1}^{\otimes S_1} \otimes \dots \otimes x_{2t}^{\otimes S_t} \right)\,.
\]
\end{corollary}

Claim \ref{claim:special-case-equivalence} also implies that Corollary \ref{coro:rank1-estimator-mean} and Corollary \ref{coro:rank1-estimator-variance} still hold and can be used in the analysis.  We now prove the one strengthened bound that we will need.

\begin{lemma}\label{lem:stronger-variance-bound}
We have
\[
\E_{z_1, \dots , z_{2t} \sim N(0,I)}[ \flatten(R_t(z_1, \dots , z_{2t}))^{\otimes 2}] \preceq (2t)^t I_{d^t} \,.
\]
\end{lemma}
\begin{proof}
Note that by Claim \ref{claim:hermite-orthogonality} and Corollary \ref{coro:hermite-variancebound},
\begin{align*}
&\E_{z_1, \dots , z_{2t} \sim N(0,I)}[ \flatten(R_t(z_1, \dots , z_{2t}))^{\otimes 2}] \\ &=     2\sum_{\substack{S_1 \cup \dots \cup S_t = [t] \\ |S_i \cap S_j| = 0}} \frac{1}{\binom{t - 1}{\mcl{C} \{ S_1, \dots , S_t \} - 1 }^2 } \flatten\left(\left( h_{|S_1|}(x_1) \right)^{(S_1)} \otimes \dots \otimes \left( h_{|S_t|}(x_t) \right)^{(S_t)}\right)^{\otimes 2} \\ &\preceq 2\sum_{\substack{S_1 \cup \dots \cup S_t = [t] \\ |S_i \cap S_j| = 0}} |S_1|! |S_2| \cdots |S_t|! I_{d^t} \\ & \preceq 2 I_{d^t}\sum_{\substack{s_1 + \dots + s_t = t \\ s_i \geq 0}} \binom{t}{s_1, \dots , s_t} s_1! \cdots s_t ! \\ & \preceq (2t)^t I_{d^t}
\end{align*}
where as before we use the trick of reindexing the sum so that $s_1, \dots , s_t$ are the sizes of $S_1, \dots , S_t$ respectively.  This completes the proof.
\end{proof}
Note that the key difference of the above compared to Corollary \ref{coro:rank1-estimator-variance} is that the RHS is $t^t$ instead of $t^{2t}$.  This is the key to improving the separation.

\subsection{Testing Samples}

All of the algorithms and results in Section \ref{sec:implicit-est} still hold (with the distribution $\mcl{D}$ set to $N(0,I)$).  We also run the testing algorithm in Section \ref{sec:testing1}.  It remains to prove that the distinguishing power of the test is stronger when specialized to Gaussians.  Recall that the key lemmas were Lemma \ref{lem:length-of-zero} and  Lemma \ref{lem:length-of-nonzero}.  The strengthened versions are stated below.

\begin{lemma}\label{lem:length-of-zero-stronger}
Let $t \in \N$ and $0 < \delta < 0.01$ be some parameters.  Let  $\Pi_t, \dots, \Pi_2 \in \R^{k \times dk} , \Pi_1 \in \R^{k \times d}$ be any matrices with orthonormal rows.  Let $\tau$ be some parameter satisfying $\tau \geq  (2t)^{t/2} k/\delta$.  Let $z \sim N(0,I)$.  Then with probability at least $1 - \delta$, Algorithm \ref{alg:test-sample} run with these parameters outputs {\sc Close} where the randomness is over $z$ and the random choices within Algorithm \ref{alg:test-sample}.
\end{lemma}

\begin{lemma}\label{lem:length-of-nonzero-stronger}
Let $t \in \N$ and $0 < \delta <  0.01$ be some parameters.  Let $z \sim N(\mu_i, I)$ where $\norm{\mu_i} \geq 10^4 ( \sqrt{\log 1/\delta} + \sqrt{t})$.  Let $\tau$ be some parameter satisfying $\tau \leq (0.4 \norm{\mu_i})^t$.  Assume that the matrices $\Pi_t, \dots,  \Pi_2 \in \R^{k \times dk}, \Pi_1 \in \R^{k \times d}$ satisfy that
\[
\norm{\Gamma_{\Pi_t, \dots , \Pi_1} \flatten( \mu_i^{\otimes t})} \geq (1 - t\eps) \norm{\mu_i}^t \,. 
\]
where
\[
\eps < \frac{\delta}{(10t \norm{\mu_i})^{4t}} \,.
\]
Then with probability at least $1 - \delta$, Algorithm \ref{alg:test-sample} run with these parameters outputs {\sc Far} (where the randomness is over $z$ and the random choices within Algorithm \ref{alg:test-sample}).
\end{lemma}

There are two key difference compared to Lemma \ref{lem:length-of-zero} and Lemma \ref{lem:length-of-nonzero}.  First, in Lemma \ref{lem:length-of-zero-stronger}, the bound on $\tau$ involves $t^{t/2}$ instead of $t^t$.  Second, in Lemma \ref{lem:length-of-nonzero-stronger}, we require separation $O(\sqrt{\log 1/\delta} + \sqrt{t})$ instead of $O(\log 1/\delta + t)$.  When we combine the two, the inequality that we need for our test to be meaningful is
\[
(0.4 \norm{\mu_i})^t \geq (2t)^{t/2} k/\delta \,. 
\]
All of the settings of parameters will be the same as before i.e. $t \sim   O\left( c^{-1}\log (k/w^*)/ \log \log (k/w^*)\right)$ and $\delta = \poly(w^*/k)$ except now we can get away with  $\norm{\mu_i} \geq (\log (k/w^*))^{1/2 + c}$ for some constant $c > 0$ which allows us to improve the minimum separation.
\\\\
The proof of Lemma \ref{lem:length-of-zero-stronger} is essentially exactly the same as the proof of Lemma \ref{lem:length-of-zero}.
\begin{proof}[Proof of Lemma \ref{lem:length-of-zero-stronger}]
Exactly the same as the proof of Lemma \ref{lem:length-of-zero} except use Lemma \ref{lem:stronger-variance-bound} instead of Corollary \ref{coro:rank1-estimator-variance} to bound the variance of $R_t$.
\end{proof}

The proof of Lemma \ref{lem:length-of-nonzero-stronger} requires one additional modification.  In particular, we will use a stronger version of Claim \ref{claim:polynomial-is-nonzero} that we now prove.
\begin{claim}\label{claim:polynomial-is-nonzero-stronger}
Let $x \in \R^d$.  Let $v \in \R^d$ be a vector such that $v \cdot x \geq 20 \sqrt{t} \norm{v}$.  Then 
\[
\langle h_t(x) , v^{\otimes t} \rangle \geq  (0.9 v \cdot x)^t \,.
\]
\end{claim}
\begin{proof}
WLOG $v$ is a unit vector.  Let $a = v \cdot x$.  Using Definition \ref{def:hermite} and elementary counting, we have
\[
\langle h_t(x) , v^{\otimes t} \rangle = \sum_{j = 0}^{\lfloor t/2 \rfloor} (-1)^j a^{t - 2j} \frac{t!}{2^j j! (t-2j)!} \,.
\]
The RHS is a polynomial in $a$ and is exactly the standard univariate Hermite polynomial, which we denote $H_t(a)$.  Note it can be easily verified that the univariate Hermite polynomials satisfy the recurrence
\[
H_t(a) = a H_{t-1}(a) - (t-1)H_{t-2}(a) \,.
\]
Thus, we can write
\[
H_t(a) = \det\begin{bmatrix} a & 1   \\
1 & a & \sqrt{2}  \\
 & \sqrt{2} & a & \ddots \\
 &   & \ddots  & \ddots &  \sqrt{t-1} \\
 & & & \sqrt{t-1} & a 
\end{bmatrix}
\]
where the blank entries are $0$.  Let the matrix on the RHS be $M$.  $M$ is a $t \times t$ matrix that has $a$ on the diagonal and adjacent to the diagonal the entries are $M_{i(i+1)} = M_{(i+1)i} = \sqrt{i}$ and $0$ everywhere else.  Note that this implies that $H_t$ has $t$ real roots (since it is the characteristic polynomial of a symmetric matrix).  Also, all roots have magnitude at most $2\sqrt{t}$ since if $|a| \geq 2\sqrt{t}$ then the matrix is diagonally dominant and cannot have determinant $0$.  Thus, we can write $H_t(a) = (a-r_1) \cdots (a - r_t)$ where $-2\sqrt{t} \leq r_1, \dots , r_t \leq 2 \sqrt{t}$.  For $a \geq 20 \sqrt{t}$, we clearly have $H_t(a) \geq (0.9 a)^t$ and this completes the proof.
\end{proof}

\begin{proof}[Proof of Lemma \ref{lem:length-of-nonzero-stronger}]
Exactly the same as the proof of Lemma \ref{lem:length-of-nonzero} except using Claim \ref{claim:polynomial-is-nonzero-stronger} instead of Claim \ref{claim:polynomial-is-nonzero}.
\end{proof}

\section{Clustering Part 1: Building Blocks}\label{sec:clustering1}

In order to prove our full result for GMMs, we need to be able to deal with mixtures for which the maximum separation is much larger (superpolynomial) than the minimum separation.  To do this, we will have a recursive clustering procedure where we are able to remove half of the components each time.  Roughly, we find a direction $v$ for which there are two components whose separation along direction $v$ is large (say at least $(\log (k/w_{\min}))^{1/2 + c}$).  We call this a signal direction.

Next, imagine projecting all of our samples onto the direction $v$.  For each component, essentially all of its samples will be in an interval of width $O(\sqrt{\log (k/w_{\min})})$.  Thus, since there are two components that are sufficiently separated, we can find an interval of width $O(\sqrt{\log (k/w_{\min})})$ for which
\begin{itemize}
    \item For at least one component, essentially all of its samples are in this interval
    \item At most half of the components ever generate samples in this interval
\end{itemize}
We then imagine going back to $\R^d$ but only keeping the samples whose projection onto $v$ lies in this interval.  Of course, after this restriction, the distribution is no longer a mixture of Gaussians because there could be a component that is cut in half by the boundary of the interval.  The key observation is that we can project all of the remaining samples (after the restriction) onto the orthogonal complement of $v$.  This new distribution will be a mixture of Gaussians in $d-1$ dimensions with at most half as many components.  Furthermore, because we restricted to an interval of width $O(\sqrt{\log (k/w_{\min})})$, we can argue that this projection onto the orthogonal complement of $v$ does not reduce separations by too much.

The actual clustering algorithm will be more complicated than the outline above for technical reasons.  There are several additional details that we need to deal with such as the fact that we may create components with very small mixing weights when we restrict to samples in an interval.  Also, the procedure for halving the number of remaining components only works when there is some  pair of means whose separation is at least $\poly(\log (k/w_{\min}))$ (which is larger, but polynomially related to the minimum separation).  Thus, after running the halving procedure sufficiently many times, we then have to run a full clustering routine that has guarantees similar to Theorem \ref{thm:main-Poincare}.

\begin{remark}
Note that for general Poincare distributions, this type of approach seems infeasible because after restricting to an interval in direction $v$, projecting onto the orthogonal complement still may not ``fix" all of the components of the mixture be translations of the same distribution again.
\end{remark}

In this section, we introduce the building blocks for our complete clustering algorithm.  The two main components in this section are Lemma \ref{lem:find-signal1} which roughly allows us to, in an arbitrary mixture, find a signal direction along which some two components have separation comparable to the maximum separation and Lemma \ref{lem:find-signal2} which allows us to do full clustering when the maximum separation is at most $\poly(\log (k/w_{\min}))$.

\subsection{Notation and Terminology}

Throughout the proceeding sections, we will use $0 < c < 1$ do denote an arbitrary (small) positive constant that does not change with the other problem parameters.  

%$k,d, 1/w^*$ (where $w^*$ is some bound on the minimum mixing weight).  Note that by the reduction in Section \ref{sec:reduce-dim}, we may assume $d \leq k$ so actually the only relevant parameters are $k$ and $1/w^*$.

\begin{definition}
For a subspace $V \subset \R^d$ and a point $x \in \R^d$, we use $\Proj_V(x)$ to denote the projection of $x$ onto $V$.
\end{definition}
\begin{definition}
For a subspace $V \subset \R^d$ we use $V^{\perp}$ to denote its orthogonal complement.
\end{definition}
\begin{definition}
For a Gaussian with mean $\mu$ and covariance $\Sigma$, we use $N_{V}(\mu, \Sigma)$ to denote its projection onto a subspace $V$. 
\end{definition}

We will now introduce terminology for describing mixtures and how well-behaved they are.

\begin{definition}\label{def:separated-mixture}
We say a GMM $\mcl{M} = w_1N(\mu_1, I) + \dots + w_k N(\mu_k, I)$ is $s$-separated if for all $i \neq j$, $\norm{\mu_i - \mu_j} \geq s$.
\end{definition}

\begin{definition}\label{def:reasonable-mixture}
For a GMM $\mcl{M} = w_1N(\mu_1, I) + \dots + w_k N(\mu_k, I)$ and parameter $w^* > 0$, we say that $\mcl{M}$ is $w^*$-reasonable if the following holds:
\[
\max_{\substack{i , j \\ w_i, w_j \geq w^*}} \norm{\mu_i - \mu_j}  \geq \sqrt{\max_{i,j } \norm{\mu_i - \mu_j}} \,.
\]
We also use the convention that a trivial mixture consisting of one component is reasonable. 
\end{definition}
\begin{remark}
Note the choice that the RHS is a square root is arbitrary.  Any constant power would suffice.
\end{remark}
Roughly, the inequality in the definition says that there are two components whose mixing weights are at least $w^*$ such that the separation between their means is comparable to the maximum separation between any two means.  Note that the notion of comparable is very loose (the two quantities can be off by a square).  

\begin{definition}\label{def:signal-direction}
For a distribution $\mcl{M}$ and parameters $p,\Delta$, we say that a direction, given by a unit vector $v \in \R^d$, is a $(p,\Delta)$-signal direction for $\mcl{M}$ if there is a real number $\theta$ such that 
\begin{align*}
\Pr_{z \sim \mcl{M}}[ v \cdot z \leq \theta - \Delta]  \geq p\\
\Pr_{z \sim \mcl{M}}[ v \cdot z \geq \theta + \Delta ] \geq p \,.
\end{align*}
\end{definition}
Roughly, a direction is a signal direction if the distribution $\mcl{M}$ is nontrivially spread out along that direction.  It will be important to note that we can easily check whether a direction is a signal direction.
\begin{claim}\label{claim:test-is-there-signal}
Let $\mcl{M}$ be a distribution on $\R^d$ and let $v \in \R^d$ be a unit vector.  For parameters $p,\Delta$, given $\poly(d, 1/p)$ samples from $\mcl{M}$, we can distinguish with high probability when $v$ is a $(p,\Delta)$-signal direction and when $v$ is not a $(0.9p, \Delta)$-signal direction. 
\end{claim}
\begin{proof}
We can simply take sufficiently many samples and check if $v$ is a $(0.95p, \Delta)$-signal direction for the empirical distribution on the samples.
\end{proof}

We will also need the following concentration inequalities about tails of a Gaussian.  The first one is standard.
\begin{claim}\label{claim:Gaussian-radius}
Consider a Gaussian $N(\mu, I)$ in $\R^d$.  For any parameter $\beta$, we have
\[
\Pr_{z \sim N(\mu, I)} [ \norm{z - \mu} \leq \sqrt{2d + \beta} ] \geq 1 - 2^{-\beta/8} \,.
\]
\end{claim}

The next inequality is closely connected to the notion of stability (see \cite{diakonikolas2019recent}) that has recently proved very useful in the field of robust statistics.
\begin{claim}[See \cite{diakonikolas2019recent}]\label{claim:eps-tails}
Consider the standard Gaussian $N(0,I)$ in $\R^d$.  Let $\eps > 0$ be some parameter.  Given $n \geq (d/\eps)^8$ samples $z_1, \dots , z_n$ from $N(0,I)$, with probability at least $1 - 2^{-d/\eps}$ the following property holds: for any subset $S \subset [n]$ with $|S| \geq \eps n$, we have
\[
\norm{\frac{1}{|S|}\sum_{i \in S} z_i  } \leq 10\sqrt{ \log 1/\eps} \,.
\]
\end{claim}

\subsection{Clustering Test}
Recall that for arbitrary Poincare distributions, one of the key ingredients was Lemma \ref{lem:main-test-Poincare}, an algorithm for testing whether two samples from a mixture $\mcl{M}$ are from the same component or not.  The analog (with strengthened quantitative bounds) that we will need for GMMs is stated below.
\begin{lemma}\label{lem:main-GMM-test}
Let
\[
\mcl{M} = w_1 N(\mu_1,I) + \dots + w_k N(\mu_k, I)
\]
be a GMM in $\R^d$.  Let $w^*, s, \delta$ be parameters and assume that $s \geq (\log k/(w^* \delta))^{1/2 + c}$ for some $0 < c < 1$.  Also assume that 
\[
\max_{i,j} \norm{\mu_i - \mu_j} \leq \min( (k/(w^*\delta))^2, s^{C} )
\]
for some $C$.  There is an algorithm that takes $n = \poly((kd/(w^*\delta) )^{C/c} )$ samples from $\mcl{M}$ and runs in $\poly(n)$ time and achieves the following testing guarantees:  for a pair of (independent) samples $z \sim N(\mu_i,I) ,z' \sim N(\mu_j, I)$,
\begin{itemize}
    \item If $i = j$ and $w_i \geq w^*$ then with probability $1 - \delta$ the output is $\textsf{accept}$
    \item If $w_i, w_j \geq w^*$ and $\norm{\mu_i - \mu_j} \geq s$ then with probability $1 - \delta$ the output is $\textsf{reject}$
\end{itemize}
where the randomness is over the samples $z,z'$ and the random choices in the algorithm.
\end{lemma}
\begin{proof}
The proof is exactly the same as the proof of Lemma \ref{lem:main-test-Poincare} except we use Lemma \ref{lem:length-of-nonzero-stronger} and Lemma \ref{lem:length-of-zero-stronger} in place of Lemma \ref{lem:length-of-nonzero} and Lemma \ref{lem:length-of-zero} respectively.
\end{proof}

\subsection{Finding a Signal Direction}

As mentioned before, to remove the constraint on the maximum separation, we will need to do recursive clustering.  The key ingredient in the recursive clustering is finding a signal direction.  In the next lemma, we show how to find a signal direction in a reasonable mixture.  Note that if a mixture is not reasonable, then we have no guarantees.  Roughly this is because if there are some components with means $\mu_i, \mu_j$ such that $\norm{\mu_i - \mu_j}$ is very large but whose mixing weights are very small then the few samples from these components will drastically affect all of our estimators.  Fortunately, in the next section (see Claim \ref{claim:exists-reasonable}), we show how to ensure that we always work with a reasonable mixture.  

The main algorithm for finding a signal direction is summarized below.  Roughly, we just randomly take two samples $z,z'$ from $\mcl{M}$.  We then take the mean $\mu$ of samples $z_i$ such that $(z,z_i)$ passes the test in Lemma \ref{lem:main-GMM-test} and the mean $\mu'$ of samples $z_i'$ such that $(z, z_i')$ passes the test in Lemma \ref{lem:main-GMM-test}.  We then test whether the direction $\mu - \mu'$ is indeed a good enough signal direction and otherwise just try again with new samples.

\begin{algorithm}[H]
\caption{{\sc Finding a Signal Direction} }
\begin{algorithmic}
\State \textbf{Input:} Sample access to GMM $\mcl{M}$
\State \textbf{Input:} Parameters $k, w^*$, positive constant $c > 0$
\State \textbf{Input:} Desired separation $\Delta$
\State Take two samples $z, z' \sim \mcl{M}$
\State Take samples $z_1, \dots , z_m , z_1', \dots , z_m'$ from $\mcl{M}$ where $m = (k/w^*)^{10^2}$
\For{i = 1,2, \dots , m}
\State Test pairs $z,z_i$ and $z', z_i'$ using Lemma \ref{lem:main-GMM-test} with parameters $w^*, \delta = (w^*/k)^{10^4}, s = 0.01 \Delta$ 
\EndFor
\State Let $S$ be the set of $i$ such that the pair $z,z_i$ is accepted.  
\State Let $S'$ be the set of $i$ such that the pair $z', z_i'$ is accepted.
\State Compute $\mu = \frac{1}{|S|}\sum_{i \in S}z_i$ and $\mu' = \frac{1}{|S'|}\sum_{i \in S'}z_i'$
\State Set $v = (\mu - \mu')/\norm{\mu - \mu'}$
\State Test if $v$ is a $(0.8w^*, 0.8\Delta )$ signal direction using Claim \ref{claim:test-is-there-signal}
\State Output $v$ if test passes
\end{algorithmic}
\label{alg:find-signal}
\end{algorithm}

\begin{lemma}\label{lem:find-signal1}
 Let $\mcl{M} = w_1N(\mu_1, I) + \dots + w_k N(\mu_k, I)$ be a GMM.  Let $w^*$ be a parameter and assume that $\mcl{M}$ is $w^*$-reasonable and $\norm{\mu_i - \mu_j} \leq O((k/w^*)^2)$ for all $i,j$.  Also assume that there are $i,j$ such that $w_i, w_j \geq w^*$ and $ \norm{\mu_i - \mu_j} \geq  (\log(k/w^*))^{1/2+c}$ for some positive constant $c$. There is an algorithm that takes $ n = \poly((dk/w^*)^{1/c})$ samples from $\mcl{M}$ and $\poly(n)$ runtime and with high probability outputs a unit vector $v$ such that $v$ is a $(w^*/2, \max_{w_i, w_j \geq w^*} \norm{\mu_i - \mu_j}/2)$-signal direction.
\end{lemma}
\begin{proof}
Let 
\[
\Delta = \max_{w_i, w_j \geq w^*} \norm{\mu_i - \mu_j} \,.
\]
It suffices now to find a $(w^*/2, \Delta/2)$-signal direction.  While we technically do not know $\Delta$, we can guess $\Delta$ within a factor of $1.1$ by using a multiplicative grid.  We then try to find a $(0.8w^*, 0.8\Delta)$ signal direction for various guesses of $\Delta$ and test them using Claim \ref{claim:test-is-there-signal} and output the one for the largest $\Delta$ that succeeds.  Thus, we can essentially assume that we know $\Delta$ up to a factor of $1.1$.  
\\\\
%Now we perform the following steps.
%\begin{enumerate}
    %\item Take two samples $z,z' \sim \mcl{M}$
    %\item Take samples $z_1, \dots , z_m, z_1', \dots , z_m'$ from $\mcl{M}$ where $m = (k/w^*)^{10^2}$ 
    %\item For each $i$, test the pairs $z,z_i$ and $z', z_i'$ using Lemma \ref{lem:main-GMM-test} with parameters $\delta = (w^*/k)^{10^4}, s = 0.01 \Delta$ (note this is valid because $\mcl{M}$ is $w^*$-reasonable)
    %\item Let $S$ be the set of $i$ such that the pair $z,z_i$ is accepted.  Let $S'$ be the set of $i$ such that the pair $z', z_i'$ is accepted.
    %\item Compute $\mu = \frac{1}{|S|}\sum_{i \in S}z_i$ and $\mu' = \frac{1}{|S'|}\sum_{i \in S'}z_i'$
    %\item Output $v = (\mu - \mu')/\norm{\mu - \mu'}$
%\end{enumerate}
Now we can just run Algorithm \ref{alg:find-signal} repeatedly.  It suffices to prove that with some inverse polynomial probability, the output $v$ is a $(0.8w^*, 0.8\Delta)$-signal direction.  We can then just repeat polynomially many times and test each output using Claim \ref{claim:test-is-there-signal} to guarantee success with high probability.  Let $i,i'$ be indices such that $\norm{\mu_i - \mu_i'} = \Delta$.  With at least $(1/w^*)^2$ probability, $z$ is drawn from $N(\mu_i, I)$ and $z'$ is drawn from $N(\mu_i', I)$.  Now using the guarantee of Lemma \ref{lem:main-GMM-test} and union bounding (note this is valid because $\mcl{M}$ is $w^*$-reasonable), with at least $1 - (w^*/k)^{10^2}$ probability, we have the following properties
\begin{itemize}
    \item For all $i \in S$, $z_i$ is a sample from $N(\mu_j, I)$ where $\norm{\mu_j - \mu_i} \leq 0.01\Delta$
    \item For all $i \in S'$, $z_i$ is a sample from $N(\mu_j, I)$ where $\norm{\mu_j - \mu_{i'}} \leq 0.01\Delta$
    \item $|S|, |S'| \geq 0.9w^*m$
\end{itemize}
However, combining these properties with Claim \ref{claim:eps-tails} implies that with high probability $\norm{\mu - \mu_i} , \norm{\mu' - \mu_{i'}}\leq 0.02 \Delta$.  This then implies that $v$ has correlation $0.9$ with the unit vector in the direction $\mu_i - \mu_{i'}$ which clearly implies that it is a $(0.8w^*, 0.8\Delta)$-signal direction.  Overall, the success probability of a single trial is at least $0.9 \cdot (1/w^*)^2$ so repeating over polynomially many trials guarantees that we succeed with high probability.
\end{proof}

\subsection{Full Clustering with Bounded Maximum Separation}

Similar to Theorem \ref{thm:main-Poincare}, we can do full clustering if the maximum separation is polynomially bounded in terms of the minimum separation.  We have a slightly more technical requirement here that there may be some components with very small mixing weights but we only need to recover the means of the components with substantial mixing weights.  This result will be used as the final step in our complete clustering algorithm when we have reduced to a submixture where the maximum separation is sufficiently small and then we can just fully cluster the remaining components.

\begin{lemma}\label{lem:find-signal2}
 Let $\mcl{M} = w_1N(\mu_1, I) + \dots + w_k N(\mu_k, I)$ be a GMM.  Let $c > 0$ be a positive constant and let $w^*$ be a parameter that we are given.  Assume that $\mcl{M}$ is $s$-separated where $s = (\log(k/w^*))^{1/2 + c}$ for some positive constant $c$.  Also assume that $ \max_{i,j} \norm{\mu_i - \mu_j} \leq (\log(k/w^*))^{4}$.  Given $ n = \poly((dk/w^*)^{1/c})$ samples from $\mcl{M}$ and $\poly(n)$ runtime, there is an algorithm that with high probability outputs a set of means $\{\wt{\mu_1}, \dots , \wt{\mu_r} \} $ such that the following properties hold:
 \begin{itemize}
     \item $r \leq k$
     \item For all $i \in [k]$ such that $w_i \geq w^*$, there is some $j \in [r]$ such that $\norm{\mu_i - \wt{\mu_j}} \leq 0.1$
     \item We have $\norm{\wt{\mu_i} - \wt{\mu_j}} \geq s/2$ for all $i \neq j$
 \end{itemize}
\end{lemma}
\begin{proof}
The proof will be very similar to the proof of Theorem \ref{thm:main-Poincare} except using the test in Lemma \ref{lem:main-GMM-test}.  %Recall by the reduction in Section \ref{sec:reductions} that we may assume $d \leq k$.  
Now we will do the following process to estimate the means of the components.
\begin{enumerate}
    \item Draw a sample $z \sim \mcl{M}$
    \item Take $m = (k/w^* )^{10^2}$ samples $z_1, \dots , z_m \sim \mcl{M}$
    \item Use Lemma \ref{lem:main-GMM-test} with parameters $w^* = (w^*/k)^{10}, \delta = (w^*/k)^{10^4}, s$ to test the pair of samples $z,z_i$ for all for all $i \in [m]$
    \item Let $S \subset [m]$ be the set of all $i$ that are {\sc Accepted} and compute $\mu = \frac{1}{|S|}\sum_{i \in S} z_i$
\end{enumerate}  
Note that if $z$ is a sample from $N(\mu_i, I)$ for some $i$ with $w_i \geq (w^*/k)^5$, then with $1 - (w^*/k)^{10^2}$ probability, the procedure returns $\mu$ such that $\norm{\mu - \mu_i} \leq 0.01$.  This is because by the guarantees of Lemma \ref{lem:main-GMM-test}, with probability at least $1 - (w^*/k)^{10^3}$ the test will accept all samples from among $z_1, \dots , z_m$ that are from the component $N(\mu_i, I)$.  Also, the only other samples that are accepted must be from components $N(\mu_j, I)$ with $w_j \leq (w^*/k)^{10}$.  With high probability, among $z_1, \dots , z_m$, there will be at least $0.9w_i m \geq 0.9(w^*/k)^5m$ samples from the component $N(\mu_i, I)$ and at most $2 k (w^*/k)^{10} m \leq (w^*/k)^9 m$ from other components $N(\mu_j, I)$ with mixing weight smaller than $(w^*/k)^{10}$.  Since we have a bound on the maximum separation $\norm{\mu_i - \mu_j} \leq (\log (k/w^*))^4$, with high probability, the mean $\mu$ of all of these samples satisfies $\norm{\mu - \mu_i} \leq 0.01$.

Now we repeat the procedure in steps $1-4$ for $l = (k/w^*)^{10^2}$ independent samples $z \sim \mcl{M}$.  This gives us a list of means say $S = \{ \wh{\mu_1}, \dots , \wh{\mu_l} \}$.  The previous argument implies that most of these estimates will be close to $\mu_i$ for some $i$ with $w_i \geq (w^*/k)^5$ and furthermore that all such components will be represented.  As in Theorem \ref{thm:main-Poincare}, to ensure that with high probability we output exactly one estimate corresponding to each true mean and no extraneous estimates, we do a sort of majority voting.  

We will inspect the estimates in $S$ one at a time and decide whether to output them or not.  Let $T$ be the set of estimates that we will output.  Note that $T$ is initially empty.  Now for each $i$, let $S_i$ be the subset of $\{ \wh{\mu_1}, \dots , \wh{\mu_l} \}$ consisting of all means with $\norm{\wh{\mu_j} - \wh{\mu_i}} \leq 0.02$.  If $|S_i| \geq 0.9 w^* l$ and $\wh{\mu_i}$ is not within $0.1$ of any element of $T$ then add $\wh{\mu_i}$ to $T$.  Otherwise do nothing.  At the end we output everything in $T$.  

We now argue that with high probability, the set $T$ satisfies the desired properties.  First, note that with high probability, all elements of $T$ must be within $0.05$ of one of the true means $\mu_i$ since otherwise, there would not be enough elements in the set $S_j$.   Also it is clear that there can be at most one element $\wh{\mu_j} \in T$ corresponding to each true mean $\mu_i$.  It remains to argue that with high probability, for all  $i \in [k]$ such that $w_i \geq w^*$, there must be some element in $T$ within $0.05$ of $\mu_i$.  This is because for such an $i$, with high probability there will be at least $0.9w_il$ elements $\wh{\mu_j}$  in $S$ such that $\norm{\wh{\mu_j} - \mu_i} \leq 0.01$.  Thus, if there are no elements already in $T$ that are close to $\mu_i$, then one of these $\wh{\mu_j}$ will be added to $T$.  Overall, we have shown that with high probability, the set $T$ satisfies the desired properties and we are done.
\end{proof}
\section{Clustering Part 2: Recursive Clustering}\label{sec:clustering2}
In this section, we put together the building blocks from Section \ref{sec:clustering1} in our complete clustering algorithm and complete the proof of Theorem \ref{thm:main-GMM}.

\subsection{Clustering Checkers}

First, we introduce the concept of a checker.  This will ease notation for later on when we need to consider various restrictions of a mixture involving restricting to samples whose projection onto a subspace $V$ is close to a certain point $p \in V$.  

\begin{definition}[Checker]\label{def:checker}
A checker, denoted $(V,p,r)$ consists of a subspace $V \subset \R^d$, a point $p \in V$ and a positive real number $r$.
\end{definition}

\begin{definition}
We say that a checker $(V,p,r)$ contains a point $x$ if $\norm{\Proj_V(x) - p} \leq r$.  We write $\chk_{V,p,r}(x) = 1$ if the checker contains the point $x$ and $\chk_{V,p,r}(x) = 0$ otherwise.
\end{definition}

\begin{definition}
Given a distribution $\mcl{M}$ and a checker $(V,p,r)$, we may take samples from $\mcl{M}$, delete all of them that do not lie inside the checker $(V,p,r)$, and then project the remaining samples onto $V^{\perp}$.  We call the resulting distribution the reduction of $\mcl{M}$ by the checker $(V,p,r)$ and denote it $\red_{V,p,r}(\mcl{M})$.
\end{definition}

\subsection{Basic Properties}
The key observation is that reducing a mixture of Gaussians by a checker $(V,p,r)$ results in a new mixture with the same components (projected onto $V^{\perp}$) but with different mixing weights.
\begin{claim}\label{claim:reducing-mixture}
Assume that we have a checker $(V,p,r)$ and let $a$ be the dimension of $V$.  Given a GMM $\mcl{M} = w_1N(\mu_1,I) + \dots + w_k N(\mu_k, I)$, the reduction of $\mcl{M}$ is 
\[
\red_{V,p,r}(\mcl{M}) = \frac{\sum_{i = 1}^k w_i \Pr_{z \sim N(\mu_i, I)}[ \chk_{V,p,r}(z) = 1] N_{V^{\perp}}(\mu_i, I  ) }{\sum_{i = 1}^k w_i \Pr_{z \sim N(\mu_i, I)}[ \chk_{V,p,r}(z) = 1]} \,.
\]
\end{claim}
\begin{proof}
This follows immediately from the fact that for a sample $z$ from a standard Gaussian $N(0,I)$, for any subspace $V$, the projections  $\Proj_V(z)$ and $\Proj_{V^{\perp}}(z)$ are independent and distributed as standard Gaussians in the respective subspaces.
\end{proof}

In light of Claim \ref{claim:Gaussian-radius}, reducing a GMM by a checker $(V,p,r)$ where $V$ has dimension $a$ essentially removes all components whose means $\mu_i$ satisfy $\norm{\Proj_V(\mu_i) -  p  } \geq r +  \omega(\sqrt{a + \log k})$.  We now make this notion more formal by defining a truncation of a reduced mixture that involves deleting such components.  We then argue that the truncation only affects the overall distribution by a negligible amount.

\begin{definition}
Let $\mcl{M} = w_1N(\mu_1,I) + \dots + w_k N(\mu_k, I)$ be a GMM and $\theta$ be some parameter.  For a checker $(V,p,r)$, define the truncated reduction of $\mcl{M}$ as   
\[
\red_{V,p,r}^{(\theta)}(\mcl{M}) = \frac{\sum_{i \in S} w_i \Pr_{z \sim N(\mu_i, I)}[ \chk_{V,p,r}(z) = 1] N_{V^{\perp}}(\mu_i, I  ) }{\sum_{i \in S} w_i \Pr_{z \sim N(\mu_i, I)}[ \chk_{V,p,r}(z) = 1]}
\]
where $S$ is defined as the set of $i \in [k]$ such that $\mu_i$ is in the checker $(V,p, r + \theta)$.  We say that the set $S$ is the set of relevant components in the truncation $\red_{V,p,r}^{(\theta)}(\mcl{M})$.
\end{definition}

Combining Claim \ref{claim:Gaussian-radius} and Claim \ref{claim:reducing-mixture}, we deduce that we can truncate after reducing by a checker while changing the distribution by only a negligible amount.
\begin{corollary}\label{coro:truncate-prob}
Let $\mcl{M} = w_1N(\mu_1,I) + \dots + w_k N(\mu_k, I)$ be a GMM.  Let $w^*, \delta$ be some parameters and let $\theta \geq (\log (k/(w^* \delta)))^{(1  +c)/2}$ for some positive constant $c > 0$.  Let $(V,p,r)$ be a checker where $V$ has dimension at most $(\theta/10)^2$.  Assume that for some $i$, we have $w_i \geq w^*$ and $\chk_{V,p,r - \theta }(\mu_i) = 1$.  Then 
\[
d_{\TV}\left( \red_{V,p,r}^{(\theta)}(\mcl{M}) , \red_{V,p,r}(\mcl{M}) \right) \leq 2^{- (\log (k/(w^* \delta)))^{1 + 0.1c} } \,.
\]
\end{corollary}
\begin{proof}
Note that by Claim \ref{claim:Gaussian-radius}, for any $\mu \in \R^d$ we have
\[
\Pr_{z \sim N(\mu, I)}[ \norm{ \Proj_V(z - \mu)} \geq 0.5\theta ] \leq 2^{-0.01 \theta^2}  \,.
\]
In particular, for any $\mu_j$ that is not in the checker $(V,p,r + \theta)$, the probability that a sample from $N(\mu_j, I)$ lands in $(V,p,r)$ is at most $2^{-0.01 \theta^2}$.  Also we assume that there is some $i$ such that $\mu_i$ is in $(V,p, r - \theta)$ and for this component $N(\mu_i, I)$, the probability that a sample lands in $(V,p,r)$ is at least $1 - 2^{-0.01 \theta^2}$.  Combining these observations gives the desired inequality.
\end{proof}

Note that we can simulate samples from $\red_{V,p,r}(\mcl{M})$ (using samples from $\mcl{M}$) and the above implies that this is essentially equivalent to simulating samples from $\red_{V,p,r}^{(\theta)}(\mcl{M})$ since their TV distance is negligible.  One more important claim that we will need is that if a checker contains the mean of one of the components, then by adjusting the radius slightly, we can find a truncation for which the resulting mixture is reasonable (recall Definition \ref{def:reasonable-mixture}).  This will allow us to apply Lemma \ref{lem:find-signal2} to find a new signal direction in $V^{\perp}$.

\begin{claim}\label{claim:exists-reasonable}
Let $\mcl{M} = w_1N(\mu_1,I) + \dots + w_k N(\mu_k, I)$ be a GMM.  Let $w^*$ be some parameter and let $\theta \geq  (\log (k/w^*))^{(1 + c)/2}$ for some positive constant $c > 0$.  Assume that 
\begin{itemize}
    \item $\max \norm{\mu_i - \mu_j} \leq O((k/w^*)^2)$
    \item The mixture $\mcl{M}$ is $\theta \cdot (\log (k/w^*))^{0.1 c/2}$-separated
    \item $w_i \geq w^*$ for all $i$
\end{itemize} 
Let $V$ be a subspace of dimension at most $(\theta/10)^2$ and $p$ be a point in $V$.  Assume that for some $i$, the checker $(V,p,\theta)$ contains $\mu_i$.  Then for a random integer $\gamma$ chosen from among $\{1, \dots , \lceil 10^4 \log \log (k/w^*) \rceil \}$, with probability at least $1/2$, the mixture 
\[
\red_{V,p,\gamma \theta}^{  (\theta)}(\mcl{M})
\]
is $0.9w^*$-reasonable.
\end{claim}
\begin{proof}
For any real number $l$, define $\alpha_{l}$ as follows.
\[
\alpha_l = \max_{\substack{ \chk_{V,p,l}( \mu_i) = 1  \\ \chk_{V,p,l}(\mu_j ) = 1  }} \norm{\Proj_{V^{\perp}}(\mu_i - \mu_j)} \,. 
\]
Now for any $\gamma$, the only way that
\[
\red_{V,p, \gamma \theta}^{  (\theta)}(\mcl{M})
\]
is not reasonable is if 
\begin{equation}\label{eq:get-squared}
\alpha_{(\gamma-1)\theta}^2  \leq \alpha_{(\gamma + 1)\theta} \,.
\end{equation}
This is because by Claim \ref{claim:Gaussian-radius}, for all means $\mu_i$ that are in $(V,p,(\gamma-1)\theta)$, essentially all of the samples from $N(\mu_i, I)$ are contained in $(V,p, \gamma \theta)$.  However, we now consider the sequence 
\[
\alpha_{ \theta}, \alpha_{2 \theta} , \dots \,. 
\]
We can lower bound the first nonzero element of the sequence as follows.  If there are two distinct means $\mu_i, \mu_j$ both contained in $(V, p ,  \gamma \theta )$ for $\gamma \leq \lceil 10^4 \log \log (k/w^*) \rceil$, we must have
\[
\alpha_{\gamma \theta}  \geq \norm{\mu_i - \mu_j} - \norm{\Proj_V(\mu_i - \mu_j)} \geq  \theta \cdot (\log (k/w^*))^{0.1c/2} - 2 \gamma \theta \geq 2 \,.
\]
Thus, since $\alpha_l$ is upper bounded by $O((k/w^*)^2)$, the condition in (\ref{eq:get-squared}) can fail for at most half of the choices of $\gamma$ and we are done.
\end{proof}

We will need one more preliminary result.  It simply states that we can correctly cluster any sample with high probability if we are given a candidate set of means $S = \{ \wt{\mu_1}, \dots , \wt{\mu_t } \}$ that are all separated and such that there is one that is close to each true mean.

\begin{claim}\label{claim:cluster-using-means}
Let $\N(\mu_1, I) , \dots ,  N(\mu_l. I)$ be some unknown Gaussians such that $\norm{\mu_i - \mu_j} \geq s$ for all $i \neq j$ where $ s = (\log l)^{( 1 + c)/2}$ and $c > 0$ is some constant.  Assume we are given a candidate set of means $S = \{ \wt{\mu_1}, \dots , \wt{\mu_r } \}$ such that 
\begin{itemize}
    \item $r \leq l$
    \item For all $i \in [l]$, there is some $f(i) \in [r]$ such that $\norm{\wt{\mu_{f(i)}} - \mu_i} \leq 0.1$
    \item For all distinct $j_1, j_2 \in [r]$, $\norm{\wt{\mu_{j_1}} - \wt{\mu_{j_2}}}  \geq 0.5s$
\end{itemize}
Then there is an efficient algorithm that with probability at least $1 - 2^{-0.01s^2 }$, given a sample from $N(\mu_i, I)$ for any $i \in [k]$, returns $f(i)$.
\end{claim}
\begin{proof}
For all $j_1,j_2 \in [r]$ with $ j_1\neq j_2$, define
\[
v_{j_1j_2} = \frac{\wt{\mu_{j_1}} - \wt{\mu_{j_2}}}{\norm{\wt{\mu_{j_1}} - \wt{\mu_{j_2}}}} \,.
\]
Given a sample $z$, we find a $j$ such that for all $j_1, j_2 \in [r]$, we have
\[
|v_{j_1j_2} \cdot (z - \wt{\mu_j})| \leq 0.1 s \,.
\]
If $z \sim N(\mu_i, I)$, then with probability at least $1 - 2^{-0.02s^2 }$, setting $j = f(i)$ clearly satisfies the above.  Also, by the assumption that $ \norm{\wt{\mu_{j_1}} - \wt{\mu_{j_2}}}  \geq 0.5s$ for all distinct $j_1, j_2$, it is clear that with probability at least  $1 - 2^{-0.02s^2 }$ that no other $j' \in [r]$ will satisfy the above.  This completes the proof.
\end{proof}

\subsection{Putting Everything Together}

We can now put everything together and describe our complete clustering algorithm.  At a high level there will be two phases.  

In the first phase, we will keep track of a subspace $V$ and point $p \in V$ and keep refining it.  In particular, in each step we will add a dimension to $V$ to get a higher dimensional subspace $V'$ and we will compute a new point $p' \in V'$.  Our goal in each step will be to maintain (roughly) the following properties
\begin{itemize}
    \item For a certain radius $r = O( (\log (k/w^*))^{(1 + c)/2})$, the checker $(V,p,r)$ always contains one of the true means $\mu_i$
    \item If there are two means $\mu_i, \mu_j$ in $(V,p,r)$ with separation at least $0.1\log^4 (k/w_{\min})$, then with at least $0.1$ probability, the refinement $V',p'$ satisfies that $(V',p',r')$ contains at most half as many means $\mu_i$ as $(V,p,r')$ for some larger radius $r'$.
\end{itemize}
The first guarantee is not difficult to achieve as we simply need to check that there are enough samples in the checker.  To achieve the second guarantee, we rely on Lemma \ref{lem:find-signal1} to find a signal direction and add this direction to $V$ to get $V'$.  We then argue that because this new direction is a signal direction, the set of true means in $(V,p,r)$ can be split into two parts along this new direction and one of these parts will have at most half as many.  The guarantees are stated formally in Lemma \ref{lem:ball-recursion}.

We will run the first phase sufficiently many times that with high probability, at  some point, we must have $V,p$ such that all means in $(V,p,r)$ have separation at most $\log^4 (k/w_{\min})$.  We can actually test this condition (with some slack) using Lemma \ref{lem:test-max-separation}.  If the test passes, in the second phase, we can simply use Lemma \ref{lem:find-signal2} to fully cluster the remaining submixture obtained by reducing $\mcl{M}$ by $(V,p,r)$.  This allows us to learn one of the components of the mixture.  In fact, we can obtain a stronger guarantee and actually identify all of the samples from this component.  See Lemma \ref{lem:final-step} for more details.  Once we have done this, we can remove these samples and recurse on the remaining  mixture.  This completes the entire clustering algorithm.  Below is an outline that summarizes our algorithm.

\begin{algorithm}[H]
\caption{{\sc Complete Clustering Algorithm (Outline)} }
\begin{algorithmic}
\State Initialize $V = 0$ (i.e. $0$-dimensional subspace), $p = 0$
\State Set $r = O((\log(k/w_{\min}))^{(1 + c)/2}  ) $ 
\For{$j = 1,2, \dots ,(\log(k/w_{\min}))^{1 + 0.1c}$ }
\State Refine $(V,p) \leftarrow (V',p')$ using Lemma \ref{lem:ball-recursion}
\State Test if means in $(V,p,r)$ have max-separation at most $ \log^4 (k/w_{\min})$ using  Lemma \ref{lem:test-max-separation}
\State Break if above test passes
\EndFor
\State Fully cluster samples in $(V,p,r)$ and identify samples from one component $N(\mu_i, I)$ (see Lemma \ref{lem:final-step})
\State Remove samples from  $N(\mu_i, I)$ and recurse on remaining
\end{algorithmic}
\label{alg:full-clustering-outline}
\end{algorithm}

The next lemma formalizes the guarantees of the refinement step.

\begin{lemma}\label{lem:ball-recursion}
Let $\mcl{M} = w_1N(\mu_1, I) + \dots + w_k N(\mu_k, I)$ be a GMM.  Let $w^* > 0$ be a parameter and $c > 0$ be a positive constant.  Assume that $\mcl{M}$ is $s$-separated where $s = (\log(k/w^*))^{1/2 + c}$ and satisfies $w_i \geq w^*$ for all $i$.  Also assume that $\max \norm{\mu_i - \mu_j} \leq O((k/w^*)^2)$. 

Assume that we are given a subspace $V$ and a point $p \in V$ where the dimension of $V$ is $a < (\log (k/w^*))^{1 + 0.1 c}$.  Assume that the checker $(V,p, 10(\log (k/w^*))^{(1 + c)/2 })$ contains some $\mu_i$.  Also let $C$ be the number of indices $i$ such that $\mu_i$ is contained in the checker 
\[
\left(V,p, (\log (k/w^*))^{2}((\log (k/w^*))^{1 + 0.1 c} - a ) \right) \,.
\]
Then there exists an algorithm that takes $n = \poly((dk/w^*)^{1/c})$ samples and $\poly(n)$ runtime and returns a subspace $V'$ and a point $p' \in V'$ with the following properties
\begin{itemize}
    \item $V'$ has dimension $a+1$ and is obtained by adding one orthogonal direction to $V$
    \item The checker 
    \[
    \left(V',p', (\log (k/w^*))^{2}((\log (k/w^*))^{1 + 0.1 c} - a - 1 ) \right) 
    \]
    contains at most $C$ of the $\mu_i$.
    \item With high probability, the checker $(V',p', 10(\log (k/w^*))^{(1 + c)/2}) $ contains some $\mu_i$
    \item If there are two $\mu_{i_1}, \mu_{i_2}$ contained in the checker $(V,p, (\log (k/w^*))^{(1 + 1.1 c)/2 })$ such that 
    \[
    \norm{\mu_{i_1} - \mu_{i_2}}  \geq 0.1(\log (k/w^*))^{4}
    \]
    then with probability at least $0.1$, the checker 
    \[
    \left(V',p', (\log (k/w^*))^{2}((\log (k/w^*))^{1 + 0.1 c} - a - 1 ) \right) 
    \]
    contains at most $C/2$ of the $\mu_i$.
\end{itemize}
\end{lemma}
\begin{proof}
We first focus on the last guarantee as it will not be difficult to maintain the first three guarantees. Let $\theta = (\log (k/w^*))^{(1 + c)/2}$ and $\beta = (\log (k/w^*))^{(1 + 1.1c)/2}$.  Choose some $\gamma$ randomly from $\{1, \dots , \lceil 10^4 \log \log (k/w^*) \rceil \}$.  Now by Claim \ref{claim:exists-reasonable}, with probability at least $1/2$, the mixture
\[
\red_{V,p, \beta + \gamma \theta}^{  (\theta)}(\mcl{M})
\]
is $0.9w^*$-reasonable.  Also note that the two components $\mu_{i_1}$ and $\mu_{i_2}$ promised in the last condition must satisfy
\[
\norm{\Proj_{V^{\perp}}(\mu_{i_2}) - \Proj_{V^{\perp}}(\mu_{i_1})} \geq 0.09 ( \log (k/w^*))^4  \,.
\]
and furthermore, their respective mixing weights in $\red_{V,p, \beta + \gamma \theta}^{  (\theta)}(\mcl{M})$ are at least $0.9w^*$.  Now by Corollary \ref{coro:truncate-prob} we can, with high probability, simulate a polynomial number of samples from $\red_{V,p, \beta + \gamma \theta}^{  (\theta)}(\mcl{M})$ by taking samples from $\red_{V,p,\beta + \gamma \theta}(\mcl{M})$ (which we can simulate using samples from $\mcl{M}$).  Also note that since the checker 
\[
(V,p, 10(\log (k/w^*))^{(1 + c)/2 })
\]
contains some $\mu_i$, samples from $\mcl{M}$ are contained in the checker $(V,p,\beta + \gamma\theta)$ with probability at least $0.9w^*$.  Thus we can apply Lemma \ref{lem:find-signal1} and with high probability we can find a signal direction.  Note this signal direction is a unit vector $v \in V^{\perp}$ such that it is a $(0.4w^*, 0.04( \log (k/w^*))^4)$-signal direction for the distribution $\red_{V,p, \beta + \gamma \theta}^{  (\theta)}(\mcl{M})$.  Note that we can check this last condition so if we repeat the above polynomially many times (for random choices of $\gamma$), we can guarantee that we have found such a signal direction.

 We let $V' = V + v$ (i.e. $V'$ is obtained by adding $v$ to the span of $V$).  Note that by the assumption that $v$ is a signal direction, we can find a real number $t$ such that there must be two components $\mu_i$ and $\mu_j$ that are relevant in $\red_{V,p,\beta  +\gamma \theta}^{  (\theta)}(\mcl{M})$ such that 
\begin{align*}
 t - v \cdot \mu_i  \geq 0.01`( \log (k/w^*))^4  \\
 v \cdot \mu_j - t \geq 0.01( \log (k/w^*))^4 \,.
\end{align*}
Now to find the new center $p'$, we do the following.  Draw a fresh set of $\poly(dk/w^*)$ samples from $\mcl{M}$.  For each sample, we keep it if and only if it is contained in the checker $(V,p, \beta + (\gamma + 2)\theta)$.  We say a sample $z$ is good if at least $0.9w^*$-fraction of other samples $z'$ satisfy 
\[
\norm{ \Proj_{V'}(z') - \Proj_{V'}(z)} \leq (\log (k/w^*))^{(1 + c)/2} \,.
\]
Note that with high probability, we will be able to find two samples $z_1,z_2$ that are both good and such that $|v \cdot z_1 - v \cdot z_2| \geq 0.01( \log (k/w^*))^4$ (this will happen as long as we take a sample from $\mu_i$ and a sample from $\mu_j$ for the two separated components promised in the previous paragraph).  Now with $1/2$ probability take $p' = z_1$ and with $1/2$ probability take $p' = z_2$.  Note that the checkers 
\begin{align*}
\left(V',z_1, (\log (k/w^*))^{2}((\log (k/w^*))^{1 + 0.1 c} - a - 1 ) \right) \\
\left(V',z_2, (\log (k/w^*))^{2}((\log (k/w^*))^{1 + 0.1 c} - a - 1 ) \right)
\end{align*}
are disjoint since $z_1,z_2$ are sufficiently separated along direction $v$.  Also,
\[
\norm{\Proj_V(z_1) - p}, \norm{\Proj_V(z_2) - p} \leq \beta + (\gamma + 2)\theta
\]
and thus both of the above checkers are contained in the checker
\[
\left(V,p, (\log (k/w^*))^{2}((\log (k/w^*))^{1 + 0.1 c} - a ) \right) \,.
\]
Thus, with $1/2$ probability, the number of means contained in the new checker is at most $C/2$.  It remains to verify the first three conditions.  The first is trivial.  The second is also trivial.  To see why the third is true, note that since $z_1$ is good, with high probability, there must be some $\mu_{i_1}$ such that $\norm{\Proj_{V'}(\mu_{i_1}) - \Proj_{V'}(z_1)} \leq 2(\log (k/w^*))^{(1 + c)/2}$ by Claim \ref{claim:Gaussian-radius} and similar for $z_2$.  This immediately implies the third condition.
\end{proof}

Next, we show that we can actually check the termination condition (with some slack), that the maximum separation of any two means in $(V,p,r)$ is at most $O(\log^4 (k/w^*))$ where $r = O( (\log (k/w^*))^{(1 + c)/2})  $.

\begin{lemma}\label{lem:test-max-separation}
Let $\mcl{M} = w_1N(\mu_1, I) + \dots + w_k N(\mu_k, I)$ be a GMM.  Let $w^* > 0$ be a parameter and $c > 0$ be a positive constant.  Assume that $\mcl{M}$ is $s$-separated where $s = (\log(k/w^*))^{1/2 + c}$ and satisfies $w_i \geq w^*$ for all $i$.  Also assume that $\norm{\mu_i - \mu_j} \leq O((k/w^*)^2)$ for all $i,j$.  Say that we are given a subspace $V$ and a point $p \in V$ where the dimension of $V$ is $a < (\log (k/w^*))^{1 + 0.1 c}$.  Assume that the checker $(V,p, 10(\log (k/w^*))^{(1 + c)/2 })$ contains some $\mu_i$.  Then there is an algorithm that takes $n = \poly((dk/w^*)^{1/c})$ samples and runs in $\poly(n)$ time, and with high probability, 
\begin{itemize}
    \item Outputs {\sc Reject} if there are $i,j$ such that $\mu_i$ and $\mu_j$ are both contained in the checker 
    \[
    (V,p, 20(\log (k/w^*))^{(1 + c)/2 })
    \]
    and that $\norm{\mu_i - \mu_j}  \geq (\log (k/w^*))^{4}$. 
    \item Outputs {\sc Accept} if for all $i,j$ such that $\mu_i$ and $\mu_j$ are both contained in the checker 
    \[
    (V,p,  (\log (k/w^*))^{(1 + 1.1c)/2 })
    \]
    we have that $\norm{\mu_i - \mu_j}  \leq  0.1 (\log (k/w^*))^{4}$. 
\end{itemize}
\end{lemma}
\begin{proof}
Let $\theta = (\log (k/w^*))^{(1 + c)/2}$.  We consider $\red_{V,p, (30 + \gamma) \theta}^{(\theta)}(\mcl{M})$ for all $\gamma = \{1, \dots , \lceil 10^4 \log \log (k/w^*)\rceil  \}$.  We simulate samples from $\red_{V,p, (30 + \gamma) \theta}^{(\theta)}(\mcl{M})$ using samples from $\red_{V,p, (30 + \gamma) \theta}(\mcl{M})$ (which can be simulated using samples from $\mcl{M}$), which by Corollary \ref{coro:truncate-prob} is equivalent with high probability for polynomially many samples.  We attempt to find a signal direction in this reduced mixture for each $\gamma$ using Lemma \ref{lem:find-signal1}.  If for any $\gamma$, we find a direction that we can check is a $(0.4w^*, 0.4 (\log (k/w^*))^{4} )$-signal direction using Claim \ref{claim:test-is-there-signal},  then we output {\sc Reject}.  Otherwise we output {\sc Accept}.  To see why this works, first note that clearly we will always output {\sc Accept} when we are supposed to accept.  Now to see that we reject when we are supposed to reject, consider the choice of $\gamma$ guaranteed by Claim \ref{claim:exists-reasonable} for which the mixture $\red_{V,p, (30 + \gamma) \theta}^{(\theta)}(\mcl{M})$ is reasonable.  For this choice of $\gamma$, by the guarantees of Lemma \ref{lem:find-signal1}, we will find a $(0.4w^*, 0.4 (\log (k/w^*))^{4} )$-signal direction with high probability and thus we will {\sc Reject}. 
\end{proof}

Finally, we show how to do the final step where we do full clustering and isolate samples from a single component.

\begin{lemma}\label{lem:final-step}
Let $\mcl{M} = w_1N(\mu_1, I) + \dots + w_k N(\mu_k, I)$ be a GMM.  Let $w^* > 0$ be a parameter and $c > 0$ be a positive constant.  Assume that $\mcl{M}$ is $s$-separated where $s = (\log(k/w^*))^{1/2 + c}$ and satisfies $w_i \geq w^*$ for all $i$.  Assume that we are given a subspace $V$ and a point $p \in V$ where the dimension of $V$ is $a < (\log (k/w^*))^{1 + 0.1 c}$.  Assume that the checker $(V,p, 10(\log (k/w^*))^{(1 + c)/2 })$ contains some $\mu_i$.  Also assume that for all $i,j$ such that $\mu_i$ and $\mu_j$ are both contained in the checker $(V,p, 20(\log (k/w^*))^{(1 + c)/2 })$, we have that $\norm{\mu_i - \mu_j}  \leq (\log (k/w^*))^{4}$.  Then there is an algorithm that takes $n = \poly((dk/w^*)^{1/c})$ samples and runs in $\poly(n)$ time, and with high probability, returns a test (which we denote $\textsf{test}$) that can be computed in $\poly(dk)$ time and has the following properties
\begin{itemize}
    \item There is some $i$ such that for a random sample $z \sim N(\mu_i, I)$, with high probability $\textsf{test}(z) = \textsf{accept}$
    \item For all other $i' \neq i$, for a random sample $z \sim N(\mu_i, I)$, with high probability $\textsf{test}(z) = \textsf{reject}$
\end{itemize}
\end{lemma}
\begin{proof}
Let $\theta =  (\log (k/w^*))^{(1+c)/2}$.  Consider the truncated reduced mixture
\[
\red_{V,p,19\theta}^{(\theta)}(\mcl{M}) \,.
\]
By Corollary \ref{coro:truncate-prob}, we can simulate a polynomial number of samples from this mixture by instead taking samples from $\red_{V,p,19\theta}(\mcl{M})$ (which we can simulate using samples from $\mcl{M}$).  Now, by assumption, all relevant components $\mu_i, \mu_j$ in $\red_{V,p,19\theta}^{(\theta)}(\mcl{M})$ must have
\[
0.9(\log(k/w^*))^{1/2 + c} \leq \norm{\Proj_{V^{\perp}}(\mu_i) - \Proj_{V^{\perp}}(\mu_j)} \leq  (\log (k/w^*))^{4} \,.
\]
Thus, we can apply Lemma \ref{lem:find-signal2} to this mixture with parameter $w^* = (w^*/k)^{10}$.  With high probability, we obtain a list of candidate means $\{\wt{\mu_1}, \dots , \wt{\mu_r} \}$  such that $r \leq k$, all of the candidate means are $0.4(\log(k/w^*))^{1/2 + c}$-separated and for all $i \in [k]$ such that $\mu_i$ is contained in the checker $(V,p, 18 \theta)$, there is some $f(i) \in [r]$ such that 
\[
\norm{\Proj_{V^{\perp}}(\wt{\mu_{f(i)}}) - \Proj_{V^{\perp}}(\mu_i)} \leq 0.1 \,.
\]
Note the last condition is because any component whose mean $\mu_i$ is contained in the checker $(V,p, 18 \theta)$ must be almost entirely contained in the checker $(V,p, 19 \theta)$ and thus will have significant mixing weight in $\red_{V,p,19\theta}^{(\theta)}(\mcl{M})$. 
\\\\
Now we draw $\poly(k/w^*)$ fresh samples from $\mcl{M}$ and restrict to those that are contained in the checker $(V,p, 17\theta)$.  With high probability, this is equivalent to drawing samples from the mixture
\[
\red_{V,p,17\theta}^{(\theta)}(\mcl{M}) \,.
\]
All of the relevant components in this mixture have mean $\mu_i$ that is contained in the checker $(V,p, 18 \theta)$.  Thus, we can apply Claim \ref{claim:cluster-using-means} to fully cluster these samples with high probability.  Recall that by assumption, there is some true mean $\mu_i$ contained in the checker  $(V,p,10 \theta)$.  Thus, we can find one of the clusters with weight at least $0.5w^*$ and such that at least half of the points in the cluster are contained in the checker $(V,p,11\theta)$.  Say that this cluster corresponds to $\mu_j$ and estimated mean $\wt{\mu_{f(j)}}$.  This cluster must be almost entirely contained in the checker $(V,p,17\theta)$.  To isolate exactly the points from $N(\mu_j, I)$ with high probability, we can first restrict to points contained in the checker $(V,p,17\theta)$ and then apply Claim \ref{claim:cluster-using-means} and restrict to points for which the output is $f(j)$.  Thus, we have an efficiently computable test that, with high probability, isolates exactly the points from the cluster $N(\mu_j, I)$ and we are done.
\end{proof}

Now can complete the proof of our main theorem, Theorem \ref{thm:main-GMM}, for learning clusterable mixtures of Gaussians.

\begin{proof}[Proof of Theorem \ref{thm:main-GMM}]
First we apply the reductions in  Section \ref{sec:reductions} to ensure that $d \leq k$ and $\norm{\mu_i - \mu_j} \leq O((k/w_{\min})^2)$ for all $i,j$.  Now we apply the algorithm in Lemma \ref{lem:ball-recursion} for $(\log(k/w_{\min}))^{1 + 0.1c}$ times (where we set $w^* = w_{\min}$).  Note that the lemma works with a trivial initialization where $V$ is $0$-dimensional and $p$ is just $0$.  Using the guarantees of Lemma \ref{lem:ball-recursion}, with high probability at some point, we will get $V,p$ (where $V$ has dimension $(\log(k/w_{\min}))^{1 + 0.1c}$) which satisfy that 
\begin{itemize}
    \item The checker $(V,p, 10(\log (k/w_{\min}))^{(1 + c)/2}) $ contains at least one of the $\mu_i$
    \item For any $\mu_i, \mu_j$ contained in the checker $(V,p, (\log (k/w_{\min}))^{(1 + 1.1c)/2}) $, we have
    \[
    \norm{\mu_i - \mu_j} \leq 0.1(\log(k/w_{\min}))^4 \,.
    \]
\end{itemize}
Thus, when we run Lemma \ref{lem:test-max-separation} to check this pair $V,p$, we will {\sc Accept} this pair of $V,p$ and move to the final step.  Note that Lemma \ref{lem:test-max-separation} also implies that for any pair $V,p$ that we {\sc Accept} and move to the final step, we have the slightly weaker guarantees that
\begin{itemize}
    \item The checker $(V,p, 10(\log (k/w_{\min}))^{(1 + c)/2}) $ contains at least one of the $\mu_i$
    \item For any $\mu_i, \mu_j$ contained in the checker $(V,p, 20(\log (k/w_{\min}))^{(1 + c)/2}) $, we have
    \[
    \norm{\mu_i - \mu_j} \leq (\log(k/w_{\min}))^4
    \]
\end{itemize}
These weaker guarantees still suffice to apply the algorithm in Lemma \ref{lem:final-step} to isolate one of the components of $\mcl{M}$ (again with $w^* = w_{\min}$).  With high probability, we can take $\poly(k/w_{\min})$ samples from this component and estimate its mean and mixing weight to within $\alpha$ (since $\alpha = (k/w_{\min})^{O(1)}$).  We can also remove all of the samples from this component from the mixture and recurse on a mixture of $k-1$ Gaussians.  Note that we can do this because our test for checking whether a sample belongs to the removed component succeeds with high probability (meaning its failure probability is smaller than any inverse polynomial) and the recursive call only takes polynomially many samples.  Thus overall the algorithm succeeds with high probability and we are done.
\end{proof}

The clustering guarantee in Corollary \ref{coro:cluster-GMM} follows as an immediate consequence of Theorem \ref{thm:main-GMM}.
\begin{proof}[Proof of Corollary \ref{coro:cluster-GMM}]
Let the estimated means computed by Theorem \ref{thm:main-GMM} for $\alpha = (w_{\min}/k)^{10}$ be $\wt{\mu_1}, \dots , \wt{\mu_k}$.  Now for all $j_1, j_2 \in [k]$ with $j_1 \neq j_2$, let
\[
v_{j_1j_2} = \frac{\wt{\mu_{j_1}} - \wt{\mu_{j_2}}}{\norm{\wt{\mu_{j_1}} - \wt{\mu_{j_2}} }} \,.
\]
Now given a sample $z$ from $\mcl{M}$, we compute the index $j$ such that for all $j_1, j_2$, we have
\[
| v_{j_1j_2} \cdot (\wt{\mu_j} - z) | \leq (\log (k/w_{\min}))^{(1 + c)/2} \,. 
\]
Note that by the guarantees of Theorem \ref{thm:main-GMM}, there is a permutation $\pi$ such that $\norm{\wt{\mu_{\pi(i)}} -\mu_i} \leq \alpha$ for all $i$.  If $z$ is a sample from $N(\mu_i , I)$, then with high probability the unique index $j$ that satisfies the above is exactly $j = \pi(i)$ and thus, we recover the ground truth clustering with high probability. 
\end{proof}

\bibliographystyle{alpha}
\bibliography{bibliography}

\appendix

\section{Basic Reductions} \label{appendix:reductions} 

Here we explain the reductions for Section \ref{sec:reductions}.  When explaining the reductions, we will work with a mixture of Poincare distributions 
\[
\mcl{M} = w_1 \mcl{D}(\mu_1) + \dots + w_k \mcl{D}(\mu_k)
\]
but it will be obvious that these reductions also work for Gaussians.

\subsubsection*{Reducing to  all Means Polynomially Bounded}\label{sec:reduce-max-norm}

By Fact \ref{fact:basic-Poincare}, with $1 - 2^{-10(d + k)/w_{\min}}$ probability, a sample $z \sim \mcl{D}(\mu_i)$ satisfies 
\[
\norm{z - \mu_i} \leq 10^4 \cdot (d + k)/w_{\min} \,.
\]
Now, we look for a pair of samples $z,z'$ such that 
\[
\norm{z - z'} \geq 10^6 ((d + k)/w_{\min})^{2} \,.
\]
Note that with high probability, such a pair exists if 
\[
\norm{\mu_i - \mu_j} \geq 1.1 \cdot 10^6 ((d + k)/w_{\min})^{2}
\]
for some $i,j$.  Let $v$ be the unit vector in the direction $z - z'$ and let $\mu_i, \mu_j$ be the components that $z,z'$ were drawn from.  We must have that 
\[
| \langle v , \mu_i - \mu_j \rangle | \geq 0.99 \cdot 10^6 ((d + k)/w_{\min})^{2} \,.
\]
Now imagine projecting all of the samples onto the direction $v$.  Note that by Fact \ref{fact:basic-Poincare}, with high probability all of the samples from a given component will lie in an interval of width at most $10^2(d + k)/w_{\min}$ after projecting onto the direction $v$.  Thus, there must be an empty interval of width at least $10^3 (d + k)/w_{\min}$ between the projections of $z$ and $z'$.  This means that no $\mu_i$ has projection in this interval and we can subdivide the mixture into two submixtures with strictly fewer components by cutting via a hyperplane normal to $v$ through the middle of this interval.  Note that for each of the two submixtures, a sample will be on the wrong side of this cut with exponentially small probability but our algorithm will only use polynomially many samples.  Thus, we can simply run our learning algorithm on each submixture.  

We have reduced to $\norm{\mu_i - \mu_j} \leq O(((d + k)/w_{\min})^2)$.  Now we can simply estimate the mean of the distribution $\mcl{M}$ and subtract it out.  Note that we have proved the reduction with a bound of $ O(((d + k)/w_{\min})^2)$.  To reduce to $O((k/w_{\min})^2)$, we can apply the reduction in the next section, which will ensure that $d \leq k$ and then we can apply this reduction again.

\subsubsection*{Reducing the Dimension}\label{sec:reduce-dim}

Next, we show that we can reduce to the case when $d  \leq k$.  We can estimate the empricial covariance of $\mcl{D}$, say $\Sigma_{\mcl{D}}$.  We can also estimate the empirical covariance of $\mcl{M}$, which is
\[
\Sigma_{\mcl{M}} = w_1 (\mu_1 \otimes \mu_1) + \dots + w_k (\mu_k \otimes \mu_k) + \Sigma_{\mcl{D}} \,.
\]
Note that since all means $\norm{\mu_i }$ are polynomially bounded by the previous reduction, we can obtain estimates for $\Sigma_{\mcl{M}}$ and $\Sigma_{\mcl{D}}$ that are accurate to within $\eps$ in Frobenius norm for any inverse polynomial $\eps$.  Thus, we have an estimate $\wt{M}$ such that 
\[
\norm{\wt{M} - \left( w_1( \mu_1 \otimes \mu_1) + \dots + w_k (\mu_k \otimes \mu_k) \right)}_{F} \leq 2\eps
\]
for any desired inverse polynomial $\eps$.  We can then take $V$ to be the subspace spanned by the top $k$ principal components of $\wt{M}$.  Using the above, we can ensure that all of the $\mu_i$ are within distance $0.1 \delta \cdot (w_{\min}/k)^{10}$ of the subspace $V$.  Let $\mcl{D}_{\Proj, V}$ be the projection of the distribution $\mcl{D}$ onto the subspace $V$.  If we project all of our samples onto the subspace $V$, we would have a mixture of translated copies of $\mcl{D}_{\Proj, V}$ where the separations are decreased by at most $0.2 \delta \cdot (w_{\min}/k)^{10}$.  Thus, we have reduced to a $k$-dimensional problem.

\end{document}